\documentclass[11pt,twoside]{article}

\usepackage{subcaption}
\usepackage{rotating}
\usepackage{old-arrows}
\usepackage{enumitem}
\usepackage{mathtools}
\usepackage{amsmath}
\usepackage[utf8]{inputenc} 
\usepackage[T1]{fontenc}    
\usepackage{booktabs}       
\usepackage{amsfonts}       
\usepackage{nicefrac}       
\usepackage{microtype}      

\usepackage{epsf}
\usepackage{epsfig}
\usepackage{fancyhdr}
\usepackage{graphics}
\usepackage{graphicx}
\usepackage{psfrag}
\usepackage{fullpage}
\usepackage{pdfpages}

\newcommand*{\backrefalt}[4]{%
    \ifcase #1 \footnotesize{(Not cited.)}%
    \or        \footnotesize{(Cited on page~#2.)}%
    \else      \footnotesize{(Cited on pages~#2.)}%
    \fi}
\usepackage{amsthm}
\usepackage{amsmath}
\usepackage{amssymb,bbm}
\usepackage{caption}

\usepackage{xcolor}
\usepackage{fancybox}
\usepackage{algorithm}
\usepackage{algpseudocode}
\usepackage{multicol}
\usepackage{caption}

\usepackage{microtype}
\usepackage{graphicx}
\usepackage{booktabs} 
\usepackage{url}
\usepackage[colorlinks,linkcolor=magenta,citecolor=blue]{hyperref}
\usepackage{changes}

\newtheorem{assumption}{Assumption}

\newtheorem{lemma}{Lemma}
\newtheorem{theorem}{Theorem}
\newtheorem{proposition}{Proposition}

\theoremstyle{definition}
\newtheorem{definition}{Definition}

\newtheorem{remark}{Remark}
\newtheorem*{condition}{Condition}

\newtheorem{innercustomgeneric}{\customgenericname}
\providecommand{\customgenericname}{}
\newcommand{\newcustomtheorem}[2]{%
  \newenvironment{#1}[1]
  {%
   \renewcommand\customgenericname{#2}%
   \renewcommand\theinnercustomgeneric{##1}%
   \innercustomgeneric
  }
  {\endinnercustomgeneric}
}

\DeclareMathOperator*{\argmax}{arg\,max}
\DeclareMathOperator*{\argmin}{arg\,min}

\newcustomtheorem{customthm}{Theorem}
\newcustomtheorem{customlemma}{Lemma}
\newcustomtheorem{customprop}{Proposition}

\long\def\comment#1{}

\newcommand{\Ecal}{\ensuremath{\mathcal{E}}}
\newcommand{\Tcal}{\ensuremath{\mathcal{T}}}
\newcommand{\Ocal}{\ensuremath{\mathcal{O}}}
\newcommand{\Ebb}{\ensuremath{\mathbb{E}}}
\newcommand{\Gcal}{\ensuremath{\mathcal{G}}}

\newcommand{\Xcal}{\ensuremath{\mathcal{X}}}

\newcommand{\Pcal}{\ensuremath{\mathcal{P}}}

\newcommand{\Rbb}{\ensuremath{\mathbb{R}}}
\newcommand{\Nbb}{\ensuremath{\mathbb{N}}}
\newcommand{\Pbb}{\ensuremath{\mathbb{P}}}

\newcommand{\divclus}{\mathsf{d}}

\newcommand{\divergence}{\mathcal{D}}

\newcommand{\norm}[1]{\left\|#1\right\|}

\allowdisplaybreaks

\begin{document}

\begin{center}


{\bf{\LARGE{Dendrogram of mixing measures: Hierarchical clustering
\\ \vspace{.05in} 
and model selection for finite
mixture models}}}

\vspace*{.2in}
{\large{
\begin{tabular}{cc}
Dat Do$^{\diamond}$, Linh Do$^{\ddagger}$, Scott McKinley$^{\ddagger}$\\
Jonathan Terhorst$^{\diamond}$, XuanLong Nguyen$^{\diamond}$
\end{tabular}
}}
\vspace*{.1in}

\begin{tabular}{c}
University of Michigan, Ann Arbor$^{\diamond}$; Tulane University$^{\ddagger}$
\end{tabular}

\vspace{.1in}

\today

\vspace*{.2in}

\begin{abstract} 
We present a new way to summarize and select mixture models via the hierarchical clustering tree (dendrogram) constructed from an overfitted latent mixing measure. Our proposed method bridges agglomerative hierarchical clustering and mixture modelling. The dendrogram's construction is derived from the theory of convergence of the mixing measures, and as a result, we can both consistently select the true number of mixing components and obtain the pointwise optimal convergence rate for parameter estimation from the tree, even when the model parameters are only weakly identifiable. In theory, it explicates the choice of the optimal number of clusters in hierarchical clustering. In practice, the dendrogram reveals more information on the hierarchy of subpopulations compared to traditional ways of summarizing mixture models. 
Several simulation studies are carried out to support our theory. We also illustrate the methodology with an application to single-cell RNA sequence analysis.
\end{abstract}
\end{center}


\section{Introduction}
 
In modern data analysis, it is often useful to reduce the complexity of a large dataset by clustering the observations into a small and interpretable collection of subpopulations. Broadly speaking, there are two major approaches.
In ``model-based'' clustering, the data are assumed to be generated by a (usually small) collection of simple probability distributions such as normal distributions, and clusters are inferred by fitting a probabilistic mixture model. Because of their transparent probabilistic assumptions, the statistical properties of mixture models are well-understood. In particular, if there is no model misspecification, i.e., the data truly come from a mixture distribution, then the subpopulations can be consistently estimated. Unfortunately, this appealing asymptotic guarantee is somewhat at odds with what is often observed in practice, whereby mixture models fitted to complex datasets often return an uninterpretably large number of components, many of which are quite similar to each other. 

The tendency of mixture models to overfit on real data leads many analysts to employ ``model-free'' clustering methods instead. A well-known example is hierarchical clustering,  which organizes the data into a nested sequence of partitions at different resolutions. It is particularly useful for data exploration as it does not require fixing a number of subpopulations \emph{a priori} and can be visualized using a dendrogram. Of course, one drawback of model-free clustering is that it does not categorize the data into subpopulations, which is desirable in many scientific applications. 
For example, in single-cell RNA analysis, it is scientifically meaningful to have an estimate of the number of different cell types. 
Since there is no model, developing statistical theory for algorithms like agglomerative hierarchical clustering is challenging. Consequently, inferences derived from hierarchical clustering are guided by intuition and pragmatism rather than theory. 


In this paper, we present a method that bridges model-based and model-free clustering by constructing a hierarchical clustering tree that is guided by mixture modelling. At a high level, the idea is to fit the data by a mixture model to get its parameters, i.e., a latent mixing measure, and then visualize this measure using a dendrogram. The outcome, which we term ``dendrogram of mixing measures'', combines the best benefits of both approaches. It infers a hierarchical clustering structure of data populations and provides theoretical support for the learned structure. 
  
\paragraph{Modeling assumption.} 
Assume that data are generated by a mixture distribution:
\begin{equation}\label{eq:true-model}
x_1, \dots, x_n \overset{\text{iid}}{\sim} p_{G_0}(x) := \int f(x|\theta) dG_0(\theta) = \sum_{i=1}^{k_0} p_i^0 f(x|\theta_i^0),
\end{equation}
where $f$ is a given density function, and $k_0$ is the number of components (order), which is often not given in reality. The \emph{(latent) mixing measure} $G_0 = \sum_{i=1}^{k_0} p_i^0 \delta_{\theta_i^0}$ encodes all the model's parameters that one wants to estimate, including the mixing weights $(p^0_1, \dots, p^0_{k_0})\in \Delta^{k_0-1}$ and the component-wise parameters $\theta^0_1, \dots, \theta^0_{k_0} \in \Theta \subset \Rbb^d$, where $\Theta$ is a known parameter space. 

Given a dataset, there are two main challenges in fitting mixture models: (i) estimating the number of components and (ii) estimating mixing measure $G_0$. 
Regarding challenge (i), because $k_0$ is not given in practice, practitioners often fit data with mixture models of various orders and report an estimate of the number of components. However, for complex datasets possessing hierarchical structures, the inferred number of components may become too large to allow a meaningful interpretation. As a concrete example, when fitting mixtures of location-scale Gaussian distributions to MNIST (handwriting digits data)\footnote{We import image data from the package sklearn in Python, then project it into 95\% variance PCA subspace and fit with Gaussian mixture models.}, the Bayesian Information Criterion (BIC) score gives the best number of components being around 100. Given the fact that there are ten digits in total, how can we make sense of the approximately 100 clusters and interpret them? Several inferred parameter $\theta_i$'s must be close to each other and form different styles of writing the same digit. This motivates us to find a representation that can link nearby components and produce succinct and interpretable approximations of the inferred mixing measure.  

For challenge (ii), if the mixture model is overfitted with a large number of components $k > k_0$, it is known that the mixing measure $G_0$ still can be estimated consistently given suitable identifiability conditions, albeit at (extremely) slow convergence rates~\cite{nguyen2013convergence, Ho-Nguyen-Ann-16,ho2019singularity,Chen1992}. The mathematical reason is that many redundant components “compete" to approximate a common true component so that their moments are cancelled out;
a phenomenon reflected in the asymptotic analysis that requires a suitably higher-order Taylor expansion around $p_{G_0}$ to derive correct rates of convergence ~\cite{heinrich2018strong,ho2019singularity}. If those redundant components can be combined in an algorithmic way, one may hope to recover a fast (and optimal) convergence rate for the relevant parameters which define true $G_0$.  
  
\paragraph{Dendrogram of mixing measures.} 
The foregoing highlights the two major and intertwined obstacles that arise when fitting a mixture model to highly heterogeneous data: that obtaining the single number $k_0$ may be inadequate for the purpose of interpretation and robust inference when the data structure is complex, and there is a real danger of misspecifying the number of components, which results in slow convergence rates for parameter estimation, even if mixture models are a sufficiently rich modelling device. To overcome these obstacles, the approach developed in this paper is a statistically and computationally efficient procedure that outputs a hierarchical clustering tree structure, a.k.a. dendrogram. We will show, both in theory and practice, that the dendrogram is a more meaningful device for data summary, interpretation and robust inference. Since it will be suitably constructed from the mixing measure $G$ overfitted with a large number of redundant components, the dendrogram comes with strong theoretical support when the mixture modelling assumption holds.
%
%
Recall the representation for mixing measure $G = \sum_{i=1}^{k} p_i \delta_{\theta_i}$. Each leaf of the dendrogram is then an atomic measure $p_i \delta_{\theta_i}$, where the $\theta_i$'s are the mixing measure's atoms, and $p_i$'s the corresponding mass. 
We construct the dendrogram by recursively merging components in a manner similar to an agglomerative hierarchical clustering algorithm, taking into account both the atoms and associated mass. In each step, we choose the nearest pair of atoms to merge and form a new atom. Therefore, the number of components in the mixing measure decreases by one after every step. 

The outcome of the procedure is a binary tree representation of the data population.
In practice, the dendrogram gives a sequence of mixing measures with different numbers
of components, enabling the practitioner to visualize the data population’s underlying
heterogeneity at varying levels of granularity. This is particularly useful when one is uncertain about the ``true'' number of mixture components. In theory, we can establish consistent estimates of the number of components and the model parameters (provided the mixture modelling assumption holds). Moreover, we can simultaneously address model selection and derive the optimal root-$n$ rates of parameter estimation when the number of mixture components is unknown. This is remarkable because it has been shown that standard MLE of weakly identifiable mixture models exhibits very slow rates of parameter estimation \cite{Ho-Nguyen-Ann-16}. It is also interesting to note that the procedure's outcome, i.e., the binary tree presentation of subpopulations, circles us back to the original development of mixture models by Pearson in his analysis of evolution biology \cite{pearson1894contributions}, who wrote that “a family probably breaks up first into two species,
rather than three or more, owing to the pressure at a given time of some particular form of natural
selection.” 
  
\paragraph{Contributions.} In summary, we make the following contributions in this paper. First, we develop a method to construct a dendrogram for a given mixing measure, which satisfies a variational characterization with respect to an optimal transport distance. The proposed method is similar to (but not the same as) the centroid linkage method in the usual hierarchical clustering for data points and is useful for capturing the hierarchy of subpopulations in the data structure. Second, we investigate the asymptotic behaviour of the topology of the dendrogram inferred from the data initially fitted with an overparameterized finite mixture model. In particular, the convergence rates of mixing measures, heights, and likelihoods at each level of the dendrogram are derived. Interestingly, the mixing measures on the dendrogram possess a pointwise optimal convergence rate to the true mixing measure despite being constructed from a slowly converging overfitted (i.e., overparameterized) mixing measure. Third, from the developed theory, we propose a novel consistent model selection method named Dendrogram Information Criterion (DIC). Via simulation studies, we demonstrate that DIC is comparable with other well-known information criteria when the model is well-specified, but it is considerably more robust when the model is misspecified. A reason is that DIC also takes the relative distances between fitted components and the magnitude of weights into account and penalizes if it is too small. The usefulness of dendrograms is then confirmed with an application to single-cell RNA sequence data.

\paragraph{Related work.} 
The statistical foundation of clustering algorithms remains underdeveloped in the literature, with many open questions. In early works \cite{hartigan1977distribution, hartigan1985statistical}, Hartigan compared hierarchical clustering to the high-density clusters method. He pointed out that most hierarchical clustering methods do not consistently find high-density clusters, raising questions about applying hierarchical clustering in practice. Asymptotic classification error of “flat" clustering methods such as the $k$-means algorithms for data generated from mixture models was studied in~\cite{dougherty2004probabilistic}. Nonetheless, generalization bounds in statistical learning theory were considered unsuitable for analyzing clustering methods, see, e.g., ~\cite{von2005towards}. Instead, it suggested focusing on convergence behaviour and stability of clustering, which is compatible in spirit to what we will pursue in this paper.
Despite the lack of statistical guarantees, there are intuitive and interesting ways of finding an optimal number of clusters in hierarchical clustering. A popular technique known as the ``elbow method'' is to plot some desired loss function against the corresponding number of clusters and choose the point that looks like a change point in this graph, where the loss function decreases sharply before and becomes flattened after this point. This intuition was made precise by the ``gap statistics'' in \cite{tibshirani2001estimating}. 

In contrast to the algorithmic literature on clustering, there is an extensive body of work on determining the number of components (order) in mixture models with consistency guarantees. The most popular method might be the Information Criterion \cite{schwarz1978estimating}. Many tests for the order of mixture models were developed, including those using the likelihood-based procedure \cite{liu2003asymptotics} and EM algorithm-based tests \cite{chen2009hypothesis, li2010testing}. Another class of consistent frequentist methods fall under the minimum distance-matching estimators \cite{james2001consistent, ho2020robust, heinrich2018strong, wei2023minimum}. In the Bayesian setting, one can select the number of components by placing a prior on this quantity of interest~\cite{richardson1997bayesian, miller2018mixture} and performing posterior inference. The resulting model, namely Mixture of Finite Mixture (MFM), was shown to produce a consistent estimate for the true number of components~\cite{guha2021posterior, miller2018mixture}. However, because most of these methods focus on estimating the number of components but do not take the mixing measure into account (such as relative distances between components and magnitude of weights), they cannot discover the hierarchy in the data. Moreover, they may be brittle to model misspecification~\cite{guha2021posterior,cai2021finite}.  

Model selection by merging procedures is a relatively recent technique. We highlight the Merge-Truncate-Merge (MTM) procedure \cite{guha2021posterior}, Group-Sort-Fure (GSF) algorithm \cite{manole2021estimating}, and Fusing of Localized Densities (FOLD) \cite{dombowsky2023bayesian}, where the core technique is to overfit then merge down. Our proposed method differs in several significant ways that we will now discuss. The MTM procedure is specialized to do model selection for nonparametric Bayesian mixtures, but it cannot improve the intrinsic slow nonparametric rate of parameter estimation after merging. GSF merges nearby components of overfitted mixing measures using a Lasso-like penalty term and can produce a tree from inferred components. However, it requires fitting the model with several different levels of penalized parameters, resulting in computational inefficiency. FOLD computes the Hellinger distances between components' density to merge, but it also requires tuning a hyperparameter that controls clusters' separation. Recently, \cite{aragam2020identifiability} studied a framework to combine single-linkage hierarchical clustering with overfitted mixing measures to construct non-parametric mixture components. In this paper, we consider the parametric setting and provide the convergence rate of the hierarchical clustering tree. Notably, our procedure is developed in a way that provably improves upon the convergence rate of the overfitted mixing measure.
Moreover, it builds a hierarchical tree of components without re-fitting the model with different tuning parameters. Thus, the method is both statistically and computationally efficient. Finally, it is worth emphasizing that our theory and methods also apply to weakly identifiable families, such as location-scale Gaussians (cf. \cite{Ho-Nguyen-Ann-16}), which are not addressable using the aforementioned methods.

\paragraph{Organization.} 
Section~\ref{sec:prelim} gives a brief review of the convergence behaviour of mixture models, the theoretical underpinning of which provided motivation for our proposed procedure of dendrogram construction. In Section~\ref{sec:dendrogram-strong}, we present the construction and model selection methods based on the dendrogram of mixing measures and the asymptotic properties thereof. With a more refined merging scheme, this strategy is extended to accommodate weakly identifiable families of mixtures such as location-scale Gaussians; see Section~\ref{sec:merge-weak}. Section~\ref{sec:experiments} provides several experiments and applications.

\paragraph{Notation.} 
Let $\mathcal{X}$ be the data space and $\Theta$ be the parameter space. We always assume $\Theta$ is a compact and convex subset of $\Rbb^{d}$. For a natural number $k$, we denote $[k] = \{1, \dots, k\}$,  $\Ecal_{k}(\Theta)$ the space of discrete distributions on $\Theta$ with exactly $k$ atoms, and $\Ocal_{k}(\Theta) = \cup_{k'\leq k} \Ecal_{k'}(\Theta)$ the space of discrete distributions on $\Theta$ with no more than $k$ atoms. (The letters $\Ecal$ and $\Ocal$ stand for \textbf{e}xact-fitted and \textbf{o}ver-fitted, respectively.) We drop $\Theta$ in $\Ecal_k$ and $\Ocal_{k}$ when there is no confusion. For a mixing measure $G = \sum_{i=1}^{k} p_i \delta_{\theta_i}$, we abuse the notation by calling each term $p_i \delta_{\theta_i}$ an “atom". (Hence, the so-called atoms in this paper embody both proportion $p_i$ and parameter $\theta_i$.) For two sequences $(a_n)_{n=1}^{\infty}$ and $(b_n)_{n=1}^{\infty}$, we write $a_n\lesssim b_n$ (or $a_n = O(b_n)$) if $a_n \leq C b_n$ where $C$ is a constant not depending on $n$. We write $a_n \gtrsim b_n$ when $b_n \lesssim a_n$, and $a_n \asymp b_n$ if $a_n \gtrsim b_n$ and $b_n \lesssim a_n$. We write $a_n \ll b_n$ (or $a_n = o(b_n)$) if $a_n / b_n\to 0$ as $n\to \infty$. For two densities $p$ and $q$, we denote by $V(p, q) = \dfrac{1}{2}\displaystyle\int |p(x) - q(x)|dx$ the Total Variation distance, $h^2(p, q) = \dfrac{1}{2} \displaystyle\int \left(\sqrt{p(x)} - \sqrt{q(x)}\right)^2 dx$ the square Hellinger distance, and $KL(p \| q) = \displaystyle \int p(x) \log \dfrac{p(x)}{q(x)} dx$ the Kullback-Leibler divergence between $p$ and $q$.  

Throughout the paper, we employ Wasserstein distances \cite{villani2009optimal} to quantify differences between mixing measures. For two mixing measures $G = \sum_{i=1}^{k} p_i \delta_{\theta_i}$ and $G' = \sum_{j=1}^{k'} p'_j \delta_{\theta'_j}$, the Wasserstein$-r$ distance (for $r\geq 1$) between $G$ and $G'$ is defined as $$W_r(G, G') = \left(\inf_{q \in \Pi(p, p')} \sum_{i, j = 1}^{k, k'} q_{ij} \norm{\theta_i - \theta_j'}^{r} \right)^{1/r},$$
where $\Pi(p, p')$ is the set of all couplings between $p = (p_1, \dots, p_k)$ and $p' = (p'_1, \dots, p'_{k'})$, i.e., $\Pi(p, p') = \{q \in \Rbb_{+}^{k\times k'} : \sum_{i=1}^{k}q_{ij} = p'_j , \sum_{j=1}^{k'}q_{ij} = p_i \forall i\in [k], j\in [k']\}$.

\section{Preliminary}\label{sec:prelim}

The methods proposed in this paper consist of a dendrogram construction procedure and model selection techniques associated with finite mixture models. These methods are motivated and partially derived from theory, particularly the theoretical understanding of the convergence behaviour of mixing measures when an overparameterized mixture model is fitted to data, which we present now.
  
\subsection{Convergence of mixture densities} \label{subsec:density-rate}
Given $n$ data $x_1, \dots, x_n \overset{iid}{\sim} p_{G_0}(x)$, where the density function $p_{G_0}(x) = \sum_{i=1}^{k_0} p_i^0 f(x | \theta_i^0)$, and the true latent mixing measure $G_0 = \sum_{i=1}^{k_0} p_i^0 \delta_{\theta_i^0} \in \Ecal_{k_0}$. Suppose the true number of components $k_0$ is unknown, but an upper bound $k\geq k_0$ is given. It is common practice to overfit the data using the $k$-mixture, i.e., mixture with at most $k$ components, yielding the MLE
$\widehat{G}_n \in \argmax_{G \in \Ocal_k} \sum_{i=1}^{n} \log p_{G}(x_i)$. Despite the overfitting, mixture models enjoy fast convergence as a density estimation device. A standard technique to derive the density convergence rate is to employ empirical process theory~\cite{Vandegeer}. Denote $\mathcal{P}_k = \{p_{G} : G\in \Ocal_{k}\}$ and $H_B(\epsilon, \Pcal_{k}, h)$ its $\epsilon$-bracketing entropy number in Hellinger distance \cite{Vandegeer}. We first provide a sufficient condition for establishing the convergence rate for density estimation, then show that it holds in many popular settings.

\begin{condition}[$\textbf{B}$.]
    There exists a constant $C > 0$ depending on $\Theta, k$, and $f$ such that $H_B(\epsilon, \Pcal_{k}, h) \leq C \log(1/\epsilon)\forall \epsilon > 0$.
\end{condition}



\begin{proposition}\label{prop:density-rate}
    Suppose that condition (\textbf{B}.) holds. 
    Then there exist universal constants $c_1, c_2$ and so that it happens with probability at least $1-c_1 n^{-c_2}$ that
    \setlength{\abovedisplayskip}{0pt}
    \begin{equation*}
        h(p_{\widehat{G}_n}, p_{G_0}) \lesssim \left(\dfrac{\log(n)}{n}\right)^{1/2} \quad \forall n\in \Nbb,
    \end{equation*}
\setlength{\belowdisplayskip}{0pt}
    where the multiplicative constant in this inequality only depends on $\Theta$ and $k$.
\end{proposition}
\begin{proposition}\label{prop:Bk-verify}
    Suppose that $\sup_{\theta\in \Theta}\norm{f(\cdot | \theta)}_{\infty}$ is bounded, $\norm{f(\cdot|\theta) - f(\cdot|\theta')}_{\infty} \lesssim \norm{\theta-\theta'}$ for all $\theta, \theta'\in \Theta$, and $f(x|\theta)$ has uniformly light tails, i.e., there exist constants $D$, $d_1, d_2$, and $d_2$ so that
    $f(x | \theta) \leq d_1 \exp(-d_2 \norm{x}^{d_3}) \forall \norm{x}\geq D, \theta\in \Theta,$
    then condition (\textbf{B}.) holds. 
\end{proposition}
In particular, one can easily check that popular kernels such as Poisson, Gamma, and Gaussian (with bounded parameter space) satisfy these conditions. For Gaussian kernel, it further requires eigenvalues of the covariance to be bounded below by a positive constant.

\subsection{Convergence of mixing measures}\label{subsec:parameter-rate}
 
A useful metric to quantify the convergence behaviour of the mixing measures arising in mixture models is the Wasserstein distances \cite{nguyen2013convergence}, partly because it avoids the label-switching problem and can be computed between mixing measures with different numbers of atoms. For a sequence of estimates $(G_n)_{n\in \Nbb} \subset \Ocal_{k}$ and a true mixing measure $G_0\in \Ecal_{k_0}$, it is of interest to derive the convergence rate of $W_r(G_n, G_0)$ for a suitable order $r$~\cite{Ho-Nguyen-Ann-16,heinrich2018strong,wei2023minimum}. To deduce this rate from the available density estimation rate, we aim to develop the so-called \emph{inverse bounds} that have the form $V(p_{G}, p_{G_0}) \gtrsim W_r^{r}(G, G_0)$,
as $V(p_{G}, p_{G_0}) \to 0$. Note that the name inverse bounds comes from the fact that we lower bound the distance between density (on the data space) by the distance between parameters. Obtaining inequalities in the other direction (i.e., upper bound $V$ in terms of $W_r$) is relatively simpler (see, e.g.,~\cite{nguyen2013convergence, nguyen2016borrowing}). Proving inverse bounds requires a refinement of identifiability conditions.

\begin{definition}[Strong identifiability] The family of kernels $\{f(x|\theta) : \theta \in \Theta\}$ (in short, $f$) is said to be $r$-th order strongly identifiable if for every distinct $\theta_1, \dots, \theta_k\in \Theta$, if there exists $(a_{i\alpha}) \subset \Rbb$, for $i\in [k]$ and tuple $\alpha = (\alpha_1, \dots, \alpha_d)\in \Nbb^{d}$ having $|\alpha| = \sum_{j=1}^{d} \alpha_j \leq r$, such that $\sum_{i=1}^{k} \sum_{|\alpha|\leq r}  a_{i\alpha} \dfrac{\partial^{|\alpha|}}{\partial \theta^{\alpha}} f(x|\theta_i) = 0$ for almost all $x$, 
then $a_{i\alpha} = 0$ for all $i \in [k]$ and $\alpha$. 
\end{definition}
In plain words, the $r$-th strong identifiability condition requires kernel $f$ and its derivatives (up to the $r$-th order) with respect to distinct parameters to be linearly independent. Given the first-order strong identifiability, one can obtain the inverse bound $V(p_{G}, p_{G_0})\gtrsim W_1(G, G_0)$ for all $G\in \Ecal_{k_0}$. Given the second-order strong identifiability and $k > k_0$, then $V(p_{G}, p_{G_0})\gtrsim  W_2^2(G, G_0)$ for all $G\in \Ocal_{k}$ (see, e.g.,~\cite{Chen1992, Ho-Nguyen-EJS-16}). Combining with Proposition~\ref{prop:density-rate}, we have that the convergence rate for $W_1(G, G_0)$ is $n^{-1/2}$ in the exact-fitted setting and rate for $W_2(G, G_0)$ is $n^{-1/4}$ in the overfitted setting (up to a logarithmic factor).  

The following fact clarifies the relationship between the convergence in Wasserstein distances of mixing measures and the convergence of components (see, e.g.,~\cite{ho2019singularity}): Fix $G_0 = \sum_{j=1}^{k_0} p^0_j \delta_{\theta^0_j}\in \Ecal_{k_0}$, and consider $G = \sum_{j=1}^{k} p_j \delta_{\theta_j}$ ranging in $\Ocal_{k}$ such that $W_r(G, G_0) \to 0$,  we have
\begin{equation}\label{eq:asymp-Wasserstein}
        W_r^r(G, G_0) \asymp \sum_{i=1}^{k_0} \left(\left|\sum_{j\in V_i} p_j - p_{i}^0 \right| + \sum_{j\in V_i} p_j \norm{\theta_j - \theta_i^0}^{r}\right),
    \end{equation}
where $V_i = V_i(G) = \{j\in [k]: \|\theta_j - \theta_i^0\| \leq \|\theta_j - \theta_{i'}^0\| \forall i'\neq i\}$ is the set of all indices $j$ such that $\theta_j$ belongs to the Voronoi cell of $\theta_i^0$ in $\Theta$.  
Hence, when $W_r^r(G, G_0) \lesssim \left({\log n}/{n}\right)^{1/2}$, it implies that there are atoms that converge to a true atom at the rate as slow as $\left({\log n}/{n}\right)^{1/2r}$. Besides, we see that for each overfitted atom, we have $p_j \norm{\theta_j - \theta_i^0}^r \lesssim \left({\log n}/{n}\right)^{1/2}$. Note that the roles of $p_j$ and $\norm{\theta_j - \theta_i^0}^{r}$ are coupled; they entail the two types of behaviour for atoms of overfitted latent mixing measures: 
\begin{enumerate}
    \item[(i)] \textbf{Redundant components:} There might be several $\theta_j$'s that converge to the same $\theta_i^0$. The sum of their proportions $p_j$ tends to $p_i^0$ but the convergence of each $p_j$ is not known;
    \item[(ii)] \textbf{Excess mass:} 
there may also atoms $\theta_j^{n}$ which “wander" anywhere on the parameter space. The probability mass associated with such atoms $p_j\lesssim (\log n / n)^{1/2} \to 0$, i.e., vanishes at a fast rate.
\end{enumerate}
Hence, not only does parameter estimation in overfitted mixtures suffer from a slow convergence rate, but it also makes the inference difficult due to the two behaviours above. The situation is even worse when the strong identifiability condition is violated, such as the popular mixture of location-scale Gaussians~\cite{Ho-Nguyen-Ann-16}. A primary mathematical reason is the phenomenon of “cancellation" -- as multiple mixture components compete to approximate the same true component, they cancel one another's effect, resulting in learning inefficiency. In theory, the cancellation is dealt with by considering higher derivatives of kernel $f$ to establish inverse bounds for the mixing measures~\cite{ho2019singularity,heinrich2018strong}. However, an inverse bound with relatively higher order $r$ reflects a relatively slower convergence of atoms. This observation suggests that merging redundant atoms can mitigate the inefficiency due to cancellation among redundant parameters and thus help recover a good convergence behaviour for the mixing measures. Indeed, we will show next that it is possible to simultaneously perform model selection and improve the convergence from overfitted mixing measures, including the situations of weak identifiability.

\section{Dendrogram for strongly identifiable mixtures}\label{sec:dendrogram-strong}
  
\subsection{Dendrogram of mixing measures}\label{subsec:tree}
 
The discussion in the previous section motivates a procedure to mitigate the cancellation phenomenon among redundant parameters: we can merge them in a controlled manner. Interestingly, such an incremental merging procedure results in a hierarchical tree structure similar to the dendrogram produced by an agglomerative clustering procedure, with the distinction here being that we will obtain a dendrogram from the latent mixing measure.
Given a discrete mixing measure $G = \sum_{i=1}^{k} p_i \delta_{\theta_i} \in \Ecal_{k}$, the output of the procedure is a binary tree that captures a hierarchy of $G$'s atoms obtained in an iterative fashion. Starting from $G = G^{(k)}$, we sequentially merge two atoms to derive mixing measure $G^{(k-1)}$ which has one less number of atoms. As a result, a sequence of $k$ mixing measures $G^{(k)}, G^{(k-1)}, \dots, G^{(1)}$ is obtained, whereby each $G^{(\kappa)}$ has exactly $\kappa$ atoms, for $\kappa\in [k]$. 
Specifically, define the following dissimilarity between two atoms $p \delta_{\theta}$ and $\pi \delta_{\eta}$:
\begin{equation}\label{eq:dissimilarity-cluster}
    \divclus(p \delta_{\theta}, \pi \delta_{\eta}) = \dfrac{1}{p^{-1} + \pi^{-1}} \norm{\theta - \eta}^2.
\end{equation}
From $G$ having $k$ atoms, we choose two atoms minimizing dissimilarity $\divclus$:
\begin{equation}\label{eq:min-pair-cluster}
    \divclus(p_{i} \delta_{\theta_{i}}, p_{j} \delta_{\theta_{j}}) = \min_{k_1 \neq k_2 \in [k]} \divclus(p_{k_1} \delta_{\theta_{k_1}}, p_{k_2} \delta_{\theta_{k_2}}),
\end{equation}
then we merge $p_{i} \delta_{\theta_{i}}$ and $p_{j} \delta_{\theta_{j}}$ together to get a new cluster $p_* \delta_{\theta_*}$, where
\begin{equation}\label{eq:merge-cluster}
    p_{*} = p_{i} + p_{j}, \quad \theta_{*} = \dfrac{p_{i}}{p_{*}} \theta_{i} + \dfrac{p_{j}}{p_{*}} \theta_{j}.
\end{equation}
Finally, we obtain a mixing measure $G^{(k-1)} = p_{*} \delta_{\theta_{*}} + \sum_{r\neq i, j} p_{r} \delta_{\theta_{r}}$ having $(k-1)$ atoms. A description of the whole procedure can be seen in Algorithm~\ref{alg:merge-atom}. The choice of merging atoms and deriving the new atom (equations~\eqref{eq:min-pair-cluster} and~\eqref{eq:merge-cluster}) are in particular faithful to hierarchical clustering and $k$-means algorithms. With the dissimilarity $\divclus$ defined in equation~\eqref{eq:dissimilarity-cluster}, we can prove the following variational characterization of the merging procedure:

\begin{algorithm}[t]
\caption{Merging of atoms}\label{alg:merge-atom}
\begin{algorithmic}[1]
\Require A mixing measure $G^{(k)} = \sum_{i=1}^{k} p_i \delta_{\theta_i}$;
\State Choose $i, j = \argmin_{k_1 \neq k_2 \in [k]} \divclus(p_{k_1} \delta_{\theta_{k_1}}, p_{k_2} \delta_{\theta_{k_2}})$;
\State Compute merged atoms' parameters: $p_{*} = p_{i} + p_{j}$ and $\theta_{*} = \dfrac{p_{i}}{p_{*}} \theta_{i} + \dfrac{p_{j}}{p_{*}} \theta_{j}$.
\State \textbf{return} A mixing measure $G^{(k-1)} = p_{*} \delta_{\theta_{*}} + \sum_{r\neq i, j} p_{r} \delta_{\theta_{r}}$ with $k-1$ atoms.
\end{algorithmic}
\end{algorithm}

\begin{algorithm}[t]
\caption{Dendrogram Inferred Clustering (DIC) algorithm}\label{alg:phylo-tree}
\begin{algorithmic}[1]
\Require A mixing measure $G^{(k)} = \sum_{i=1}^{k} p_i \delta_{\theta_i}$. 
\State Initiate $\Tcal(G) = (V, E, d)$, where the $k$-th level of $V$ contains all atoms of $G^{(k)}$, $E = \varnothing$, and $d = (d^{(k)}, d^{(k-1)}, \dots, d^{(2)})$ is an array of length $k-1$. 
\For{$\kappa$ from $k$ to $2$} 
    \State Apply Algorithm~\ref{alg:merge-atom} to $G^{(\kappa)}\in \Ecal_{\kappa }$ to get $G^{(\kappa-1)} \in \Ecal_{\kappa-1}$;
    \State Add all atoms of $G^{(\kappa-1)}$ as vertices of the $(\kappa-1)$-th level of $V$;
    \State Add two edges to $E$ connecting the two atoms of $G^{(\kappa)}$ merged into that of $G^{(\kappa-1)}$;
    \State Record $d^{(\kappa)} = \divclus(p \delta_{\theta}, \pi \delta_{\nu})$, where $p \delta_{\theta}$ and $\pi \delta_{\nu}$ are two merged atoms.
\EndFor
\State \textbf{return} $\Tcal(G) = (V, E, d)$, and $\{G^{(j)}\}_{j=1}^{k}$.
\end{algorithmic}
\end{algorithm}

\begin{proposition}\label{prop:Wasserstein-variational}
    For $k\geq 2$ and a measure $G^{(k)} \in \Ecal_{k}$, then $G^{(k-1)}$ constructed in Algorithm~\ref{alg:merge-atom} has exactly $k-1$ atoms and is the “Wasserstein projection" of $G^{(k)}$ onto $\Ocal_{k-1}$, i.e., 
    \begin{equation*}
        G^{(k-1)} = \argmin_{G\in \Ocal_{k-1}} W_2^2(G, G^{(k)}).
    \end{equation*}
    Moreover, the dissimilarity between two merged atoms $p_i \delta_{\theta_i}$ and $p_j \delta_{\theta_j}$ is the squared length of the projection, i.e.,  $\divclus(p_i \delta_{\theta_i}, p_j \delta_{\theta_j}) = W_2^2(G^{(k)}, G^{(k-1)})$.
\end{proposition}
Having presented the algorithm to choose and merge a mixing measure with $k$ atoms to $k-1$ atoms, we now describe the dendrogram of $G$ that emerges by repeatedly applying the merging procedure. Starting from $G = G^{(k)}$, at every step $\kappa$ from $k$ to $2$ (backward), we apply Algorithm~\ref{alg:merge-atom} to $G^{(\kappa)}$ to get $G^{(\kappa-1)}$ having $\kappa-1$ atoms.  

\begin{definition}[Dendrogram]\label{def:dendrogram-mixing-measure}
    The dendrogram of a mixing measure $G$ is a tuple $\Tcal(G) = (V, E, d)$, where set $V$ contains $k$ levels in which the $\kappa$-th level contains $\kappa$ atoms of $G^{(\kappa)}$ as vertices, set $E$ contains the edges specifying the merged vertices, and $d = (d^{(k)}, \dots, d^{(2)})$ where $d^{(\kappa)}$ is the minimal of atoms' dissimilarity $\divclus$ over all pairs of atoms of $G^{(\kappa)}$.
\end{definition} 
When we represent $\Tcal(G)$ on a graph (specifically a hierarchical tree), $d^{(\kappa)}$ is the height between $\kappa$-th level and $(\kappa-1)$-th level. The procedure to construct the dendrogram of $G$ is given by the Dendrogram Inferred Clustering (DIC) algorithm (Algorithm~\ref{alg:phylo-tree}). Figure~\ref{fig:tree-demonstration} shows an example of the dendrogram constructed from an overfitted mixing measure with seven atoms learned from $n$ i.i.d. data generated from a true mixture model with three components. Because of the overfitting, there may exist many redundant atoms in this mixing measure that estimate the same one among the true three atoms. Those atoms are merged along the levels of the dendrogram --- we will show later that this procedure turns out to possess an improved parameter estimation behaviour compared to the original mixing measure.
Moreover, as the overfitted mixing measure is a consistent estimate of the true mixing measure, even if at a slow rate, we expect that the dendrogram also has a limit in a precise sense. Indeed, it will be shown that the height of the dendrogram will be of the order $O(n^{-1/2})$ at all the overfitted levels, while it is $O(1)$ at the exact-fitted and all under-fitted levels. This asymptotic behaviour will be utilized to devise a cut through the dendrogram, leading to a consistent model selection scheme. 

\begin{figure}
    \centering
    {\includegraphics[width = 0.8\textwidth]{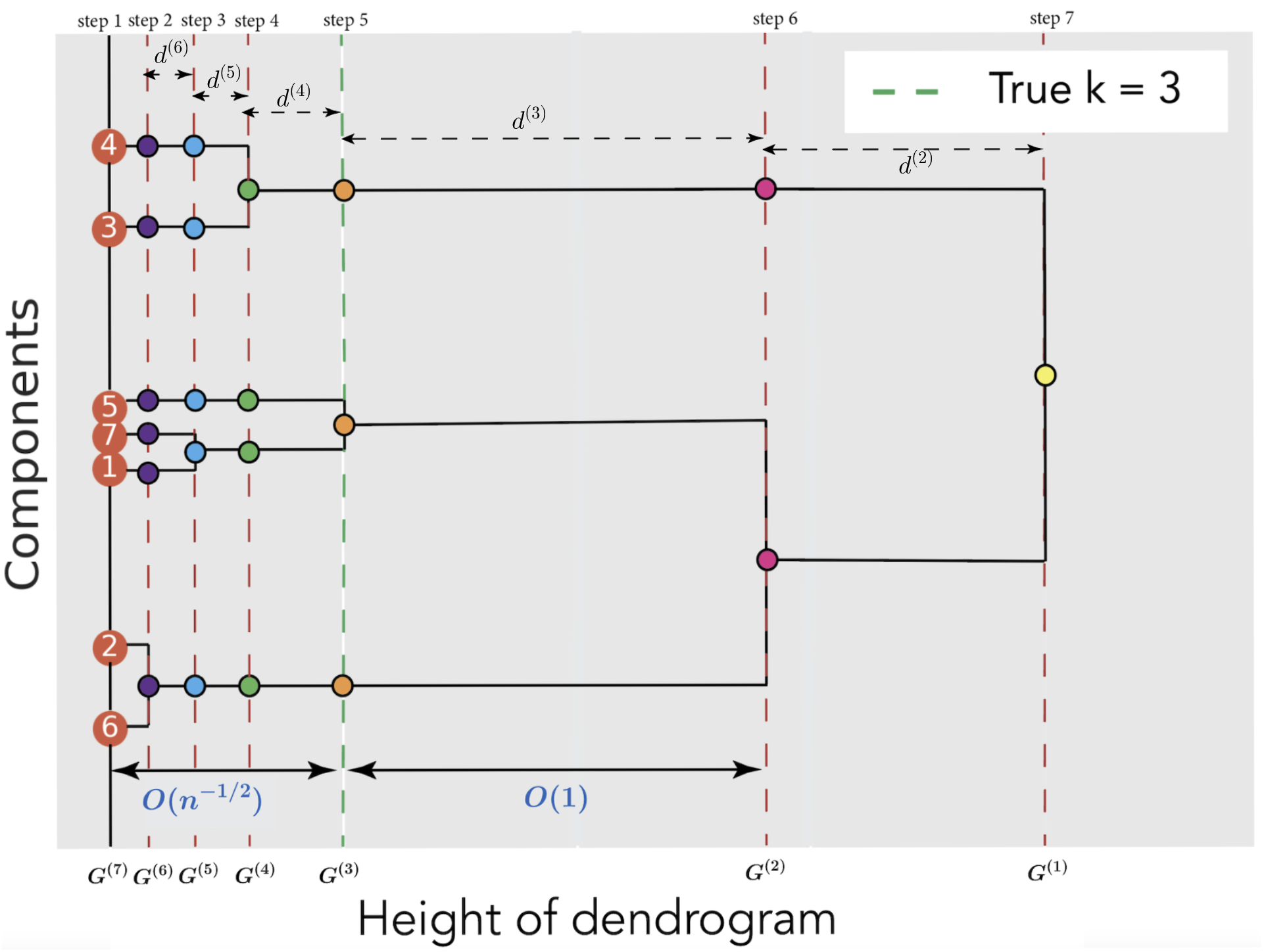}}
    \caption{\centering Dendrogram for the overfitted mixing measure}
    \label{fig:tree-demonstration}
\end{figure}

By merging atoms using dissimilarity $\divclus$ as in equation~\eqref{eq:dissimilarity-cluster}, our dendrogram of mixing measure is most similar to the dendrogram using centroid linkage in Agglomerative Hierarchical Clustering. A key distinction lies in the definition of $\divclus$, which is a product of two terms. The first term represents the harmonic mean of atoms' proportions, which helps to merge the subpopulation with a small proportion into nearby subpopulations. This is meaningful to the mixing measure estimate arising in overfitted mixture models because it can eliminate their excess masses. The second term in $\divclus$ is simply the usual distance metric between two centroids, which can be used to merge redundant atoms. Together, they are useful to post-process the overfitted mixing measures. We would like to highlight that other linkages, such as the single linkage, can also work in a similar manner. However, it requires more effort to deal with excess mass and redundant atoms separately. See Appendix~\ref{sec:single-linkage} for a discussion.

\subsection{Asymptotic properties of the dendrogram}
 
Now, we investigate the asymptotic behaviour of the dendrogram of exact-fitted and overfitted mixing measures in finite mixture models. We are concerned specifically with \textit{convergence rate of mixing measures}, the \textit{height}, and the \textit{likelihood of the model} at each level of the tree. Recall that the core of this theory is the inverse bounds that link distances between densities to those of mixing measures, where we lower bound $V(p_G, p_{G_0})$ by $W_1(G, G_0)$ in the exact-fitted setting, and by $W_2^2(G, G_0)$ in the overfitted setting. Because the cardinality of the support of the mixing measures in the tree varies as we merge, we need a more refined metric that captures both the behaviours of $W_1$ and $W_2^2$. For a mixing measure $G = \sum_{j=1}^{k} p_j \delta_{\theta_j} \in \Ocal_{k}$ and $G_0 \in \Ecal_{k_0}$, we define the following divergence:
\begin{equation}\label{eq:distace-Voronoi}
\divergence(G, G_0) = \sum_{i=1}^{k_0} \left(\left|\sum_{j\in V_i} p_j - p_{i}^0 \right| + \left\|\sum_{j\in V_i} p_j (\theta_j - \theta_i^0) \right\| + \sum_{j\in V_i} p_j \norm{\theta_j - \theta_i^0}^2\right),
\end{equation}
where $V_i = \{j\in [k]: \|\theta_j - \theta_i^0\| \leq \|\theta_j - \theta_{i'}^0\| \forall i'\neq i\}$. Because of the asymptotic relationship~\eqref{eq:asymp-Wasserstein}, we can see that $W_2^2(G, G_0)\lesssim \divergence(G, G_0) \lesssim W_1(G, G_0),$ as any of these goes to 0. Besides, $\divergence(G, G_0) \asymp W_1(G, G_0)$ when $G\in \Ecal_{k_0}$.
Hence, convergence in $\divergence$ is stronger than the typical convergence in $W_2^2$ seen in the literature. At the heart of our analysis is the following inverse bound:
\begin{lemma}\label{lem:inv-bound} Fix $G_0 \in \Ecal_{k_0}$. Suppose the family $f(\cdot |\theta)$ is second-order identifiable. Then, for $k \geq k_0$ and $G\in \Ocal_{k}$, as $V(p_{G}, p_{G_0}) \to 0$, we have $V(p_G, p_{G_0}) \gtrsim \divergence(G, G_0),$
where the multiplicative constant in the inequality depends only on $G_0, \Theta$, and $k$.
\end{lemma}
This inverse bound indicates that if we correctly merge all the redundant atoms in each Voronoi cell together by
$p_i^* = \sum_{j\in V_i} p_j$ and $\theta_i^* = \dfrac{1}{p_i^*}\sum_{j\in V_i} p_j \theta_j$ for all 
$i\in [k_0]$, then the mixing measure $G_* = \sum_{i=1}^{k_0} p_i^* \delta_{\theta_i^*}$ satisfies $V(p_{G_*}, p_{G_0})\gtrsim \divergence(G_*, G_0) \asymp W_1(G_*, G_0)$, which leads to the fast convergence rate for $G_*$. But this is not possible because $(V_i)_{i=1}^{k_0}$ and $k_0$ are actually unknown. However, we will show in Theorem \ref{thm:asymptotic-dendrogram} that our sequential merging scheme (Algorithm~\ref{alg:merge-atom}) achieves exactly this behaviour in an asymptotic sense.
For a mixing measure $G \in \Ecal_{k}$, denote $G^{(k)}, \dots, G^{(1)}$ by the latent mixing measures induced from the dendrogram (Algorithm~\ref{alg:phylo-tree}) with $G^{(\kappa)}$ having $\kappa$ atoms, for $\kappa = 1, \dots, k$. We first establish a desirable property:
\begin{lemma}\label{lem:order-mix-measure}
    As $\divergence(G^{(k)}, G_0) \to 0$, we have
        $\divergence(G^{(k)}, G_0) \gtrsim \divergence(G^{(k-1)}, G_0) \gtrsim \dots \gtrsim \divergence(G^{(k_0)}, G_0)$,
    where the multiplicative constants depend only on $G_0, \Theta$, and $k$.
\end{lemma}
  
\paragraph{Fast convergence of mixing measures arising in the dendrogram.} Now let $\widehat{G}_n\in \Ecal_k$ be the MLE of $G_0\in \Ecal_{k_0}$ based on $n$ i.i.d. samples from $p_{G_0}$, where $k \geq k_0$. By combining the results above, we have the asymptotic behaviour of all latent mixing measures in the dendrogram of $\widehat{G}_n$. Denote by $\widehat{G}_n = \widehat{G}_n^{(k)}, \widehat{G}_n^{(k-1)}, \dots, \widehat{G}_n^{(1)}$ the mixing measures in the dendrogram of $\widehat{G}_n$ and $G_0 = G_0^{(k_0)}, G_0^{(k_0-1)}, \dots, G_0^{(1)}$ on the dendrogram of true $G_0$. 

\begin{theorem}\label{thm:asymptotic-dendrogram}
    Suppose that $f(x|\theta)$ is second-order strongly identifiable and satisfies condition (\textbf{B}.) Then, there exist universal constants $c_1, c_2>0$ such that with probability at least $1 - c_1 n^{-c_2}$, we have $
   \divergence(\widehat{G}_n^{(\kappa)}, G_0) \lesssim \left(\dfrac{\log n}{n}\right)^{1/2}$ for all $\kappa\in [k_0, k]$, where constant in this inequality depends on $G_0, \Theta$ and $k$. In particular, with the same probability, we have
    \begin{equation}
        W_2(\widehat{G}_n^{(\kappa)}, G_0) \lesssim \left(\dfrac{\log n}{n}\right)^{1/4}, \quad \text{and}\quad W_1(\widehat{G}_n^{(\kappa')}, G_0^{(\kappa')}) \lesssim \left(\dfrac{\log n}{n}\right)^{1/2}
    \end{equation}
    for all $\kappa\in [k_0 + 1, k]$ and $\kappa'\in [k_0]$.
\end{theorem}
This theorem establishes that the latent mixing measure obtained from the overfitted $\widehat{G}_n$ at the level $\kappa\leq k_0$ will have the root-$n$ convergence rate to the corresponding measure obtained from $G_0$, even though the initial mixing measure $\widehat{G}_n$ is overfitted and has a slower $n^{-1/4}$ convergence rate, and there is no need to re-fit the model with varying numbers of components from data. Later in Section \ref{sec:merge-weak} we will see even more substantial efficiency gain for weakly identifiable overfitted mixture models.

\paragraph{Heights of the dendrogram.} We now study the asymptotic property of the heights $(d^{(\kappa)})_{\kappa=1}^{k}$. We will show that the heights of the dendrogram of the estimated mixing measure converge to those of the true estimator at the root-$n$ rate. Denote the sequence of heights by \begin{equation}\label{eq:dn-kappa}
d_n^{(\kappa)} = \min \dfrac{1}{(\widehat{p}_i)^{-1} + (\widehat{p}_j)^{-1}} \norm{\widehat{\theta}_i - \widehat{\theta}_j}^2 > 0,
\end{equation}
where the minimum is taken over all pairs of atoms $(\widehat{p}_i \delta_{\widehat{\theta}_i}, \widehat{p}_j \delta_{\widehat{\theta}_j})$ of $\widehat{G}_n^{(\kappa)}$, for $\kappa \in [k]$. The corresponding heights of the dendrogram of true mixing measure $G_0$ are denoted by
$d_0^{(\kappa)} =  \min \dfrac{1}{p_i^{-1} + p_j^{-1}} \norm{\theta_i - \theta_j}^2 > 0,$
where the minimum is taken over all pairs of atoms $(p_i \delta_{\theta_i}, p_j \delta_{\theta_j})$ of $G_0^{(\kappa)}$, for $\kappa \in [k_0]$.
\begin{theorem}\label{thm:asymptotic-height}
    With the same condition and probability as in Theorem~\ref{thm:asymptotic-dendrogram}, we have
    \begin{equation}
        d_n^{(\kappa)} \lesssim \left(\dfrac{\log n}{n}\right)^{1/2}, \quad \text{and}\quad
        \left|d_n^{(\kappa')} - d_0^{(\kappa')} \right|
        \lesssim \left(\dfrac{\log n}{n}\right)^{1/2},
    \end{equation}
    for all $\kappa\in [k_0 + 1, k]$, $\kappa'\in [k_0]$. The multiplicative constants depend on $G_0, \Theta$ and $k$.
\end{theorem}
Hence, as $n\to \infty$, $d_n^{(\kappa)} = O((\log n / n)^{1/2})$ for all $\kappa > k_0$ but $d_n^{(\kappa)}$ tends to $d_0^{(\kappa)} > 0$ for all $\kappa \leq k_0$, with a high probability. An illustration is given by Figure~\ref{fig:tree-demonstration}. We will exploit this asymptotic behaviour as a model selection criterion and will discuss this further in Section~\ref{sec:model-select-strong}. Note that by Proposition~\ref{prop:Wasserstein-variational}, $d_0^{(\kappa)}$ is the squared length of the projection of $G_0^{(\kappa)}$ onto $\Ocal_{\kappa-1}$. So $d_n^{(k_0)}, \dots, d_n^{(k_*)}$ gets smaller as $G_0$ is closer to the subspace of mixing measures with at most $k_* - 1$ atoms, for some $k_* > k_0$, which leads to more difficulties in estimating the true number of components $k_0$. This partially explains the slow minimax rate and convergence rate when we allow the atoms to arbitrarily overlap~\cite{heinrich2018strong,wei2023minimum}.
  
\paragraph{Likelihood of mixing measures in the dendrogram.}
To develop a suitable model selection technique, it is essential to study the behaviour of the likelihood function. Relevant notions include the entropy and relative entropy (a.k.a. the KL divergence), which arise as the expected log-likelihood. For a density $p_{G}$, let its entropy be denoted by $H(p_{G}) = - \Ebb_{X\sim p_{G}} \log p_{G}(X)$,
and the average log-likelihood by \begin{equation}\label{eq:ln-kappa}
\overline{\ell}_n(G) = \dfrac{1}{n}\sum_{i=1}^{n} \log p_{G}(x_i).
\end{equation}
The following theorem establishes the convergence behaviour of the likelihood function $\overline{\ell}_n$. We need a mild technical condition on the relative moments of the model's probability density ratios.
\begin{condition}[\textbf{K}.]
There exists $\delta, \epsilon_0, M > 0$ such that $\displaystyle\int \left(\dfrac{p_{G_0}(x)}{p_{G}(x)}\right)^{\delta} p_{G_0}(x) dx \leq M$ for any $G \in \Ocal_k(\Theta)$ satisfying $W_2(G, G_0)\leq \epsilon_0$.
\end{condition}

\begin{theorem}\label{thm:asymptotic-likelihood}
Assume the conditions in Theorem~\ref{thm:asymptotic-dendrogram} and condition (\textbf{K}.). Suppose further that $h(f(\cdot|\theta), f(\cdot, \theta'))\lesssim \norm{\theta-\theta'} \forall \theta, \theta'\in \Theta$. Then, with the same probability as in Theorem~\ref{thm:asymptotic-dendrogram}, 
\begin{equation*}
    \left|\overline{\ell}_n(\widehat{G}_n^{(\kappa)}) + H(p_{G_0})\right| \lesssim \left(\dfrac{\log n}{n}\right)^{1/4},\quad \forall \kappa\in [k_0, k].
\end{equation*}
Assume additionally that there exists a measurable function $m: \Xcal\to \Rbb$ such that $\sup_{G\in \Ocal_{\kappa}} |\log p_{G}(x)| \leq m(x)$ for all $x\in \Xcal$ and $\kappa < k_0$. Then we have
\begin{equation*}
    \overline{\ell}_n(\widehat{G}_n^{(\kappa)}) + H(p_{G_0}) \to - KL(p_{G_0} \| p_{G_{0}^{(\kappa)}}) < 0, \quad \forall \kappa\in [k_0 - 1],
\end{equation*}
in $\Pbb_{p_{G_0}}$-probability as $n\to \infty$.
\end{theorem}
This theorem establishes that the average log-likelihood of the model on the dendrogram has the same limit for all levels $\kappa \geq k_0$. Meanwhile, there will be a gap between the exact-fitted ($k_0$) and all the under-fitted levels $\kappa < k_0$. The combined information from the likelihood and dendrogram will prove useful for designing model selection procedures. Note that the uniform bounded condition for the log-likelihood of the model at the under-fitted levels is familiar in the literature (see, e.g., \cite{keener2010theoretical} Chapter 9). It is satisfied for common exponential families with compact parameter spaces, such as Gaussian, Student, Binomial, and Negative Binomial distributions.

\subsection{Model selection via the dendrogram}\label{sec:model-select-strong}
 
We are ready to address the questions of model selection associated with our method for dendrogram construction. In this subsection, we always assume that conditions (\textbf{B}.), condition (\textbf{K}.), strong identifiability, and condition of Theorem~\ref{thm:asymptotic-likelihood} are satisfied.

\paragraph{“Where to cut the dendrogram?"} is a common question in Agglomerative Hierarchical Clustering. Cutting the dendrogram at a height will specify the number of clusters one wants to obtain and interpret. Here, we provide an answer to this question from an asymptotic viewpoint. Choose any sequence $\epsilon_n$ such that $(\log n / n)^{1/2} \ll \epsilon_n \ll 1$. Let $k_n$ be the number of components when we cut the dendrogram at the height $\epsilon_n$, i.e., 
$k_n = \min\left\{\kappa\in [2, k]: \sum_{i = \kappa}^{k} d_n^{(i)} \leq \epsilon_n\right\} - 1.$ The consistency of estimator $k_n$ then follows.
\begin{proposition}\label{prop:cut-dendrogram}
    $k_n \to k_0$ in $\Pbb_{p_{G_0}}$-probability, as $n \to \infty$.
\end{proposition}

\paragraph{DIC (Dendrogram Information Criterion).} For every $\kappa\in \{2, \dots, k\}$, denote by $d_n^{(\kappa)}$ the length of the $\kappa$-th level of the dendrogram and $\overline{\ell}_n^{(\kappa)}$ the average log-likelihood of the mixture model with mixing measure $G_n^{(\kappa)}$ as in equations~\eqref{eq:dn-kappa} and~\eqref{eq:ln-kappa}. Consider $\text{DIC}_n^{(\kappa)} = -(d_n^{(\kappa)} + \omega_n \overline{\ell}_n^{(\kappa)})$,
where $\omega_n$ is any slowly increasing sequence so that $1 \ll \omega_n \ll (n/ \log n)^{1/4}$. For example, a valid choice is $\omega_n = \log(n)$. We choose $k_n^* = \argmin_{\kappa \in [2, k]} \text{DIC}_n^{(\kappa)}$.
\begin{proposition}\label{prop:DIC-consistency}
Suppose that $k_0\geq 2$, then $k_n^* \to k_0$ in $\Pbb_{p_{G_0}}$-probability as $n\to \infty$.
\end{proposition}

\begin{remark}
    While Akaike Information Criteria (AIC) and Bayesian Information Criteria (BIC) only involve the likelihood of the model and the number of components, DIC further takes the information of mixing measures into account. Recall that the stepwise penalty term $d_{n}^{(\kappa)}$ is the squared length of the projection of $\widehat{G}_n^{(\kappa)}$ onto $\Ocal_{\kappa-1}$, and we partially seek to maximize this number by choosing $k_n^* = \argmin \text{DIC}_n^{(\kappa)}$. Hence, $\text{DIC}_n^{(\kappa)}$ is heavily penalized when $d_{n}^{(\kappa)}$ is small, which is equivalent to the event that there are two nearby atoms in the mixing measure or there is an atom with a small mass. As a result, DIC is more robust to model specification compared to AIC and BIC. We will empirically demonstrate this point in Section~\ref{sec:experiments}.  
\end{remark}

\section{Dendrogram for location-scale Gaussian mixtures}\label{sec:merge-weak}
  
\subsection{Location-scale Gaussian mixtures}

The previous section addresses strongly identifiable mixture models in considerable generality. In this section, we turn to the weakly identifiable setting, i.e., when the second-order strongly identifiable condition is violated. We focus on Gaussian location-scale mixture model, which is the most popular representative of weakly identifiable models known to exhibit very slow parameter estimation behaviour in the overfitted setting \cite{Ho-Nguyen-Ann-16,ho2019singularity}. The location-scale Gaussian kernel is denoted by $f(x | \mu, \Sigma)$. The parameters of interest are $(\mu, \Sigma) \in \Theta \times \Omega$, where $\Theta$ is a compact subset of $\Rbb^{d}$ and $\Omega$ is a compact set containing all symmetric and positive-definite matrices in $\Rbb^{d\times d}$ having bounded eigenvalues (both above and below by positive constants). This kernel is known to violate the strong identifiable condition due to the heat equation:
\begin{equation}\label{eq:heat-eq}
    \dfrac{\partial^2 f(x|\mu, \Sigma)}{\partial \mu\partial \mu^{\top}} = 2\dfrac{\partial f(x|\mu, \Sigma)}{\partial \Sigma}\,\, \forall x, \mu, \Sigma,
\end{equation}
so the asymptotic theory presented earlier does not hold. Given $n$ i.i.d. samples from $p_{G_0} = \sum_{j=1}^{k_0} p_j^0 f(x|\mu_j^0, \Sigma_j^0),$ 
for $G_0 = \sum_{j=1}^{k_0} p_j^0 \delta_{(\mu_j^0, \Sigma_j^0)}$. Let $\Ocal_{k, c}$ denote the space of mixing measures with at most $k$ atoms having masses bounded below by $c > 0$. Due to the singularity structure~\eqref{eq:heat-eq}, \cite{Ho-Nguyen-Ann-16} showed that the convergence rate for the overfitted MLE $\widehat{G}_n$ ranging in $\Ocal_{k, c_0}$ is $\left(\dfrac{\log n}{n}\right)^{1/2 \overline{r}(k - k_0 + 1)}$, for $k > k_0$ and $c_0 \in (0, \min_{j=1}^{k_0} p_j^0)$, where $\overline{r}(k)$ is defined as the smallest integer $r$ such that the system of polynomial equations $\sum_{j=1}^{k} \sum_{\substack{n_1, n_2\in \Nbb\\ n_1 + 2 n_2 = \alpha}} \dfrac{c_j^2 a_j^{n_1} b_j^{n_2}}{n_1 ! n_2 !} = 0$, for $\alpha = 1, \dots, r$,
does not have any nontrivial solution $(a_j, b_j, c_j)_{j=1}^{k} \subset \Rbb$. A set of solutions is considered nontrivial if all variable $c_j$'s are non-zero and at least one of $a_j$'s is non-zero. It was shown that $\overline{r}(2) = 4, \overline{r}(3) = 6$, and $\overline{r}(4) \geq 7$. Therefore, the parameter estimation rate can be as slow as $n^{-1/8}$ when over-fitting by one and $n^{-1/12}$ when over-fitting by two components. 
The slow convergence rate implies that when one overparameterizes the mixture model with redundant components, an excessive amount of data is required to achieve a good recovery error for parameter estimation. An illustration of an over-fitting location-scale Gaussian mixture can be seen in Appendix A.1. 

We now aim to construct a dendrogram of mixing measures in location-scale Gaussian mixtures to achieve a fast estimation rate of order $n^{-1/2}$ for the quantities of interest.
The main insight for establishing the parameter estimation rate of overfitted location-scale Gaussian mixtures comes from the inverse bound $V(p_{G}, p_{G_0})\gtrsim W_{\overline{r}(k - k_0 + 1)}^{\overline{r}(k - k_0 + 1)}(G, G_0)$ for all $G\in \Ocal_{k, c_0}$~\cite{Ho-Nguyen-Ann-16}. 
The essential part of proving this inverse bound is to use the relationship~\eqref{eq:heat-eq} to convert all derivatives (in both $\mu$ and $\Sigma$) of the Taylor expansion of $f(x|\mu, \Sigma) - f(x|\mu^0, \Sigma^0)$ to derivatives with respect to $\mu$ only, then use the strong identifiability of Gaussian mixtures with respect to the location parameter up to order $\overline{r}$. However, similar to the general strong identifiable setting studied in Section~\ref{sec:dendrogram-strong}, here the inverse bound does not take advantage of all linear independence relationships of derivatives of the Gaussian kernel but uses only the zero and $\overline{r}$-th order. By invoking the linear independence more carefully, we can provide a tighter inverse bound, which in turn suggests an effective merging scheme for constructing the dendrogram. For any $G = \sum_{j=1}^{k} p_j \delta_{(\mu_j, \Sigma_j)} \in \Ocal_{k}$, denote:
\setlength{\abovedisplayskip}{0pt}
\begin{align}\label{eq:D-weakly}
    \divergence_{\Gcal}(G, G_0) & = \sum_{i=1}^{k_0} \left(\left|\sum_{j\in V_i} p_j - p_i^0 \right| + \sum_{j\in V_i} p_j\left(\norm{\mu_j - \mu_i^0}^{\overline{r}(|V_i|)} + \norm{\Sigma_j - \Sigma_i^0}^{\overline{r}(|V_i|)/2}\right) \right.\nonumber \\
    & \hspace{-.8cm} \left. + \norm{\sum_{j\in V_i}p_j (\mu_j - \mu_i^0)}  + \norm{\sum_{j\in V_i}p_j \left((\mu_j - \mu_i^0) (\mu_j - \mu_i^0)^{\top} + \Sigma_j - \Sigma_i^0)\right)}\right),
\end{align}
\setlength{\abovedisplayskip}{0pt}
where $V_i \equiv V_i(G) = \{j\in [k]: \|\mu_j - \mu_i^0\|^2 + \|\Sigma_j - \Sigma_i^0\| \leq \|\mu_j - \mu_{i'}^0\|^2 + \|\Sigma_j - \Sigma_{i'}^0\| \forall i'\neq i\}$ contains all indices $j$ such that $(\mu_j, \Sigma_j)$ belongs to the Voronoi cell of $(\mu_i^0, \Sigma_i^0)$ in $\Theta \times \Omega$, for $i \in [k_0]$, the matrices norm is Frobenius, and we take $\overline{r}(1) = 2$ as convention. The following bound plays the central role in our construction of the dendrogram:
\begin{lemma}\label{lem:inv-bound-weak}
    For $G \in \Ocal_{k, c_0}$, as $V(p_{G}, p_{G_0})\to 0$, we have $V(p_G, p_{G_0}) \underset{}{\gtrsim} \divergence_{\Gcal}(G, G_0)$.
\end{lemma}
Letting $\overline{r}(G) = \max_{i\in [k_0]} \overline{r}(|V_i|)$, it holds that $\divergence_{\Gcal}(G, G_0)\gtrsim W_{\overline{r}(G)}^{\overline{r}(G)}(G, G_0) \gtrsim  W_{\overline{r}(k_0 - k+1)}^{\overline{r}((k_0 - k+1)}(G, G_0)$. Hence, this inverse bound is stronger than the existing results in literature~\cite{Ho-Nguyen-Ann-16, manole2022refined}. Moreover, $\divergence_{\Gcal}$ can capture both the slow rate of order $1 / (2\overline{r}(G))$ for the overfitted mixing measure and the fast rate of order $1/2$ for the merged mixing measure at the exact-fitted level.
\begin{algorithm}
\caption{Merging of atoms for the mixture of location-scale Gaussians}\label{alg:merge-atom-weak}
\begin{algorithmic}[1]
\Require A mixing measure $G^{(k)} = \sum_{i=1}^{k} p_i \delta_{(\mu_i, \Sigma_i)}$.
\State Choose $i, j = \argmin_{k_1 \neq k_2 \in [k]} \divclus(p_{k_1} \delta_{(\mu_{k_1}, \Sigma_{k_1})}, p_{k_2} \delta_{(\mu_{k_2}, \Sigma_{k_2})}).$
\State Compute merged atoms' parameters by 
$p_{*} = p_{i} + p_{j}, \mu_{*} = \dfrac{p_{i}}{p_{*}} \mu_{i} + \dfrac{p_{j}}{p_{*}} \mu_{j},$
and
\begin{equation}\label{eq:merge-variance}
    \Sigma_* = \dfrac{p_{i}}{p_{*}} \left(\Sigma_{i} + (\mu_i - \mu_*)(\mu_i - \mu_*)^{\top} \right) + \dfrac{p_{j}}{p_{*}} \left(\Sigma_{j} + (\mu_j - \mu_*)(\mu_j - \mu_*)^{\top} \right).
\end{equation}
\State \textbf{return} Mixing measure $G^{(k-1)} = p_{*} \delta_{(\mu_*, \Sigma_*)} + \sum_{r\neq i, j} p_{r} \delta_{(\mu_r, \Sigma_r)}$ with $k-1$ atoms.
\end{algorithmic}
\end{algorithm}

\paragraph{Dendrogram for the mixture of location-scale Gaussians.} 
The strengthened inverse bound and the form of $\divergence_{\Gcal}$ in equation~\eqref{eq:D-weakly} motivate a merging scheme that preserves the first-order terms with respect to $\mu$ and $\Sigma$. Define the dissimilarity between two atoms to be: $\divclus(p_i \delta_{(\mu_i, \Sigma_i)}, p_j \delta_{(\mu_j, \Sigma_j)}) = \dfrac{1}{p_i^{-1} + p_j^{-1}} \left(\norm{\mu_i - \mu_j}^2 + \norm{\Sigma_i - \Sigma_j}\right).$
The one-step merging scheme for location-scale Gaussian mixtures is presented in Algorithm~\ref{alg:merge-atom-weak}. The whole dendrogram is then constructed sequentially via Algorithm~\ref{alg:phylo-tree}. Note the fundamental difference compared to the strongly identifiable case (Algorithm~\ref{alg:merge-atom}) when it comes to merging the covariance parameter: Here, there is a dependence of the covariance parameter $\Sigma$ in the mean parameter $\mu$ when merging takes place. Intuitively, there is a linear dependence between the first-order term of $\Sigma$ and the second-order term of $\mu$ as manifested by the structural equation~\eqref{eq:heat-eq}. This dependence is “resolved" by our updating of the merging covariance with the “variance of the merged means", as in equation~\eqref{eq:merge-variance}. A more detailed discussion of the choice of $\divclus$ is presented in Section~\ref{subsec:discussion-dendrogram-locationscale}.

  
\subsection{Asymptotic theory and model selection}
 
\paragraph{Asymptotic behaviour of the dendrogram.} We now present theoretical results regarding the convergence rate of mixing measures, heights, and the likelihood of each level in the dendrogram derived from the mixture of location-scale Gaussians. Suppose there are $n$ i.i.d. observation generated from a mixture of Gaussians $p_{G_0}$, with true latent mixing measure $G_0 = \sum_{i=1}^{k_0} p_i^0 \delta_{(\mu_i^0, \Sigma_i^0)} \in \Ecal_{k_0}$, and we obtain the MLE mixing measure $\widehat{G}_n = \sum_{i=1}^{k} p_i^{n} \delta_{(\mu_i^n, \Sigma_i^n)} \in \Ocal_{k, c_0}$, where $k \geq k_0$ and lower bound $c_0$ on the mixing proportion satisfies $c_0 \in (0, \min_{i\in [k_0]} p_i^0)$. We still denote $\widehat{G}_n^{(\kappa)}$ and $d_{n}^{(\kappa)}$ (resp., ${G}_0^{(\kappa)}$ and $d_{0}^{(\kappa)}$) the latent mixing measure and the height of the level $\kappa$ in the dendrogram of $\widehat{G}_n$ (resp., $G_0$).
\begin{theorem}\label{thm:asymp-theory-gaussian}
    There exist universal constants $c_1$ and $c_2$ such that for probability at least $1 - c_1 n^{-c_2}$, we have all the inequalities below for $\kappa\in [k_0 + 1, k]$ and $\kappa'\in [k_0]$, of which the multiplicative constants depend on $G_0$, $\Theta \times \Omega$, and $k$. Firstly, for the convergence of mixing measures, we have
    \begin{equation}\label{eq:weak-conv-mixing-measure}
        W_{\overline{r}(\widehat{G}_n)}(\widehat{G}_{n}^{(\kappa)}, G_0) \lesssim \left(\dfrac{\log n}{n} \right)^{1/2\overline{r}(\widehat{G}_n)},\quad \text{and}\quad
        W_{1}(\widehat{G}_{n}^{(\kappa')}, G_0^{(\kappa')}) \lesssim \left(\dfrac{\log n}{n} \right)^{1/2}.
    \end{equation}
    For the heights of the dendrogram, we have
    \begin{equation}\label{eq:weak-conv-heights}
        d_n^{(\kappa)}  \lesssim \left(\dfrac{\log n}{n}\right)^{1/\overline{r}(\widehat{G}_n)},\quad \text{and}\quad\left|d_n^{(\kappa')} - d_0^{(\kappa')} \right| \lesssim \left(\dfrac{\log n}{n}\right)^{1/2}.
    \end{equation}
    Finally, for the average log-likelihood, $\left|\overline{\ell}_n(\widehat{G}_n^{(\kappa)}) + H(p_{G_0}) \right| \lesssim \left(\dfrac{\log n}{n}\right)^{1/2\overline{r}(\widehat{G}_n)}$,
    and
    \begin{equation*}\label{eq:weak-conv-likelihood}
\overline{\ell}_n(\widehat{G}_n^{(\kappa')}) + H(p_{G_0}) \to - KL(p_{G_0} \| p_{G_{0}^{(\kappa')}}) < 0,
\end{equation*}
in $\Pbb_{p_{G_0}}$-probability as $n\to \infty$.
\end{theorem}

Note the fast rate $(\log n/n)^{1/2}$ of the mixing measures $\widehat{G}_{n}^{(\kappa')}$ associated with every level $\kappa' \leq k_0$. This is meaningful because $k_0$ can be estimated consistently as follows.

\paragraph{DIC for mixtures of location-scale Gaussians.} Similar to the strong identifiability case, we consider the statistics $\text{DIC}_n^{(\kappa)} = -(d_n^{(\kappa)} + \omega_n \overline{\ell}_n^{(\kappa)})$,
where $\omega_n$ is any slowly increasing sequence so that $1 \ll \omega_n \ll (n/ \log n)^{1/2\overline{r}(\widehat{G}_n)}$. Although we do not know $\overline{r}(\widehat{G}_n)$ in general, a sufficiently good choice is $\omega_n = \log(n)$, which increases to infinity but is slower than any polynomial. Choose $k_n^* = \argmin_{\kappa \in [k]} \text{DIC}_n^{(\kappa)}$. We have the following result of consistency.
\begin{proposition}\label{prop:DIC-consistency-gaussian}
    $k_n^* \to k_0$ in $\Pbb_{p_{G_0}}$-probability as $n\to \infty$.
\end{proposition}

\section{Simulation studies and real data illustrations}\label{sec:experiments}
  
\subsection{Synthetic data}
 
\subsubsection{Fast parameter estimation via the dendrogram}\label{subsubsec:rate-dendrogram}
We first illustrate that the merged mixing measure in the dendrogram can achieve the fast 
convergence rate to the true mixing measure, although it is constructed from a slowly converging overfitted mixing measure (estimated from empirical data). Here, we consider both strong and weak identifiability settings. In each setting, the data is generated from a mixture of 3 two-dimensional Gaussian distributions, with uniform mixing proportion $(p^{0}_{1}, p^{0}_{2}, p^{0}_{3}) = (1/3, 1/3, 1/3)$ and true mean parameters being $\theta_1^{0} = (2, 1), \theta_2^{0} = (0, 6)$ and $\theta_3^{0} = (-2, 1)$, respectively. In the former setting, covariance matrices of all three true components are the identity matrix and are known, while $G_0 = \frac{1}{3}\delta_{\theta_1^0} + \frac{1}{3}\delta_{\theta_2^0} + \frac{1}{3}\delta_{\theta_3^0}$ is to be estimated. In the latter setting, the true covariance matrices are $\Sigma_1^0 = ((.5, .5), (.5, .1)), \Sigma_2^0 = ((.5, -.1), (-.1, .1))$ and $\Sigma_3^0 = ((.25, .5), (.5, 2.))$, and $G_0 = \frac{1}{3}\delta_{(\theta_1^0, \Sigma_1^0)} + \frac{1}{3}\delta_{(\theta_2^0, \Sigma_2^0)} + \frac{1}{3}\delta_{(\theta_3^0, \Sigma_3^0)}$ is to be estimated. For each setting, we consider the logarithm of the sample sizes $\log_{10}(n)$ ranging from $2$ to $4$ (so that $n$ ranges from $100$ to $10,000$) and generate $n$ samples from the true distribution. The \textbf{e}xact-fitted MLE $\widehat{G}_n^{e}\in \Ecal_{3}$ and \textbf{o}ver-fitted MLE $\widehat{G}_n^{o}\in \Ocal_{5}$ are then learned from data using the EM algorithm. We use Algorithm~\ref{alg:phylo-tree} to merge the overfitted measure $\widehat{G}_n^{o}$ to get the \textbf{m}erged measure $\widehat{G}_{n}^{m}\in \Ecal_{3}$ in the dendrogram, then measure the error of all estimators $\widehat{G}_{n}^{o}, \widehat{G}_{n}^{e}, \widehat{G}_{n}^{m}$ to the true mixing measure $G_0$ in the Wasserstein distances. Each experiment is repeated 64 times, and we plot the average and quartile bar of the logarithm of the error in Figure~\ref{fig:rates-convergence} (best to see with colour). 

It can be seen that the merged algorithm improves the convergence rate of the overfitted mixing measures from $n^{-1/4}$ (in Wasserstein-2 distance) to $n^{-1/2}$ (in Wasserstein-1 distance) in the strongly identifiable setting and from $n^{-1/12}$ (in Wasserstein-6 distance) to $n^{-1/2}$ (in Wasserstein-1 distance) in the weakly identifiable setting, which is comparable with the exact-fitted mixing measures. Recall that we use Wasserstein-6 distance due to the solution structure of the system of polynomial equations (Section~\ref{sec:merge-weak}) and $\overline{r}(5-3+1) = \overline{r}(3)=6$. The convergence rate of overfitted mixing measures in Wasserstein-1 distance may vary anywhere from $n^{-1/4}$ to $n^{-1/2}$ in the strongly identifiable setting and from $n^{-1/12}$ to $n^{-1/2}$ in the weakly identifiable setting, depending on the relationship between overfitted components; however, the rate is typically not the fast root-$n$ rate. Notably, in the location-scale Gaussian experiments, the merged measures are mostly similar to the exact-fitted mixing measures so that they almost overlap. It can be partially explained by investigating the EM algorithm and merging scheme in Algorithm~\ref{alg:merge-atom-weak}. Hence, the mixing measure in the exact-fitted level of the dendrogram has a root-$n$ convergence rate to the true parameter. It is remarkable that this measure is directly constructed from the slowly convergent overfitted mixing measure without re-fitting the model with data.
\begin{figure}[t!]
      \centering
      \subcaptionbox*{\scriptsize (a) Strongly identifiable setting \par}{\includegraphics[width = 0.48\textwidth]{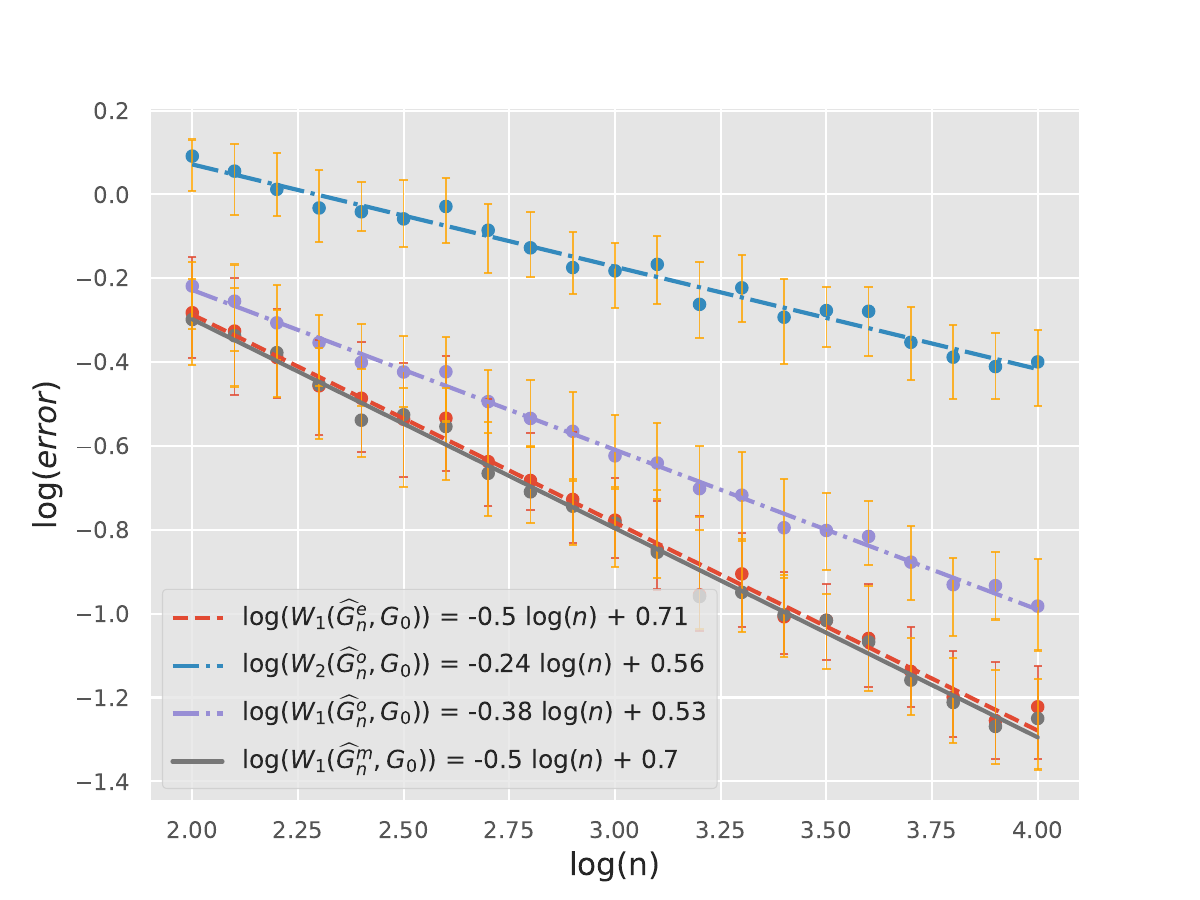}}
      \subcaptionbox*{\scriptsize 
        \centering (b) Weakly identifiable setting \par}{\includegraphics[width = 0.48\textwidth]{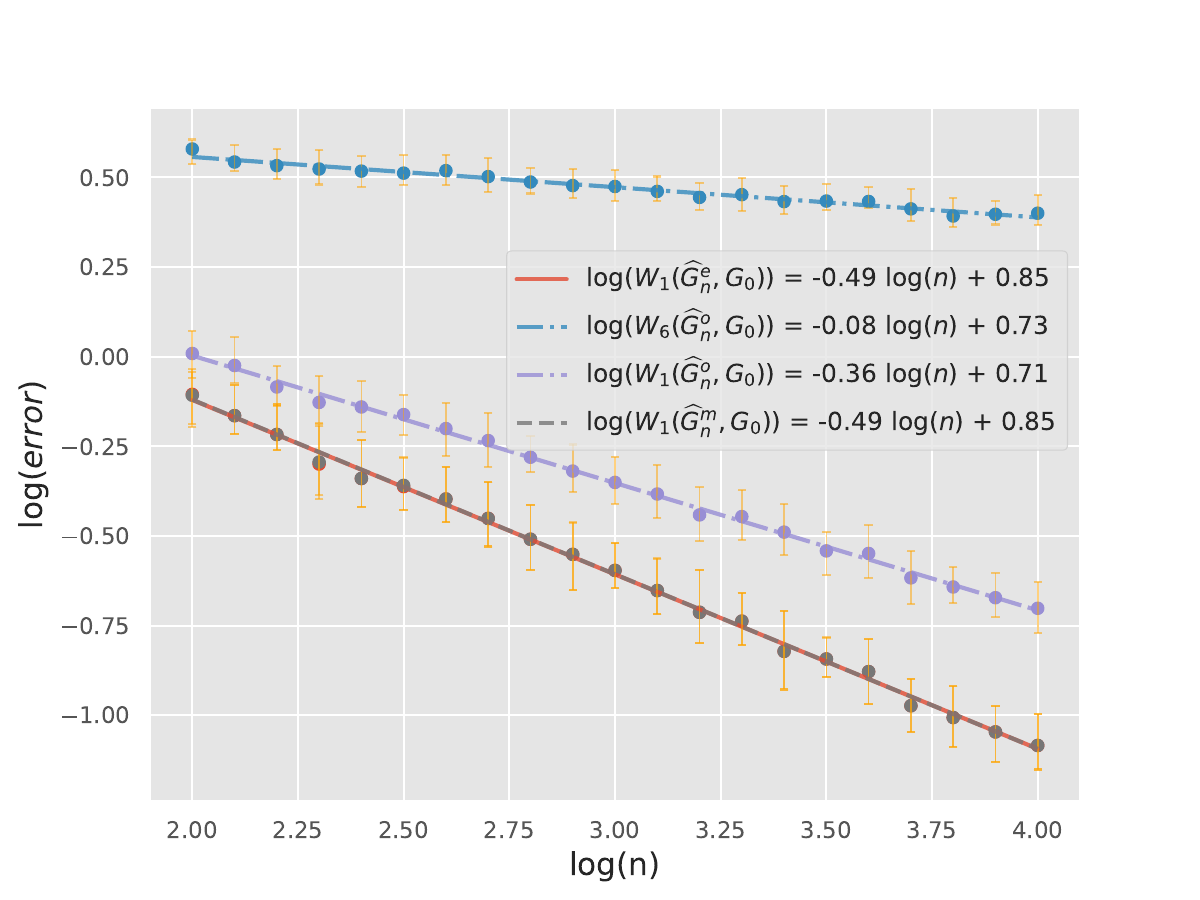}}
      \caption{\centering Rates of convergence of overfitted, exact-fitted, and merged mixing measures}\label{fig:rates-convergence}
\end{figure}

\subsubsection{Model selection with DIC}\label{subsubsec:DIC}
Next, we illustrate the property of the model selection method using DIC while comparing it with Akaike Information Criteria (AIC) and Bayesian Information Criteria (BIC). It is worth recalling that the power of the dendrogram is the capacity to represent hierarchical structures in subpopulations, whereas the various Information Criteria report a single number that can be used to select a reasonable number of subpopulations. With the dendrogram, we can additionally illustrate the merging process and hierarchy of atoms representing the data population's heterogeneity.
  
\paragraph{Simulation setup.} We focus on the model selection for the widely applied location-scale Gaussians. In each simulation setting and sample size $n$, $n$ i.i.d. observations were generated according to the specified setting. The sample sizes range from 10 to 10,000. We performed model selection methods with the upper bound number of components $k = 10$, replicated the experiments 100 times and recorded the number of times the information criterion chose the correct number of components and the average number of components. For AIC and BIC, this procedure contains fitting the mixture of $\kappa$ location-scale Gaussians to the data, for $\kappa = 1, \dots, 10$, then calculating the AIC and BIC scores to find the minimizer. For our method, only a single fit with $k=10$ components is needed, then the dendrogram is obtained from the inferred mixing measure, and the DIC is calculated accordingly. Unless stated otherwise, all the penalization scales $\omega_n$ in DIC are set to $\log(n)$. 
  
\paragraph{Well-specified regime.} First, we consider the setting where the true data-generating model is also a location-scale mixture of Gaussians. For simplicity, we choose it to be the distribution described in Section~\ref{subsubsec:rate-dendrogram}, where three Gaussian components are fairly separated with a uniform mixing proportion. It is noted that all methods handle the model selection in this setting similarly and well, as the true model belongs to the family of fitting models. Due to space constraints, the illustration of this setting is put in Appendix A.1.
 \begin{figure}[t!]
      \centering
      \subcaptionbox*{\scriptsize (a) Histogram with true density \par}{\includegraphics[width = 0.495\textwidth]{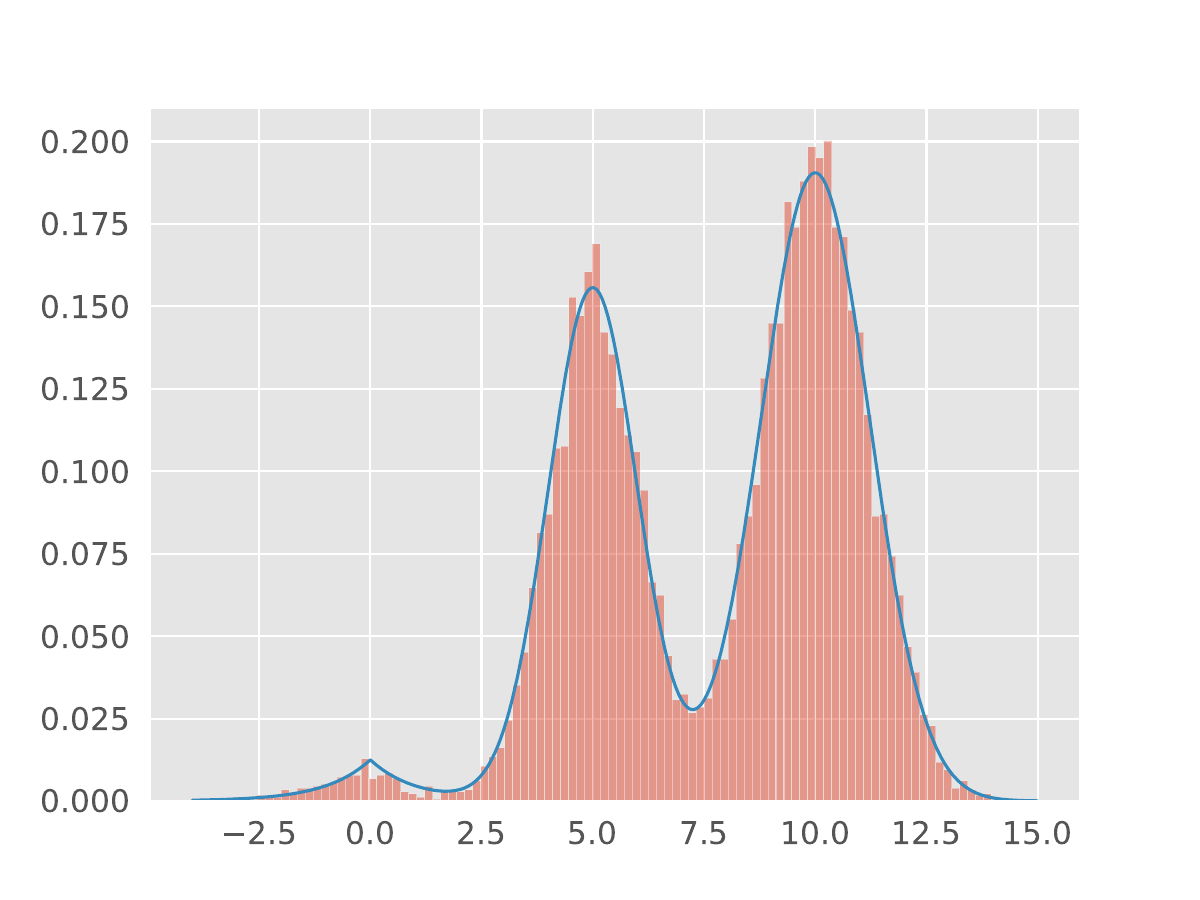}}
      \subcaptionbox*{\scriptsize (b) Dendrogram of mixing measure with 10 atoms\par}{\includegraphics[width = 0.495\textwidth]{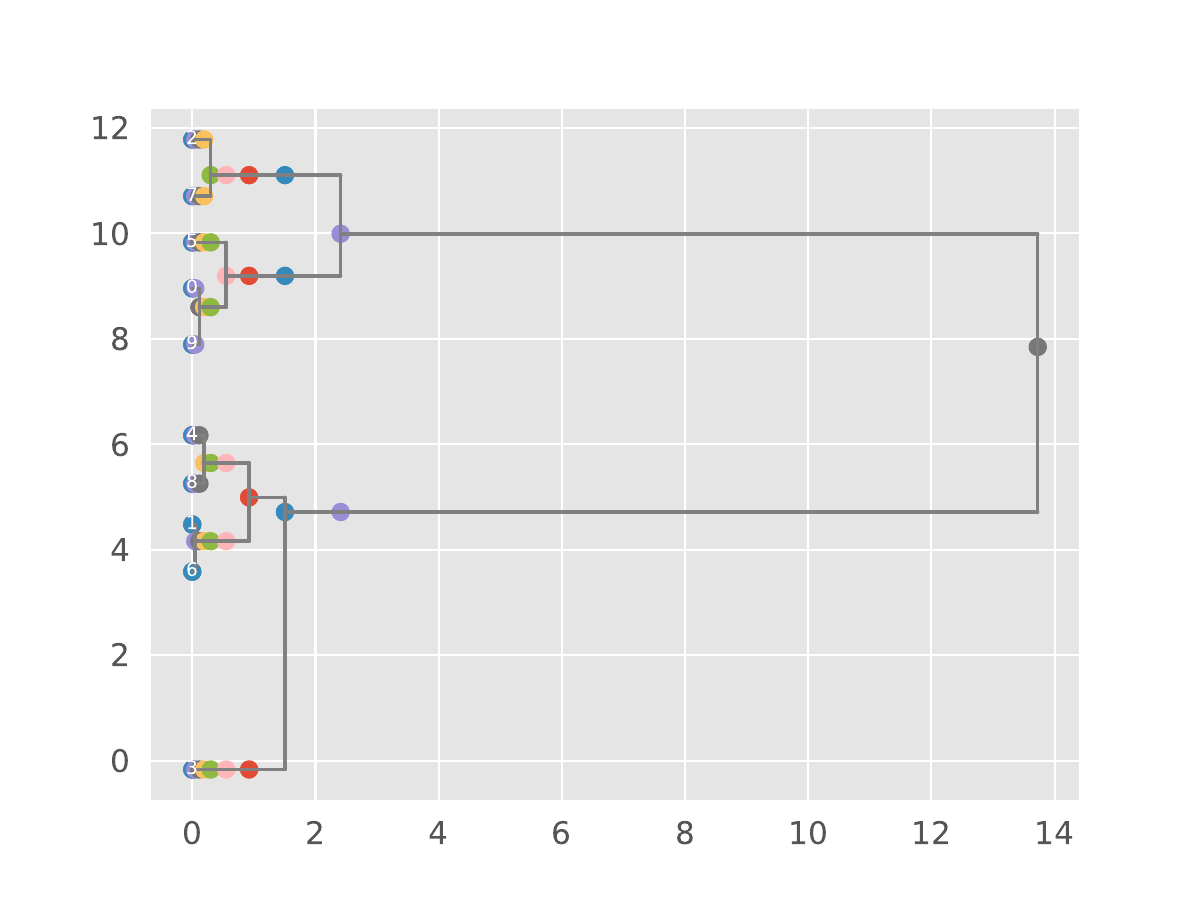}}
      \subcaptionbox*{\scriptsize (c) Proportion of choosing $k_0 = 2$ \par}{\includegraphics[width = 0.48\textwidth]{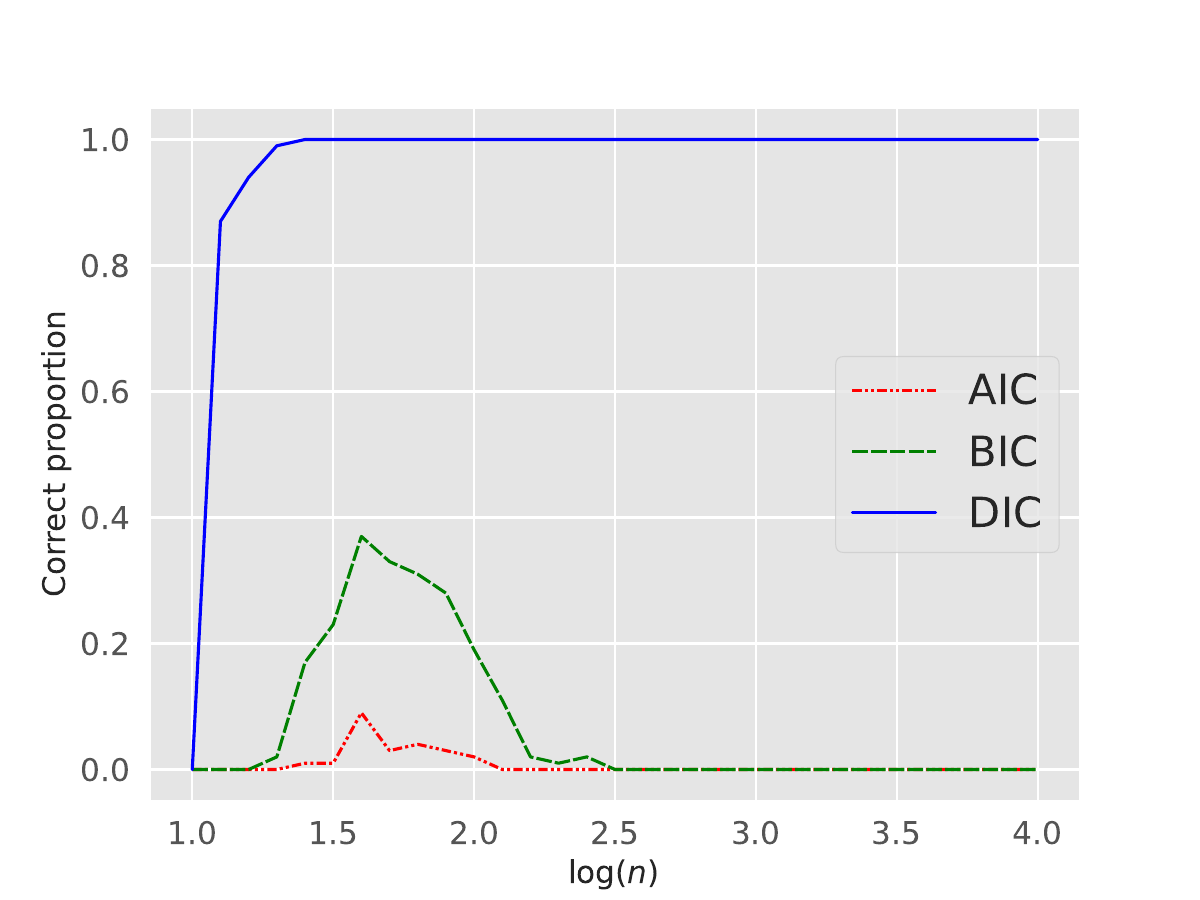}}
      \subcaptionbox*{\scriptsize (d) Average choices number of components \par}{\includegraphics[width = 0.48\textwidth]{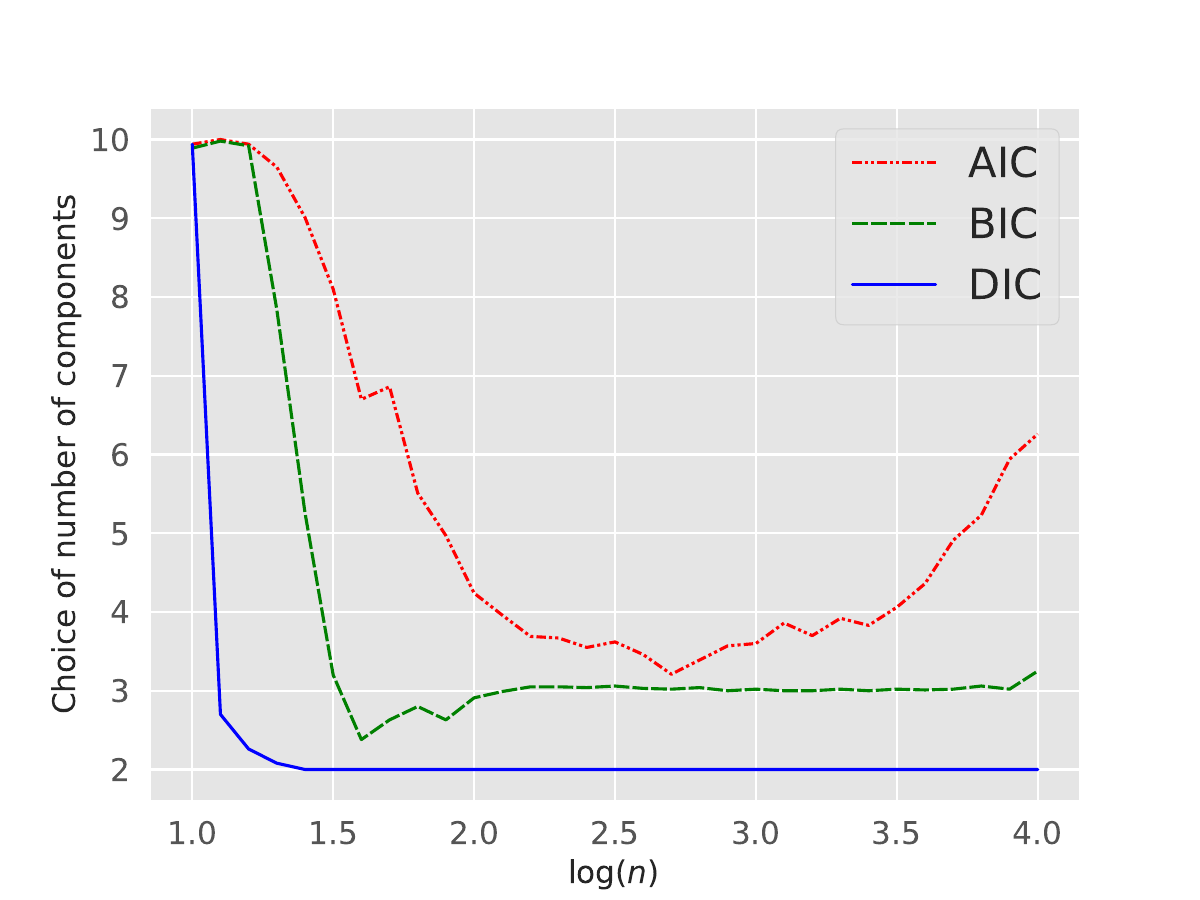}}
      \caption{\centering Experiments with $\epsilon$-contamination data}\label{fig:eps-contamination}
\end{figure}
  
\paragraph{Misspecified regime 1: $\epsilon$-contamination.} Next, we consider the setting where the true data-generating model is a mixture of Gaussians contaminated by a component from a different family of densities, where there are noticeable performances between AIC, BIC, and DIC. We choose the true model exactly the same as in \cite{cai2021finite}, where they show that model selection using Mixture of Finite Mixture (MFM) does not reliably specify the correct number of components under small contamination. The true generative distribution is $p_0 = (1-\epsilon) p_{G_0} + \epsilon q$, where $\epsilon > 0$ is the contaminated level, $p_{G_0}$ is a mixture of two location-scale Gaussians with $G_0 = p_1^0 \delta_{(\mu_1^0, (\sigma_1^0)^2)} + p_2^0 \delta_{(\mu_2^0, (\sigma_2^0)^2)} = 0.4 \delta_{(5, 1)} + 0.6 \delta_{(10, 1.5)}$, and $q$ is a Laplace distribution with location 0 and scale 1. The contaminated level is $\epsilon = 0.01$, which is relatively small. Thus, we wish to be able to detect that $k_0 = 2$ in this task. 

The density function and a histogram of data with $n=2000$ can be seen in Figure~\ref{fig:eps-contamination}(a). The proportion of picking the correct number of components $k_0 = 2$ and the average numbers of components are plotted in Figure~\ref{fig:eps-contamination}(c,d). It shows the similar behaviour of AIC and BIC: they can only detect $k_0 = 2$ with sample sizes that are not too small or large, which also aligns with the behaviour of MFM in~\cite{cai2021finite}. The DIC prefers atoms with significantly large proportions, so it tends to pick $k_0 = 2$ and is more robust to contamination. As the sample size gets large, BIC interprets the data as a mixture of three components, while AIC chooses the model with many more components. Although DIC detects $k_0=2$, when we look into the dendrogram (Figure~\ref{fig:eps-contamination}(b)), we can still see there is a small portion of the Laplace component, which is soon merged into the closest normal component. This demonstrates that the hierarchical tree of mixing measures is considerably more informative than an estimate of the number of components when the model is misspecified.
   
\paragraph{Misspecified regime 2: Skew-normal mixtures.} In this case, the true data generating model is a mixture of two skew-normal distributions $p_{0}(x) = 0.4 * \text{Skew}(5, 1, 20) +  0.6 * \text{Skew}(8, 1.5,20)$. Recall that the density function of the skew-normal distribution with parameters $\mu$ (location), $\sigma^2$ (scale) and $\alpha$ (skewness) is $f(x|\mu, \sigma^2, \alpha) = 2\phi(x) \Phi(\alpha x)$, where $\phi$ and $\Phi$ are the pdf and cdf of the normal distribution with mean $\mu$ and variance $\sigma^2$. The skewness levels $\alpha$ of the two components in the true distribution are equal to 20, which indicates that each component is skewed to the right. The dendrogram with model selection results can be seen in Figure~\ref{fig:skew}. Because the true data distribution is not a mixture of location-scale Gaussians, the number of components chosen by AIC and BIC seems to diverge to the upper bound of $10$ as the sample size increases. For DIC, the term regarding the dendrogram's heights, which is the squared length of the Wasserstein projection (Proposition~\ref{prop:Wasserstein-variational}), puts more penalization on its score when the mixing measure is close to a mixing measure with fewer atoms. Hence, it actively penalizes redundant atoms and excess masses. The dendrogram of the mixing measure of mean parameters in Figure~\ref{fig:skew}(b) indicates that there are two large subpopulations of atoms.

\begin{figure}[t!]
      \centering
      \subcaptionbox*{\scriptsize (a) Histogram with true density \par}{\includegraphics[width = 0.495\textwidth]{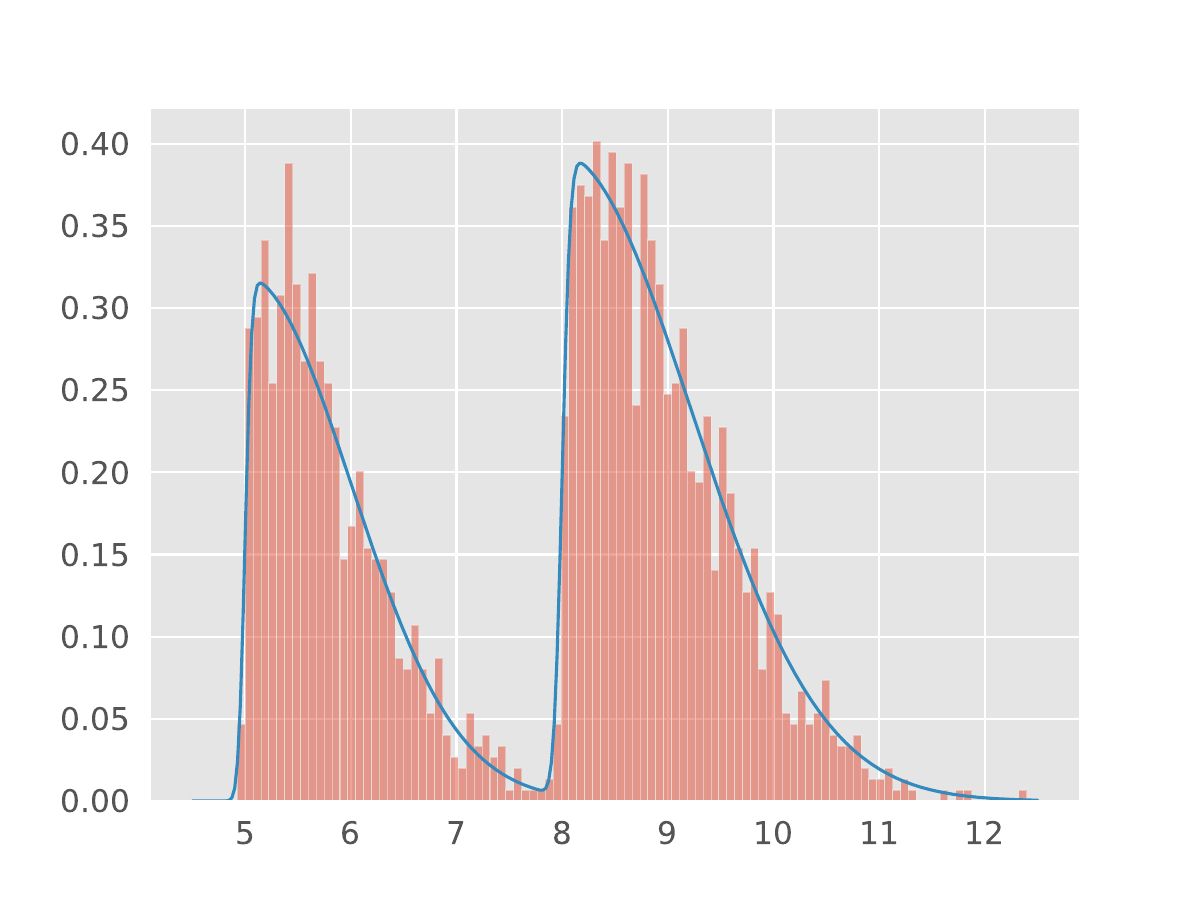}}
      \subcaptionbox*{\scriptsize (b) Dendrogram of mixing measure with 10 atoms\par}{\includegraphics[width = 0.495\textwidth]{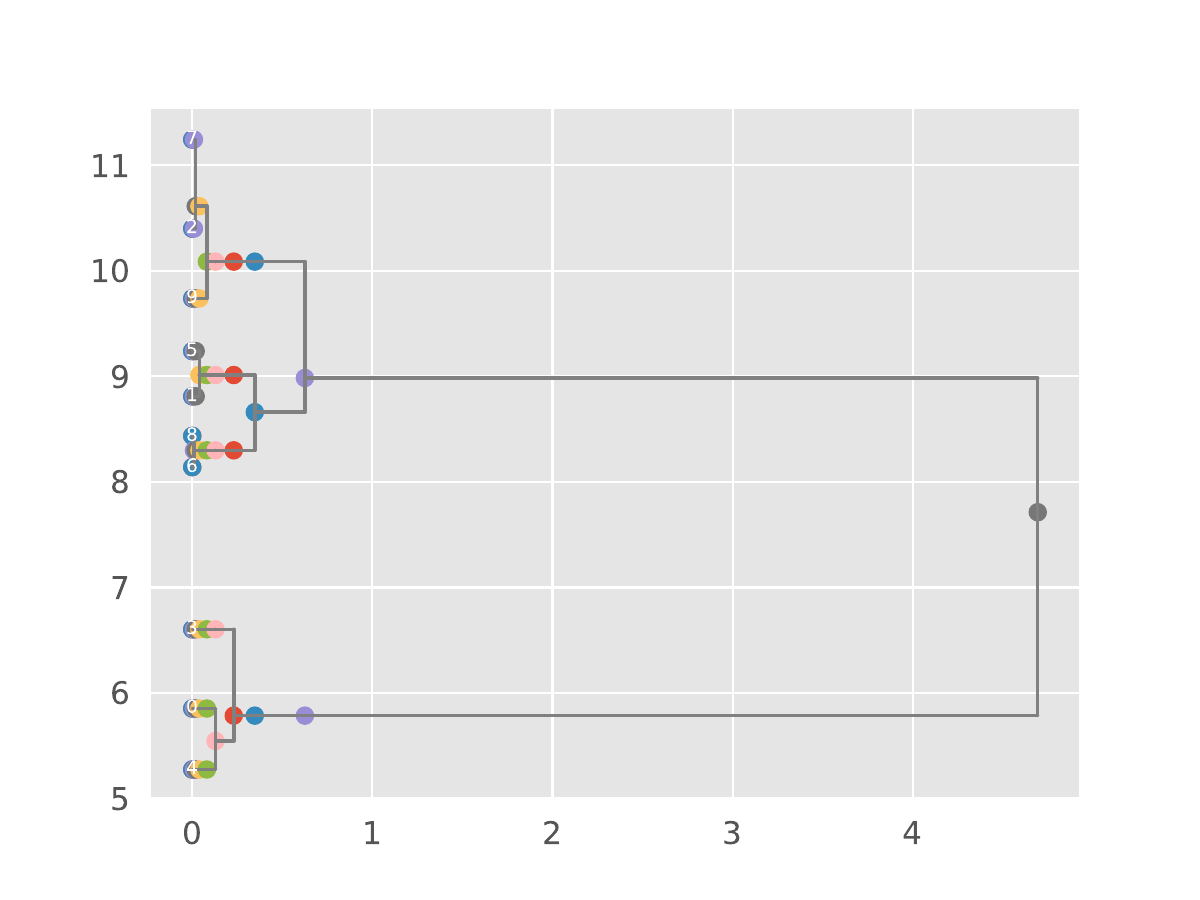}}
      \subcaptionbox*{\scriptsize (c) Proportion of choosing $k_0 = 2$ \par}{\includegraphics[width = 0.48\textwidth]{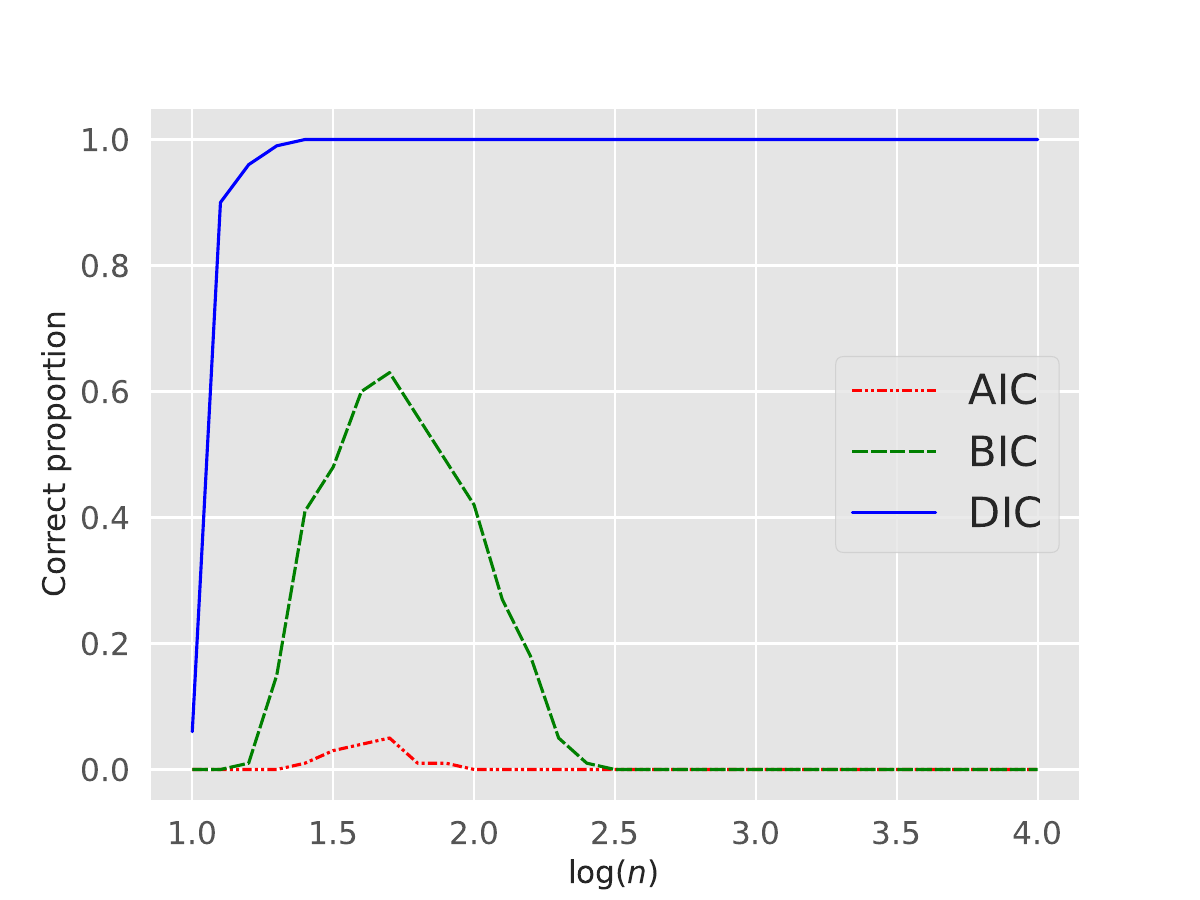}}
      \subcaptionbox*{\scriptsize (d) Average choices number of components \par}{\includegraphics[width = 0.48\textwidth]{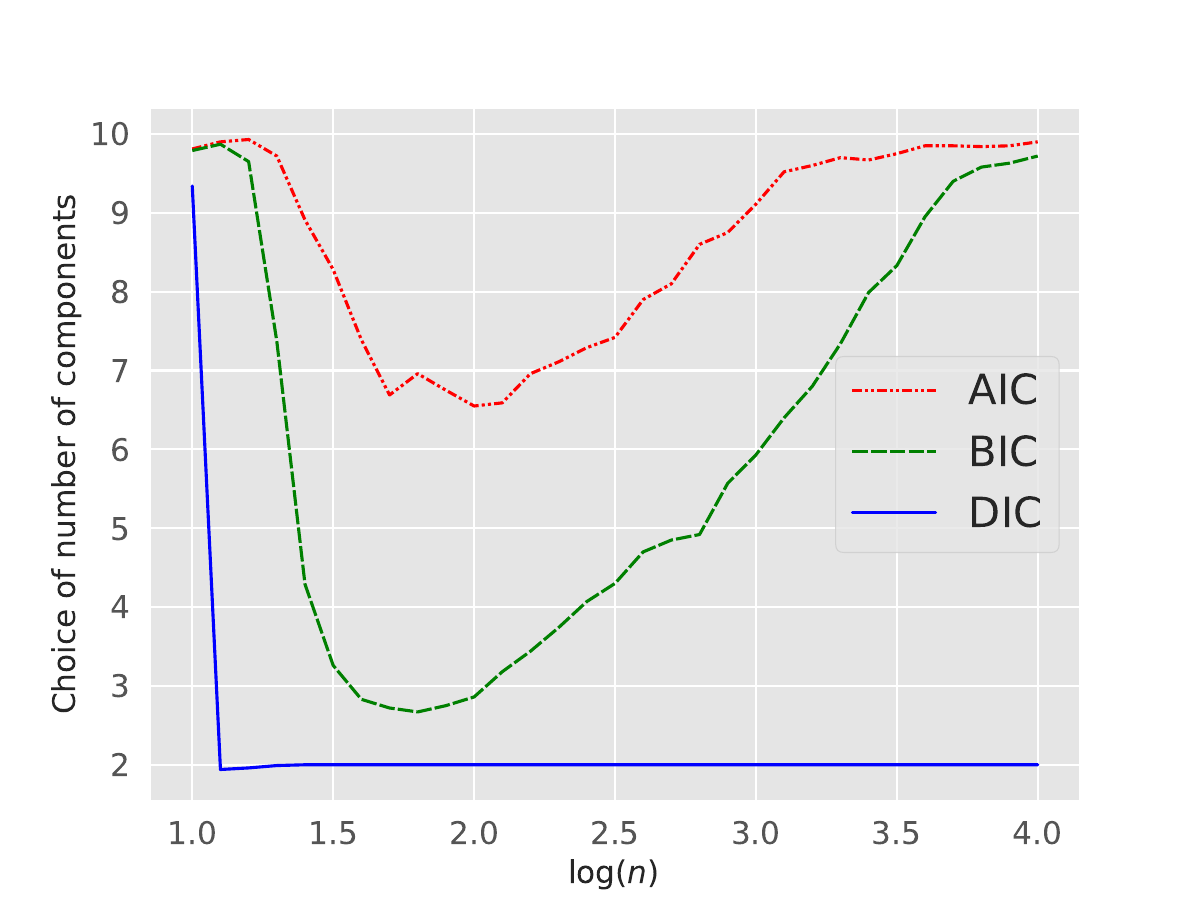}}
      \caption{\centering Experiments with the mixture of skew-normal data}\label{fig:skew}
\end{figure}




  
\subsection{Real data illustrations}

\begin{figure}[t!]
      \centering
      \subcaptionbox*{\scriptsize (a) Dendrogram in the first PC \par}{\includegraphics[width = 0.495\textwidth]{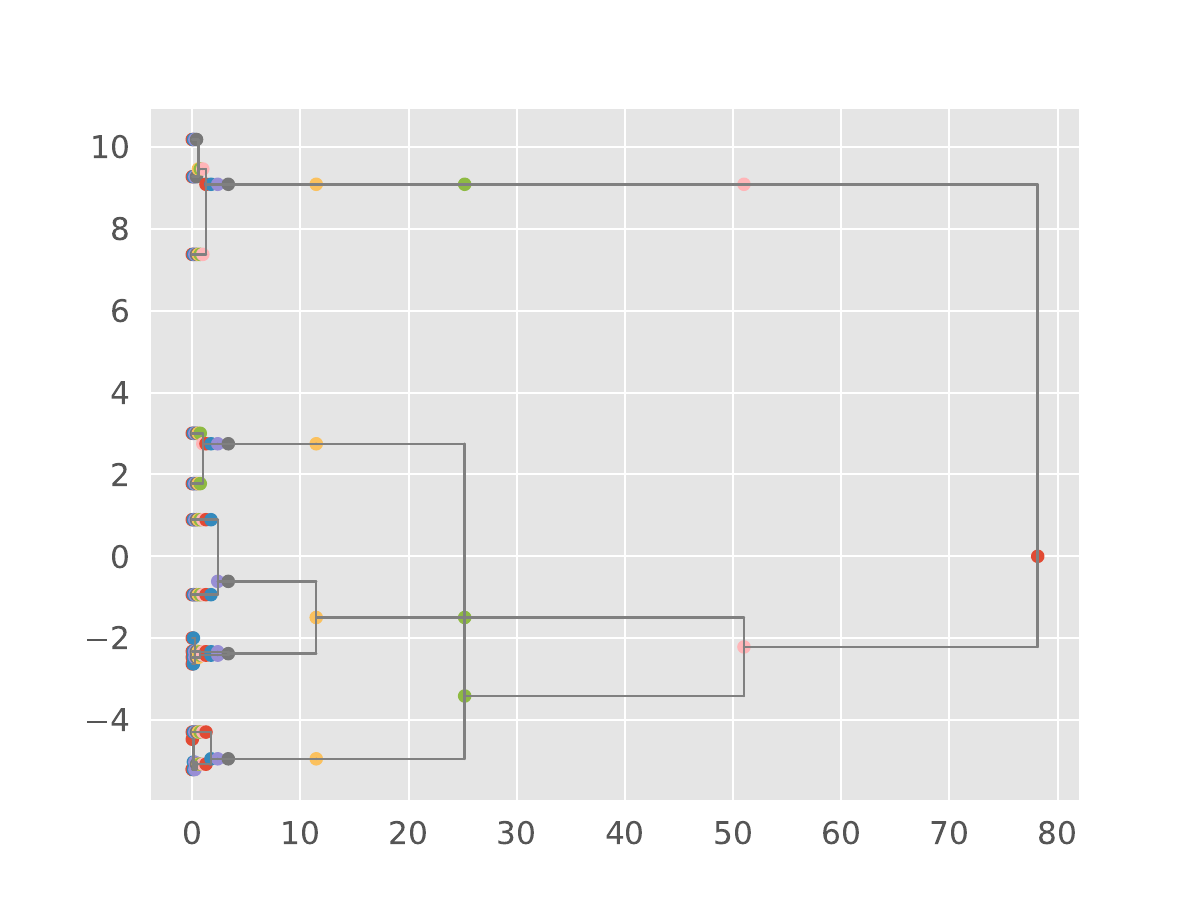}}
      \subcaptionbox*{\scriptsize (b) Dendrogram in the second PC\par}{\includegraphics[width = 0.495\textwidth]{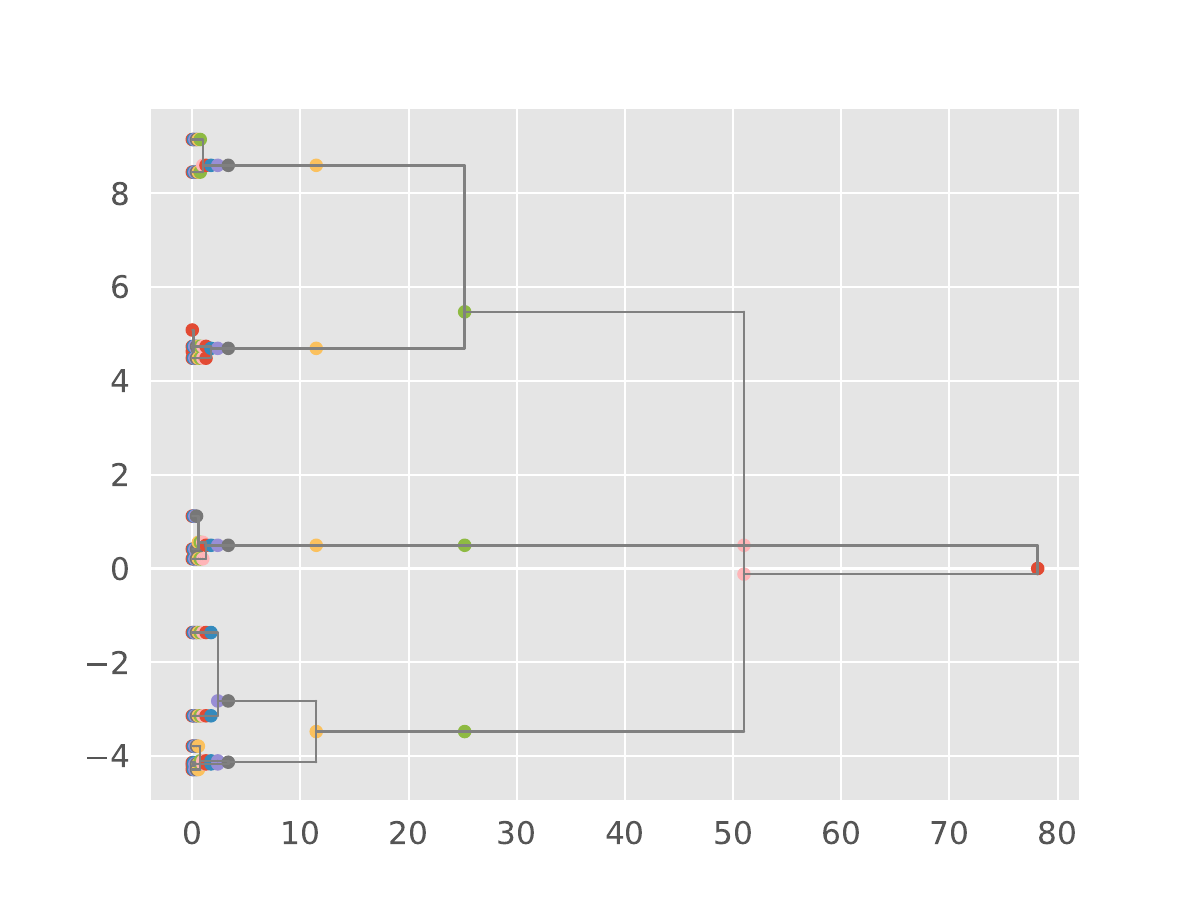}}
      \subcaptionbox*{\scriptsize (c) Heights between levels in dendrogram\par}{\includegraphics[width = 0.48\textwidth]{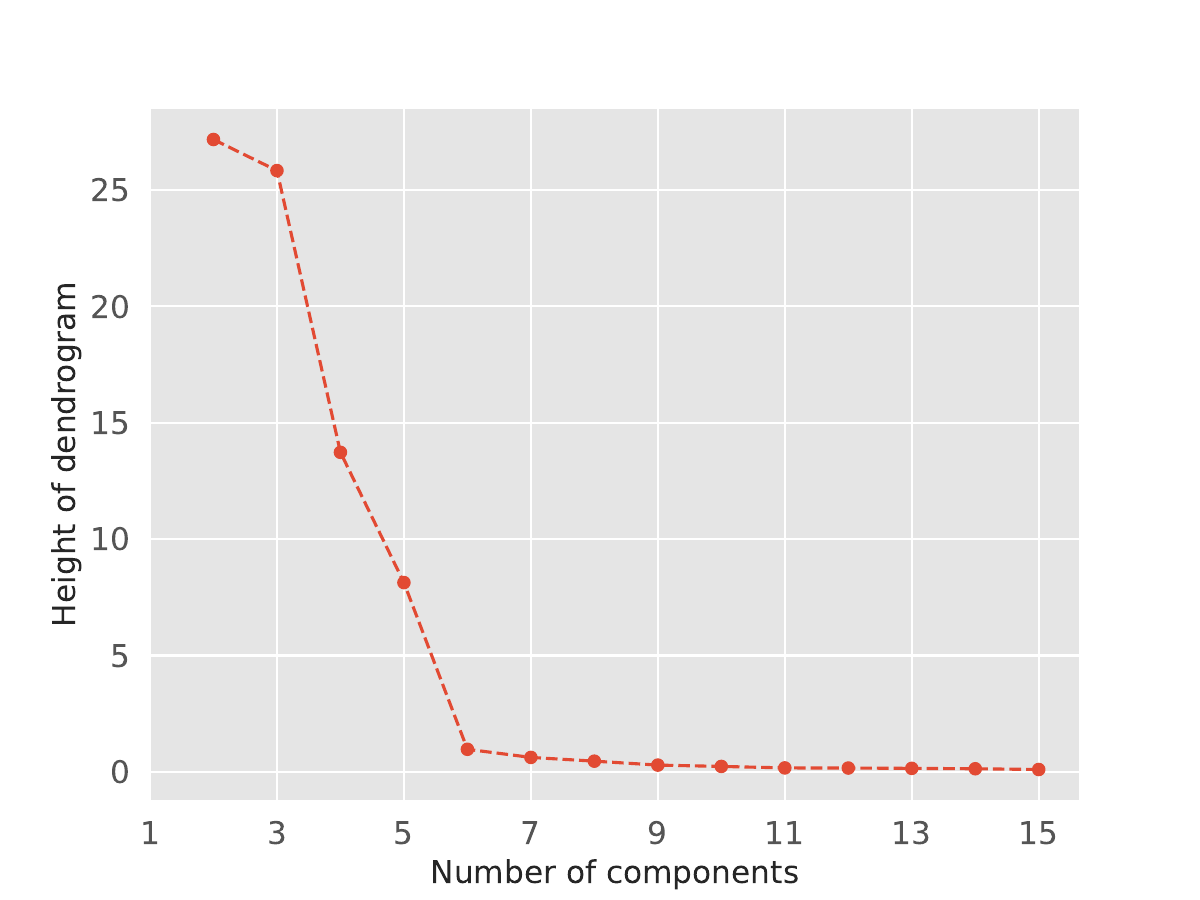}}
      \subcaptionbox*{\scriptsize (d) Likelihood of levels in dendrogram \par}{\includegraphics[width = 0.48\textwidth]{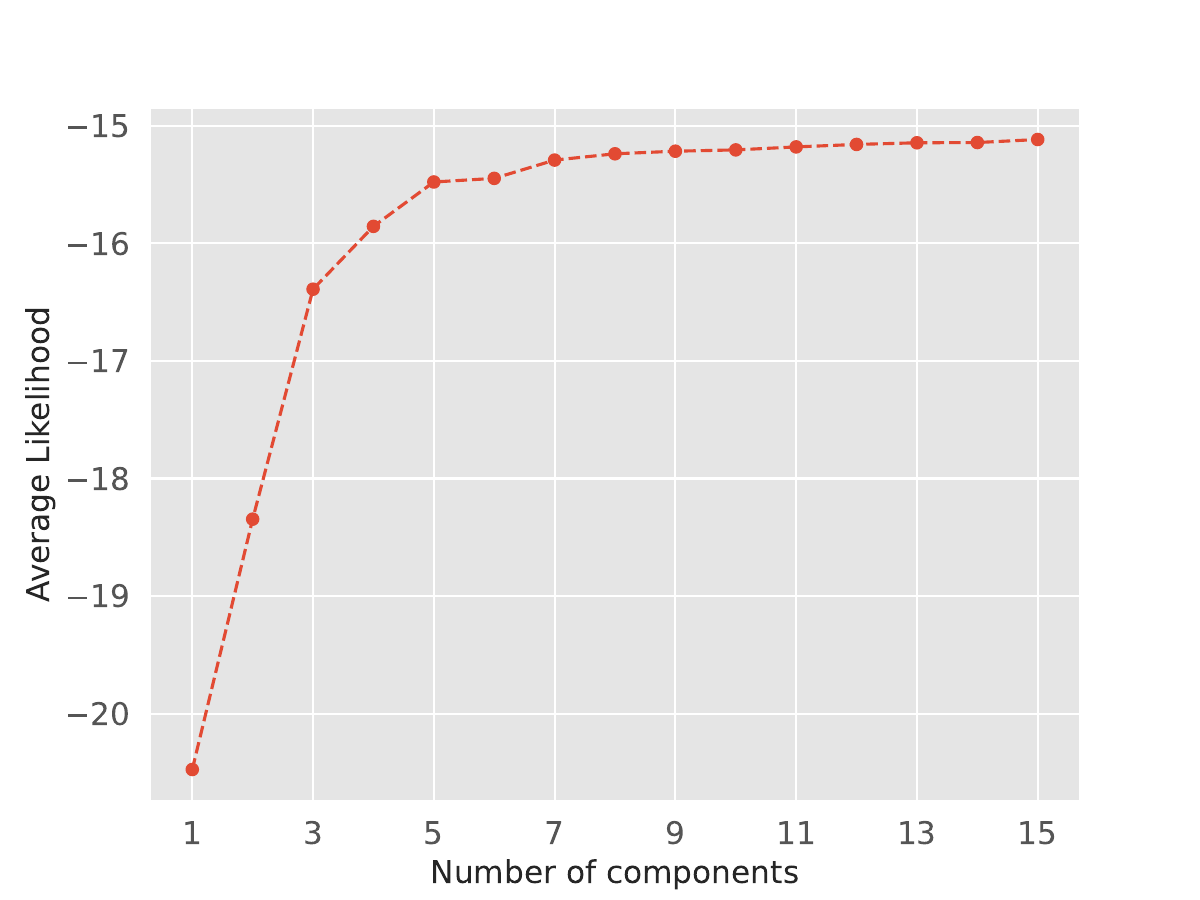}}
      \caption{\centering Dendrogram of mixing measures inferred from Single-cell data}\label{fig:sce}
\end{figure}

We consider an application of the dendrogram of mixing measures in Single-cell RNA-sequencing data presented in \cite{zheng2017massively}, where they collected 68,000 peripheral blood mononuclear cells and labelled them based on expression profiles of 11 reference transcriptomes from known cell types. There are five cell types in the dataset: memory T cells, B cells, naive T cells, natural killer cells, and monocytes. After cleaning the data (dropping low-count cells and performing log pseudo-count transform), there are 41,159 cells left. 

We project the unlabeled data onto the first 10 Principal Component (PC) spaces and fit with a mixture of location-scale Gaussians with 15 (mixture) components. The dendrograms of the fitted mixing measure in the first two Principle Components are plotted in Figure~\ref{fig:sce}(a,b), while the corresponding heights and likelihood of each level in the dendrogram can be seen in (c\&d). The theory informs us that the heights of the overfitted mixing measures tend to 0, and the likelihoods at exact-fitted and overfitted levels have the same limit as the sample size gets large. We can see the heights are approximately 0, and the likelihood does not increase much after the fifth level. Hence, by inspecting both the plots of the heights and likelihood, we conclude that there should be five cell types in the data. 

For this data, BIC chooses 13 components, and DIC chooses 3. We notice that DIC can prefer under-fitted models sometimes due to the magnitude mismatch between likelihoods and parameters. To overcome this issue, one can plot the height and the likelihood along the dendrogram separately and perform model selection heuristically using those information. Otherwise, one can eliminate all orders that he believes are under-fitted before using DIC. Indeed, when starting from 5 to 15 components, DIC chooses 5 as the best fit. Finally, by inspecting the dendrogram and using nearest neighborhood classification, we see that the two merging components of the mixing measures with 5 components correspond to the class of memory T cells and naive T cells, which are most similar to each other compared to other cell types. Hence, not only the has the dendrogram efficiently performed model selection, it has also revealed hierarchy in complex data, thus enhancing interpretability of the mixture model parameter estimates.

\section{Conclusion and future investigation}
  
We proposed a method for dendrogram construction in mixture modelling-based inference, established the asymptotic properties of our method, and demonstrated its usefulness for addressing heterogeneous data. 
Our method can also be viewed as providing the statistical foundation for a class of data-driven and pragmatic hierarchical clustering algorithms widely employed in practice. Our work shows that mixture models continue to be useful as a tool for unravelling the complexity underlying heterogeneous data, and one can do that by producing a nested class of latent structures induced by the model (namely, the dendrogram of mixing measures) rather than relying on a single point estimate of the original (mixture) model. The outcome of our estimation procedure is more informative and amenable to robust inference. We established this in theory, e.g., the dendrogram can enable a fast parameter estimation rate even for weakly identifiable models. We also demonstrated in practice the robustness of the dendrogram and the learning of clustering even when the model is misspecified.
Possible extensions include hierarchical models satisfying identifiability conditions, such as Hidden Markov Model~\cite{gassiat2014posterior} and admixture model~\cite{nguyen2015}. Furthermore, the asymptotic behavior of the dendrogram under the model misspecification is worth studying. Finally, different settings may open up the possibility of choosing different dissimilarities and linkages in constructing the dendrogram for inference.

\section*{Acknowledgement}
This research was supported in part by the National Science Foundation (via grants DMS-2015361 (XLN), DMS-2052653 (JT), the National Institute of General Medical Sciences of the NIH under award number R35GM151145, and a research gift from Wells Fargo (XLN). LD and SAM are supported by the NSF-Simons Southeast Center for Mathematics and Biology (SCMB) through the grants from NSF DMS-1764406 and Simons Foundation/SFARI 594594. The content is solely the responsibility of the authors and does not necessarily represent the official views of the NIH.
\bibliography{dat}
\bibliographystyle{apalike}

\newpage
\appendix
\begin{center}
{\bf{\LARGE{Supplements to "Dendrogram of mixing
measures: Learning latent hierarchy and model
selection for finite mixture models"}}}
\end{center}

\section{Additional experiments}
\subsection{Dendrogram in well-specified setting}
Here, we provide an illustration and model selection results for the data in Section~\ref{subsubsec:rate-dendrogram} and Section~\ref{subsubsec:DIC}, the well-specified setting. A simulated data of 300 observations with contour plots of 3 Gaussian components can be seen in Figure~\ref{fig:merge_demo}(a). We started by overfitting the data with a mixture of 10 components. The results contour plots can be seen in Figure~\ref{fig:merge_demo}(b). Because location-scale Gaussian components spread and fit into several parts in the data, the convergence of parameter estimation is predictably bad, as described in Section~\ref{sec:merge-weak}. However, as those components merge together using the described algorithm, the merged parameters become increasingly concentrated around the true parameters. At the exact-fitted level, the parameter estimates nearly coincide the true parameters. The simulation in Section~\ref{subsubsec:rate-dendrogram} shows that the convergence rate of the merged mixing measures is comparable with that of the exact-fitted mixing measure estimates. Next, we compared DIC with AIC and BIC for model selection for this data. For each sample size $n$ (ranging from 10 to 10,000), we generate $n$ i.i.d. observation from this distribution, fit Gaussian mixture models with up to ten components, and record the optimal number of components chosen by AIC, BIC, and DIC. The penalization scale is chosen to be $\omega_n = 0.5 \log(n)$ for DIC. The experiment is replicated 100 times, and the percentage of replications that each model selection method picks correct $k_0 = 3$ is plotted in Figure~\ref{fig:well-specified}. It is observed that all methods have similar performance in this setting because the true generative model lies in the space of fitted models. 

\begin{figure}[t!]
      \centering
      \subcaptionbox*{\scriptsize (a) Data with true contour plot \par}{\includegraphics[width = 0.32\textwidth]{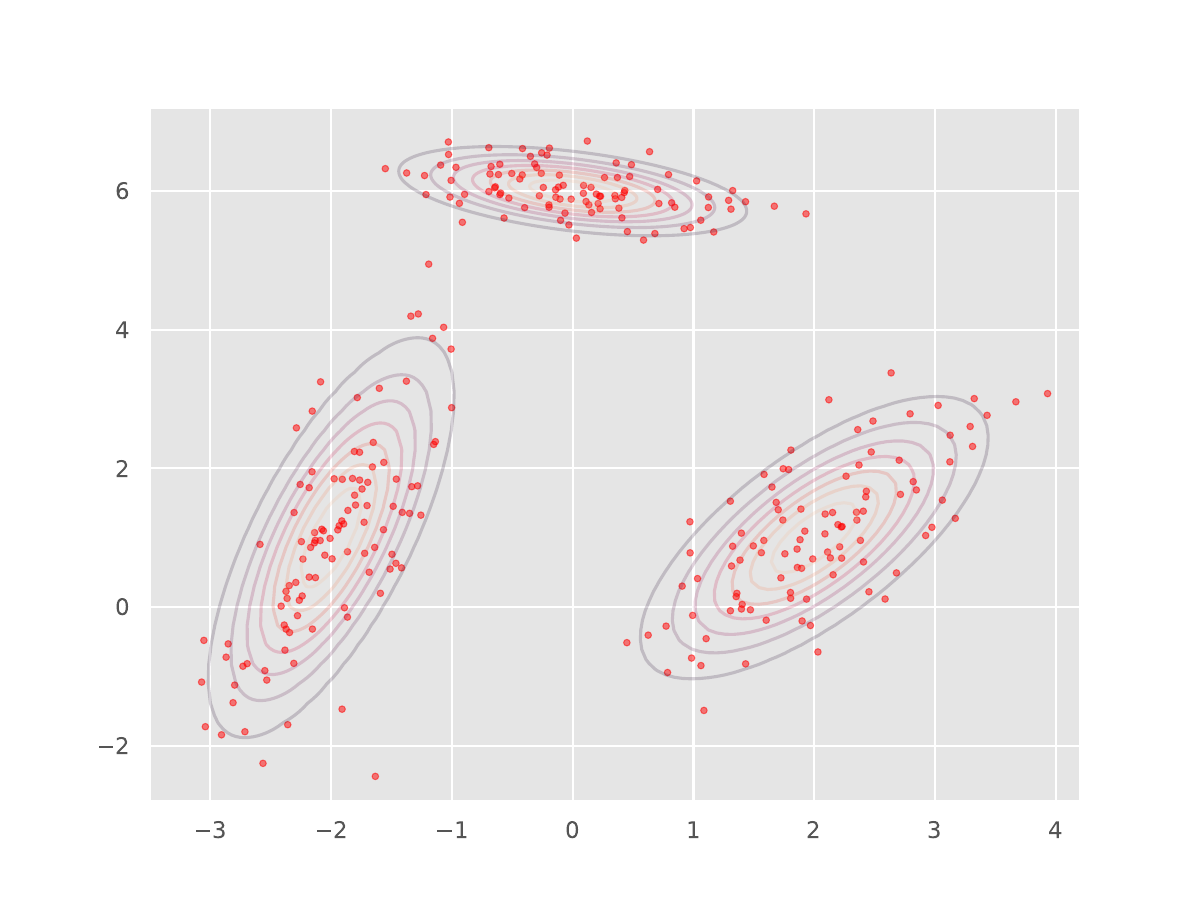}}
      \subcaptionbox*{\scriptsize (b) Over-fitting $k=10$ \par}{\includegraphics[width = 0.32\textwidth]{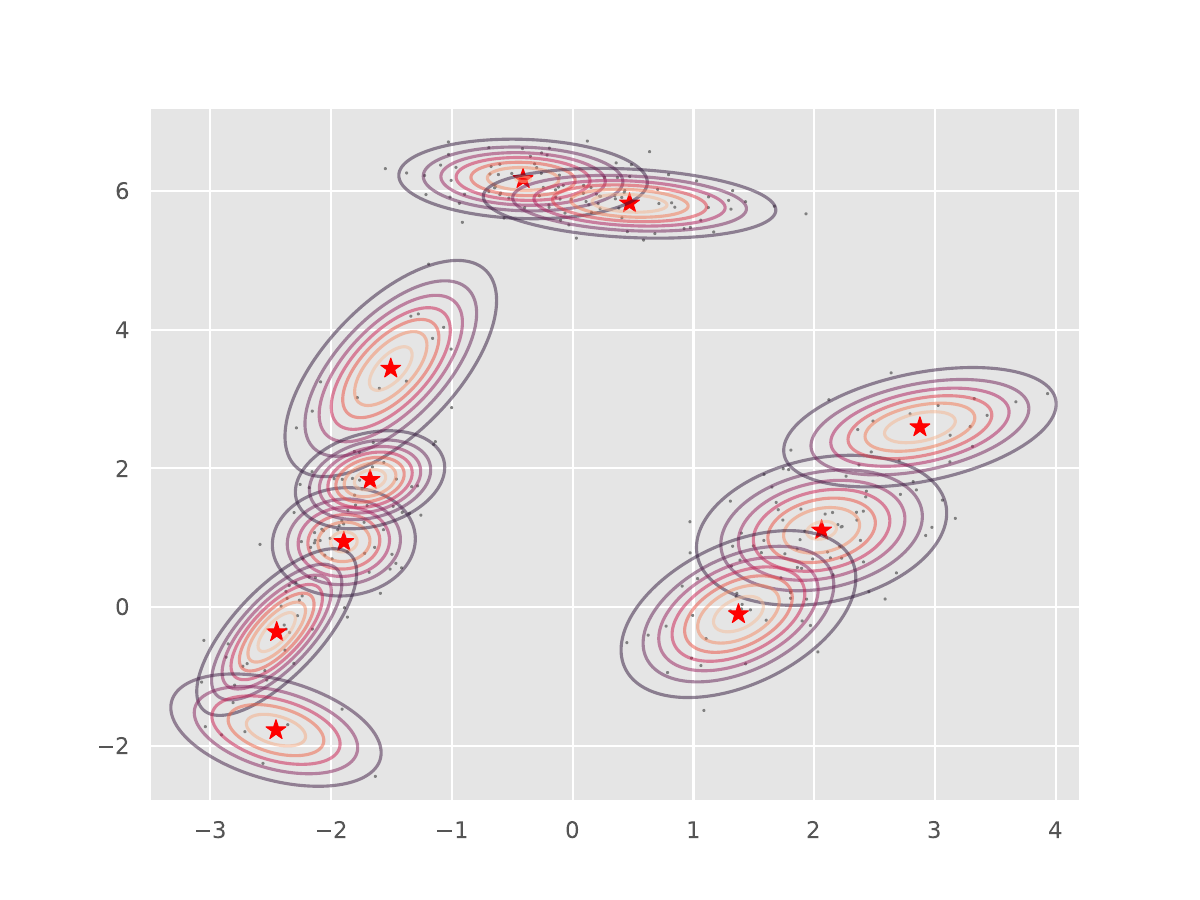}}
      \subcaptionbox*{\scriptsize (c) Merge $\kappa = 8$ \par}{\includegraphics[width = 0.32\textwidth]{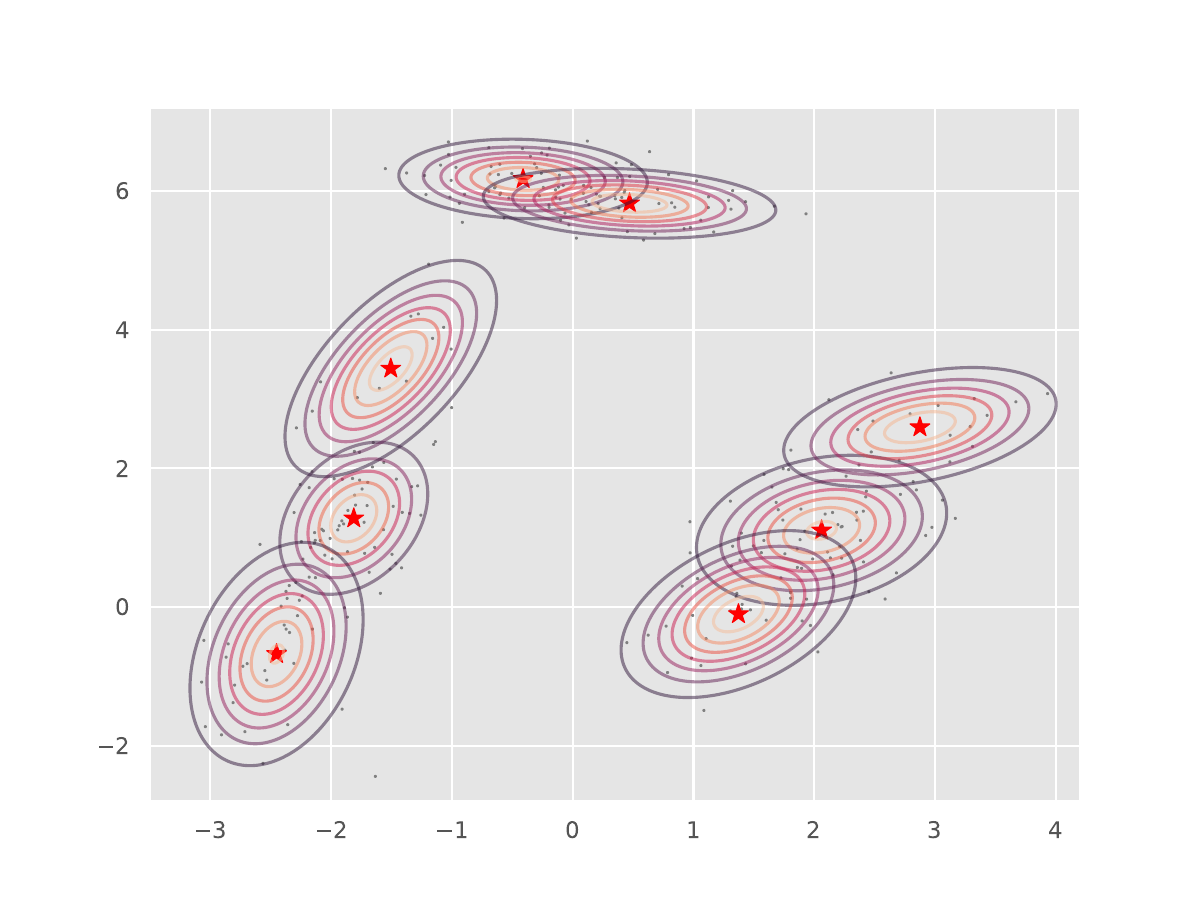}}
      \subcaptionbox*{\scriptsize (f) Merge $\kappa = 3$ \par}{\includegraphics[width = 0.32\textwidth]{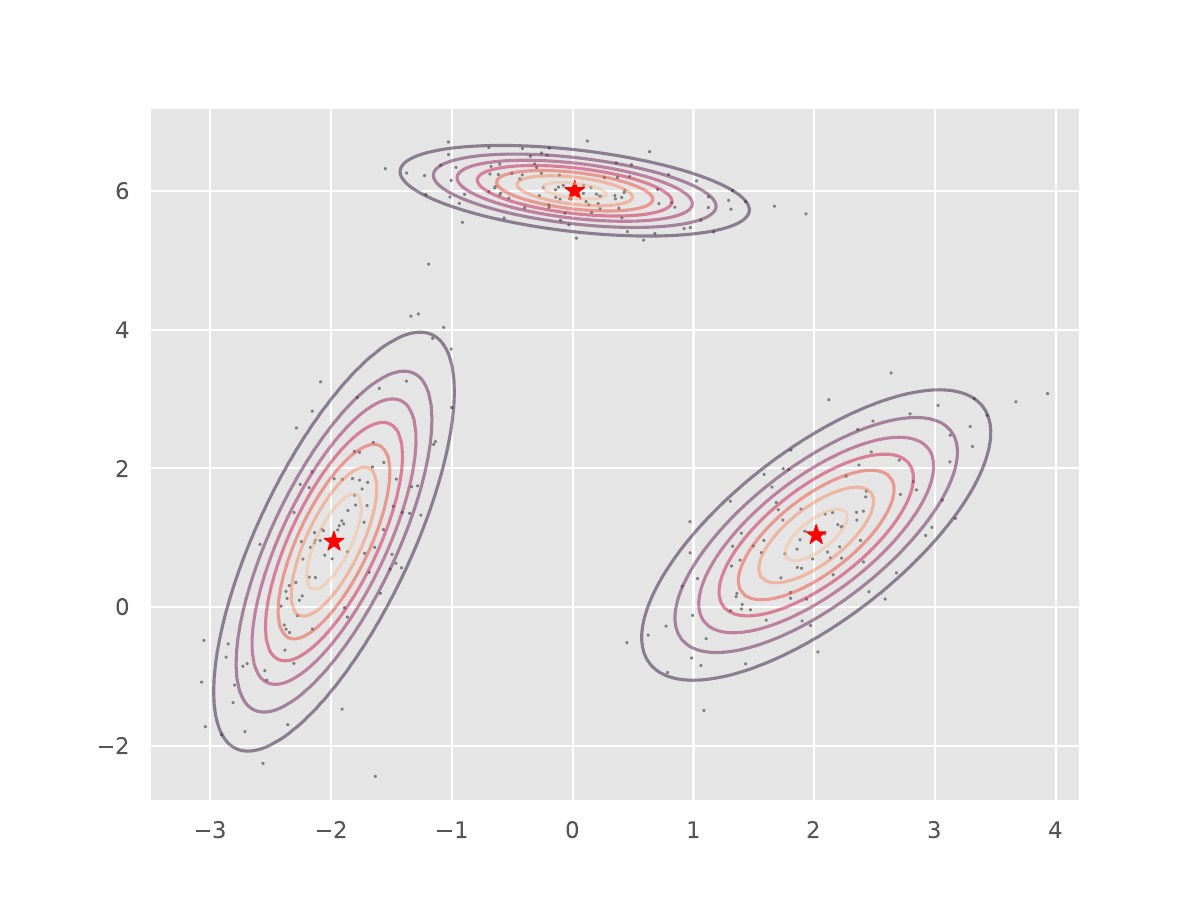}}
      \subcaptionbox*{\scriptsize (e) Merge $\kappa = 4$ \par}{\includegraphics[width = 0.32\textwidth]{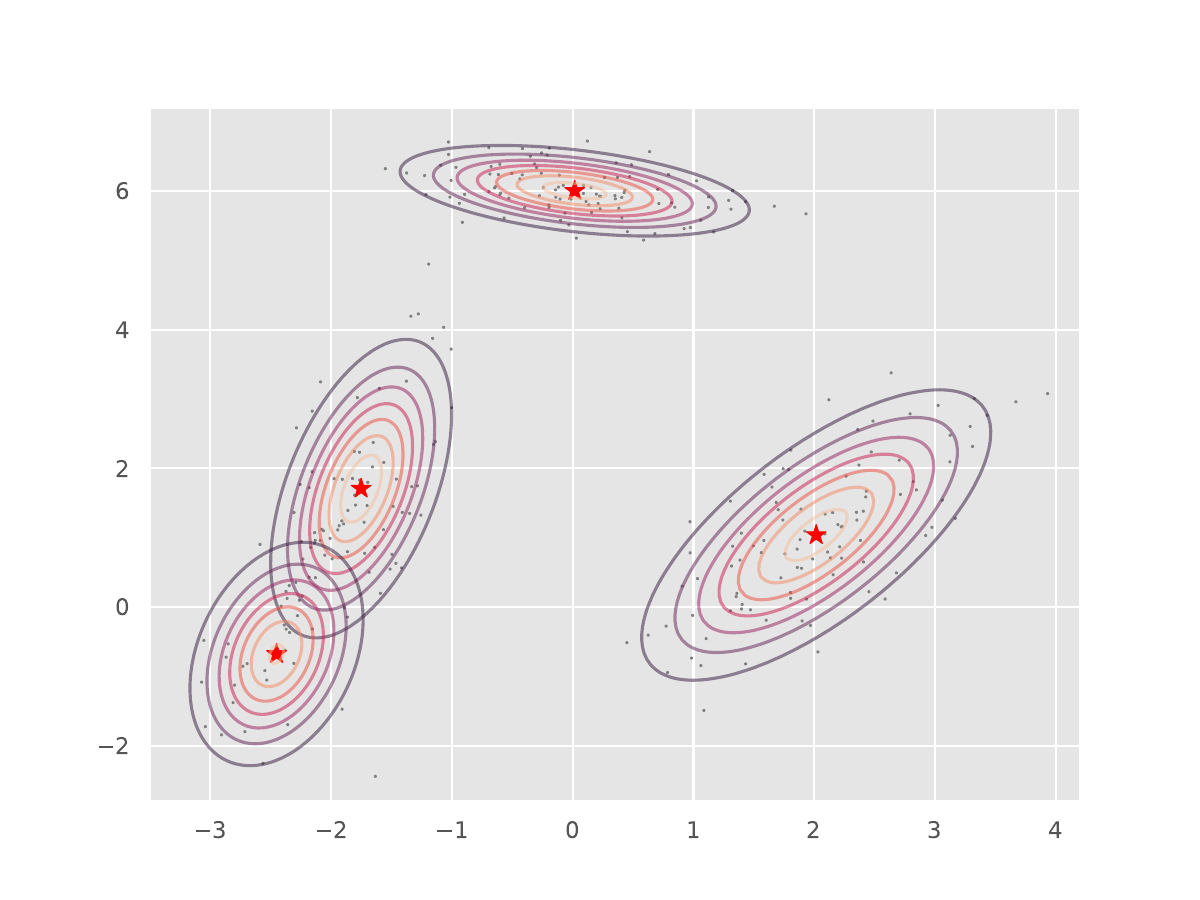}} 
      \subcaptionbox*{\scriptsize (d) Merge $\kappa = 6$ \par}{\includegraphics[width = 0.32\textwidth]{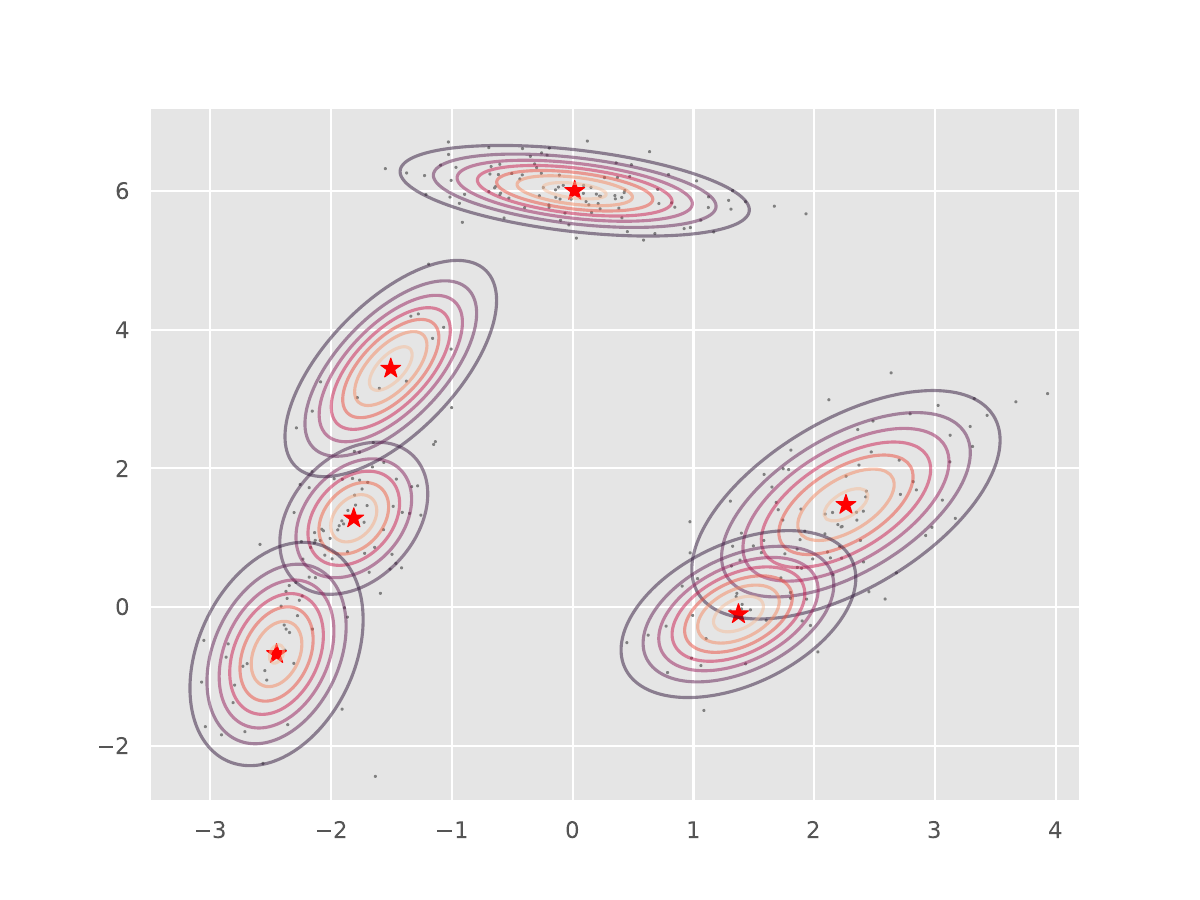}}  
      \caption{\centering Mixing measures on the dendrogram with $k=10$ and $k_0=3$ (displaying in the clockwise order for ease of comparing between the merged measure and the truth)} \label{fig:merge_demo}
\end{figure}
\begin{figure}[t!]
    \centering
    \includegraphics[width=0.6\textwidth]{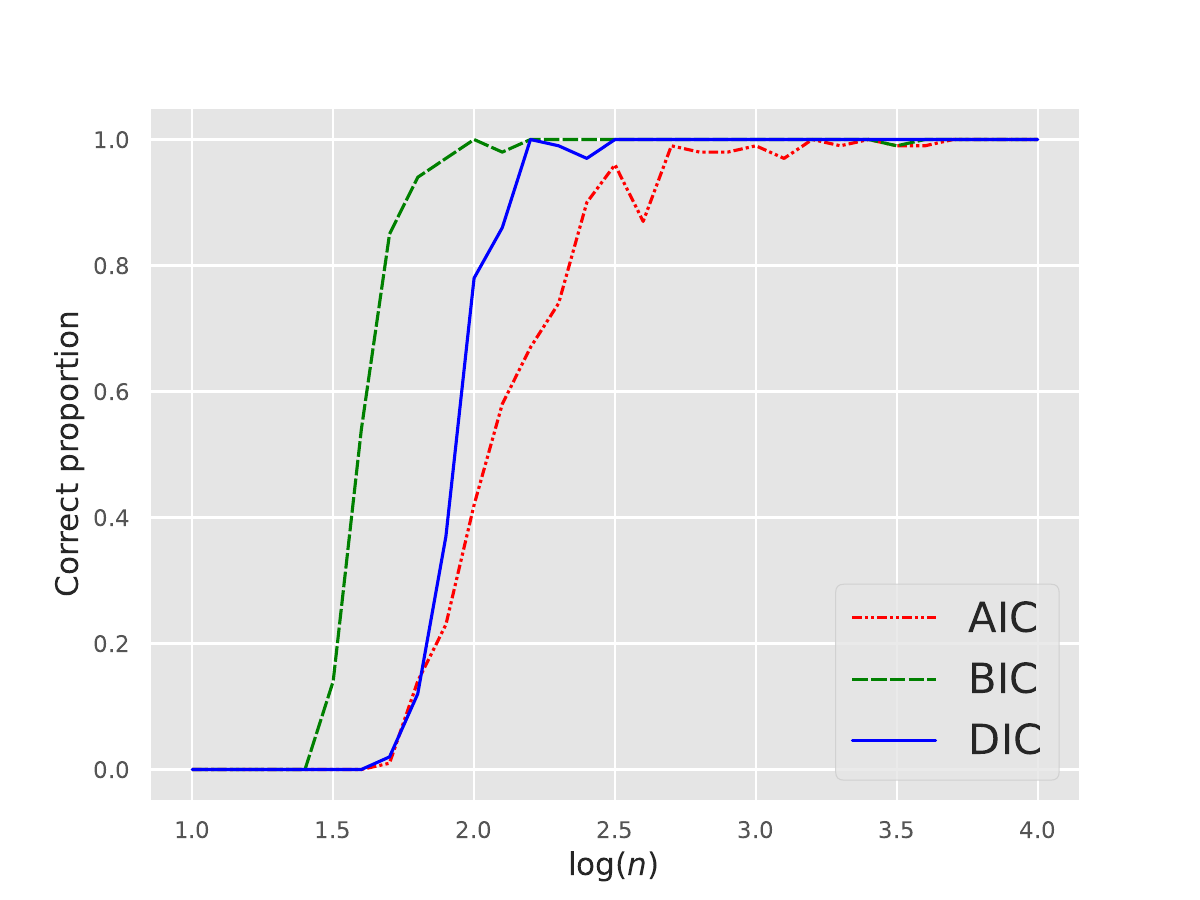}
    \caption{\centering Proportion of choosing correct $k_0 = 3$ in the well-specified setting}
    \label{fig:well-specified}
\end{figure}

\subsection{Popular datasets in mixture modelling}
In this section, we will consider the application of dendrogram for some famous datasets for benchmarking mixture modelling methods \cite{32Zoe}: Acidity, Enzyme, and Galaxy data. The first case study is the Acidity dataset \cite{acid}, which consists of the log acidity index for 155 lakes in the North-Eastern United States. The Enzyme data \cite{enzyme} consists of measurements of enzymatic activity in blood for an enzyme involved in the metabolism of carcinogenic substances (velocity and substrate concentration) for a group of 245 unrelated
individuals. The Galaxy data \cite{galaxy} is a small dataset of 82 measurements of galaxy speeds
from 6 segments of the sky. For each data, we plot the histogram with kernel density estimation. We then fit a location-scale Gaussian mixture model with 10 components and plot the dendrogram with the heights and likelihoods at different levels. Based on the theory, we want to choose the number of components so that the heights of all levels after this number are approximately 0. For the Acidity data (Figure~\ref{fig:acidity}), using two components make a good fit, and indeed both AIC and DIC choose $k=2$. The Enzyme data (Figure~\ref{fig:enzyme}) has a heavy right tail. Hence, by looking at the heights of the dendrogram (Figure~\ref{fig:enzyme}(e)), we can choose between 2 or 3 components. Both AIC and DIC choose 2. The dendrogram plot in Figure~\ref{fig:enzyme}(b) is informative as it describes that two components corresponding to the right tail in level $\kappa=3$ will merge to get the two chosen components. For Galaxy data, its histogram shows that there are two noticeable modes and one small mode in the right tail of the data. DIC chooses 2 components, which is similar to the Zmix procedure in~\cite{32Zoe}. Upon investigating the heights, we can argue that the number of components can reasonably be anywhere from 2 to 6. Inspecting the deep levels of the dendrogram in Figure~\ref{fig:galaxy}, we can see that there are three subpopulations of atoms varying around 10, 20, and 33, respectively. Hence, the dendrogram gives us a more detailed illustration and interpretation of mixing measures inferred from mixture models compared to choosing a single number of components to describe the heterogeneous data population.

\begin{figure}[t!]
      \centering
      \subcaptionbox*{\scriptsize (a) Histogram with kernel density estimation \par}{\includegraphics[width = 0.495\textwidth]{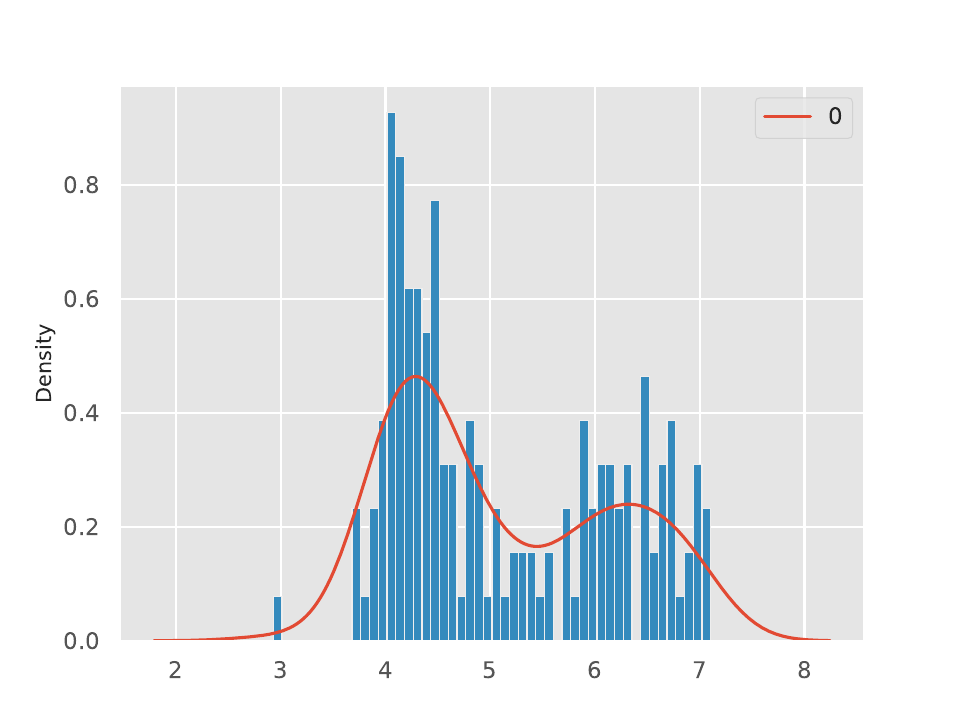}}
      \subcaptionbox*{\scriptsize (b) Dendrogram\par}{\includegraphics[width = 0.495\textwidth]{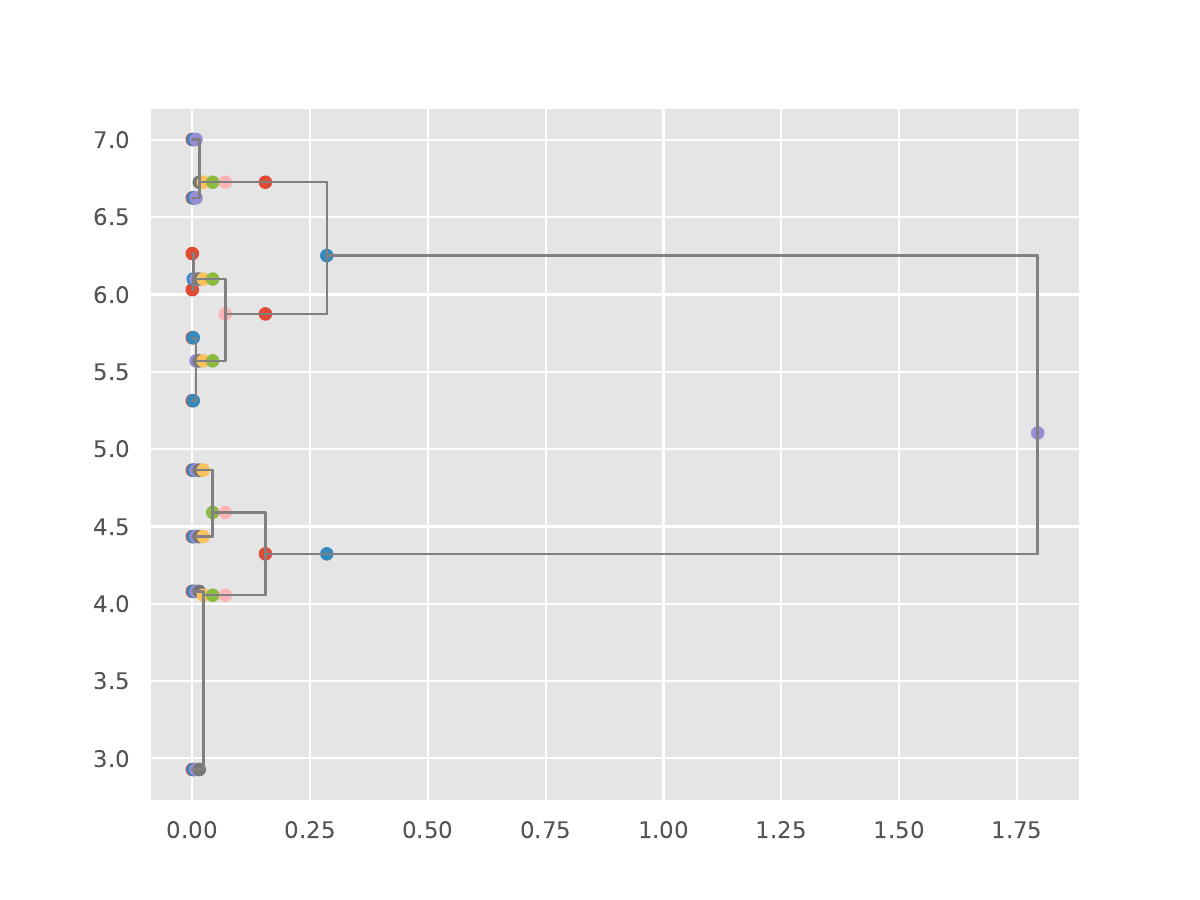}}
      \subcaptionbox*{\scriptsize (c) AIC and BIC\par}{\includegraphics[width = 0.48\textwidth]{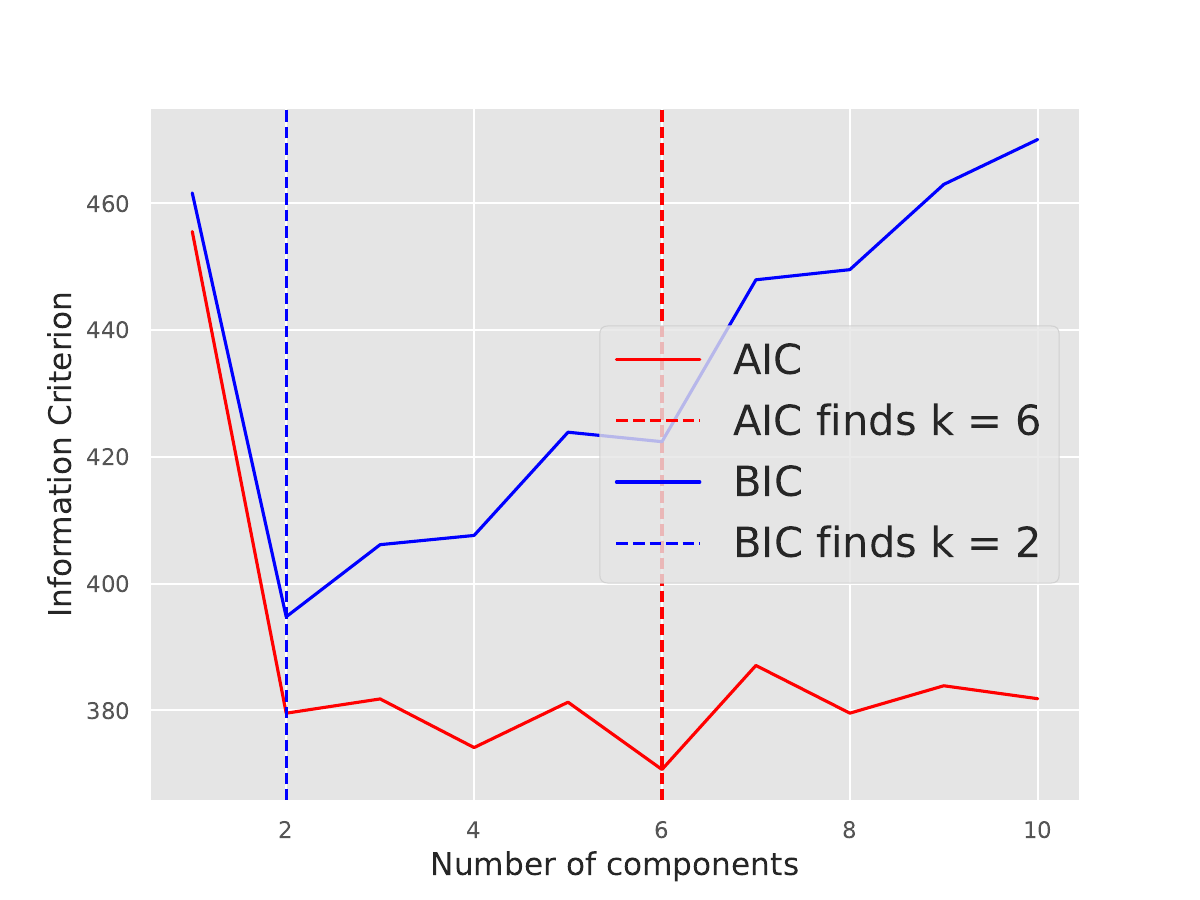}}
      \subcaptionbox*{\scriptsize (d) DIC \par}{\includegraphics[width = 0.48\textwidth]{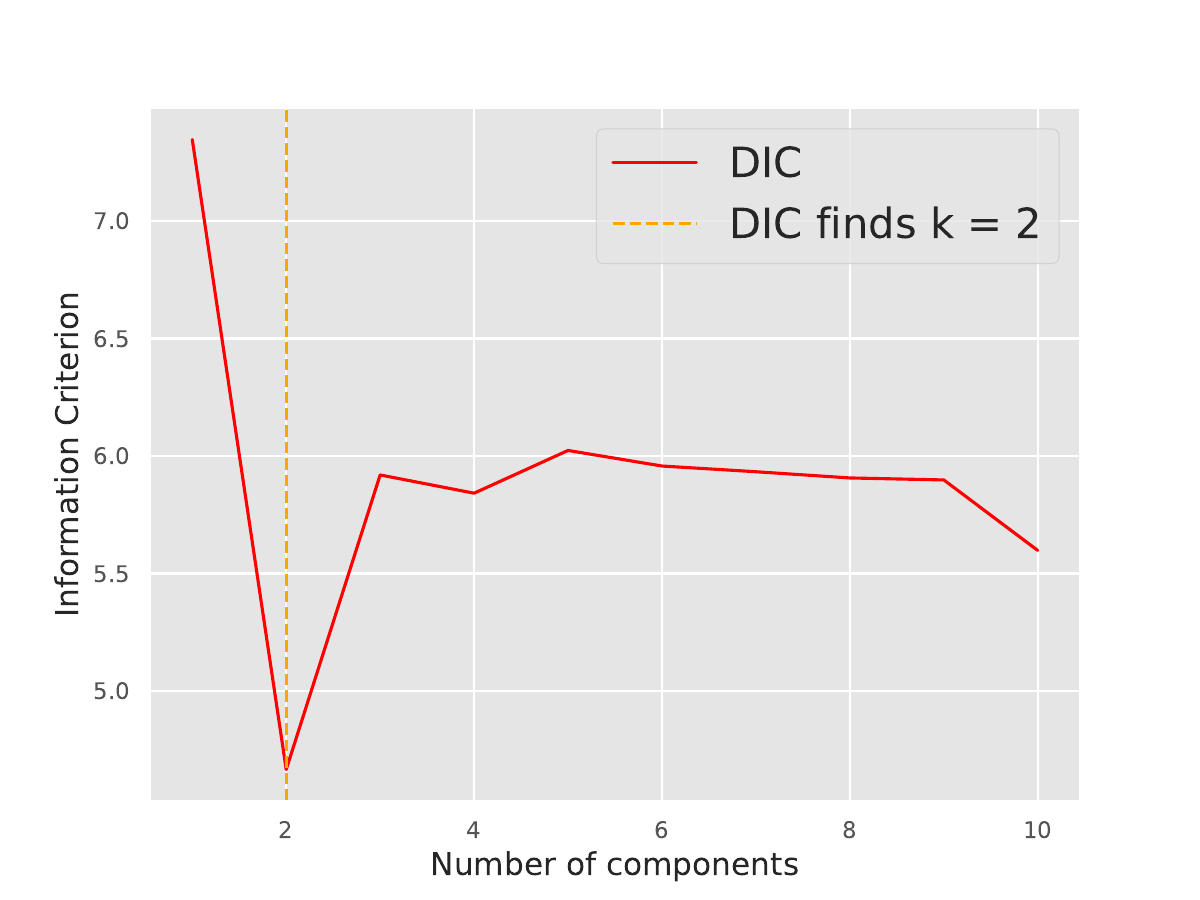}}
      \subcaptionbox*{\scriptsize (e) Heights of levels in dendrogram\par}{\includegraphics[width = 0.48\textwidth]{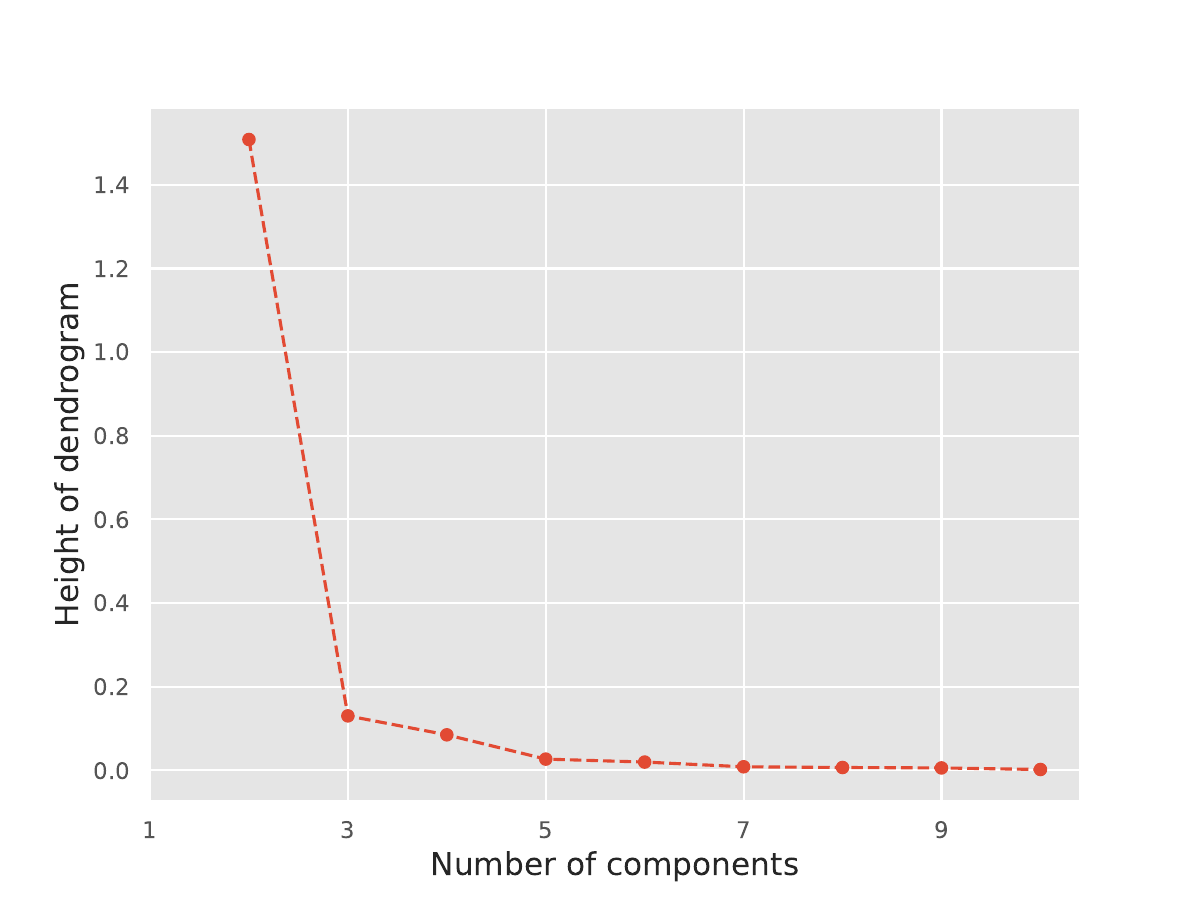}}
      \subcaptionbox*{\scriptsize (f) Likelihood of levels in dendrogram \par}{\includegraphics[width = 0.48\textwidth]{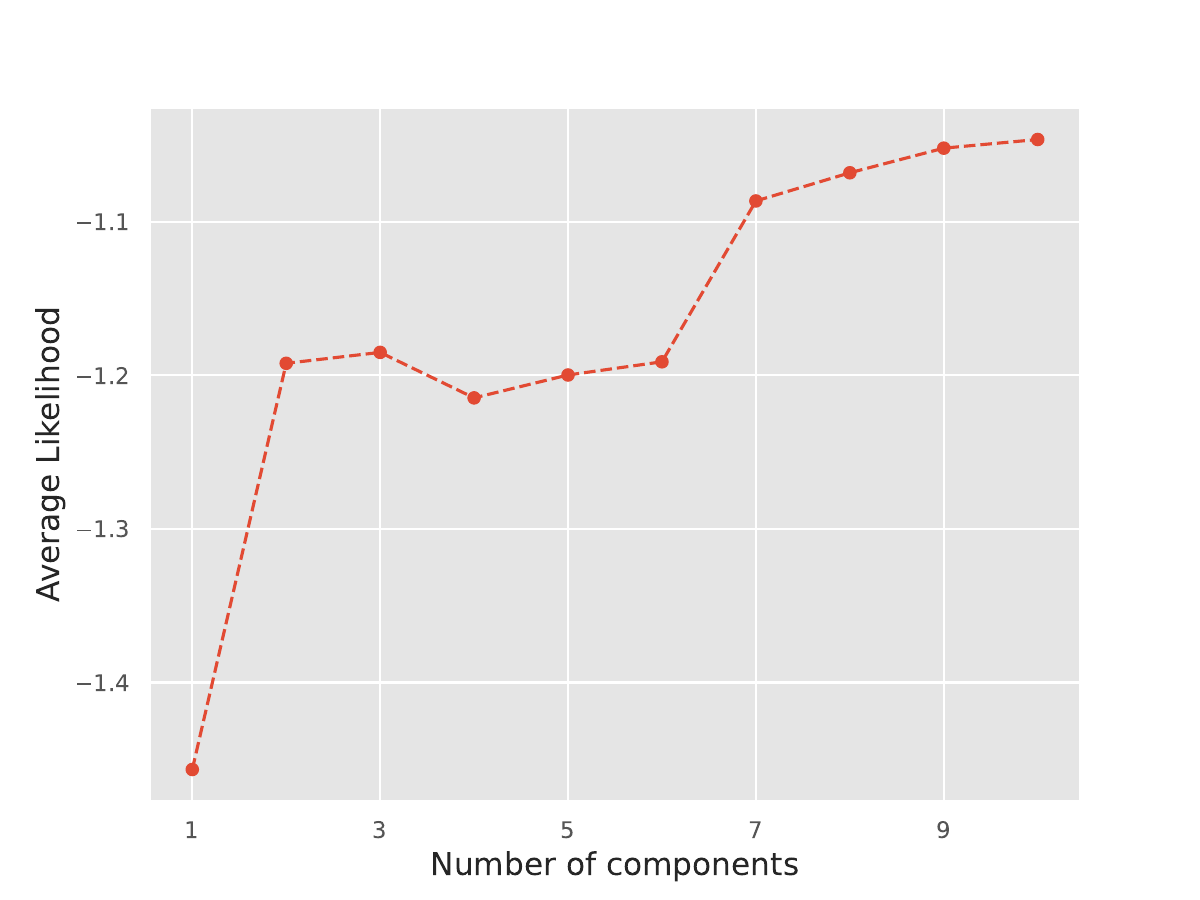}}
      \caption{\centering Dendrogram of mixing measures inferred from Acidity data}\label{fig:acidity}
\end{figure}

\begin{figure}[t!]
      \centering
      \subcaptionbox*{\scriptsize (a) Histogram with kernel density estimation \par}{\includegraphics[width = 0.495\textwidth]{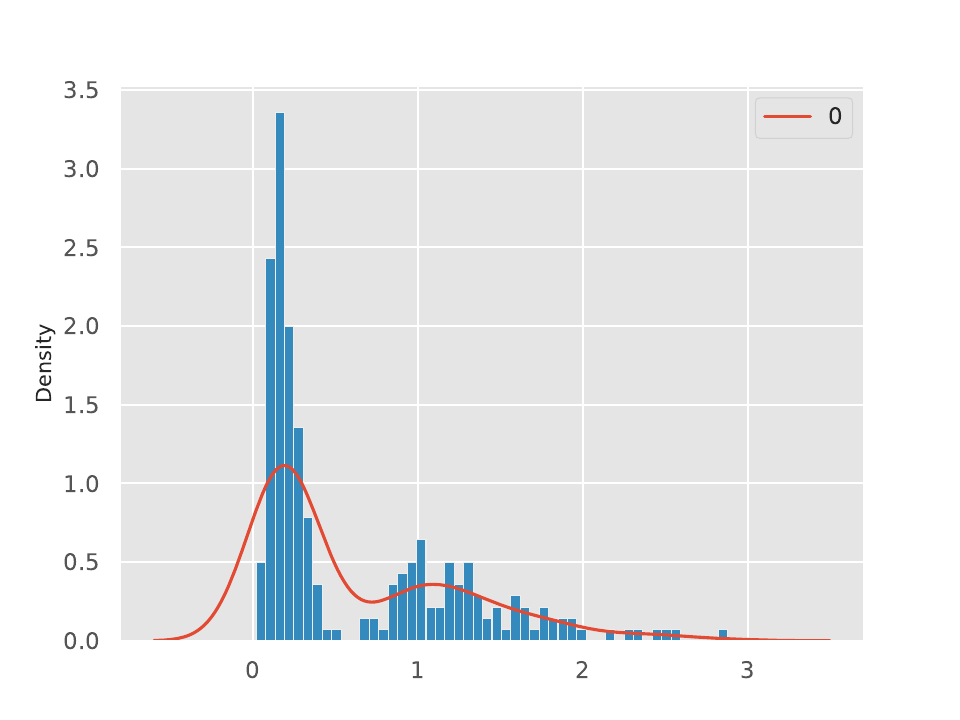}}
      \subcaptionbox*{\scriptsize (b) Dendrogram\par}{\includegraphics[width = 0.495\textwidth]{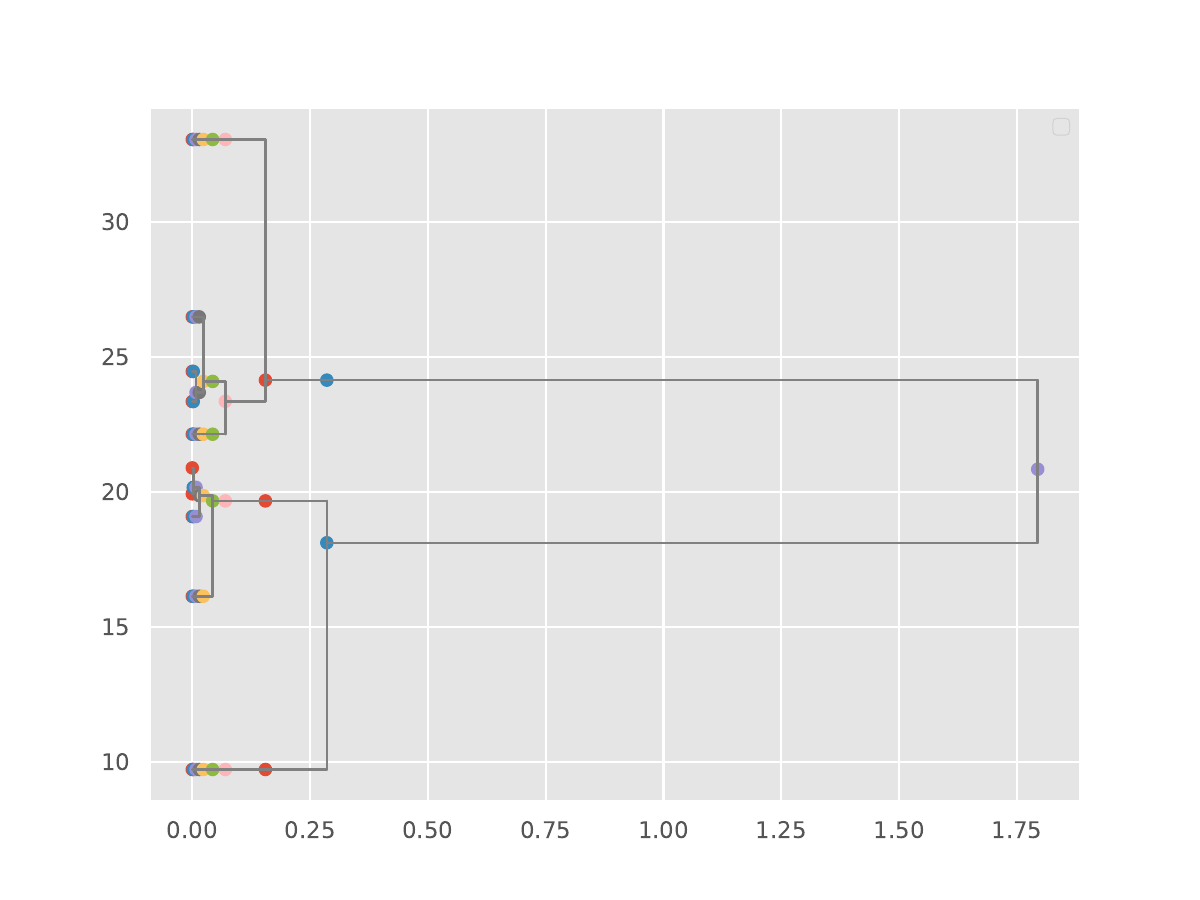}}
      \subcaptionbox*{\scriptsize (c) AIC and BIC\par}{\includegraphics[width = 0.48\textwidth]{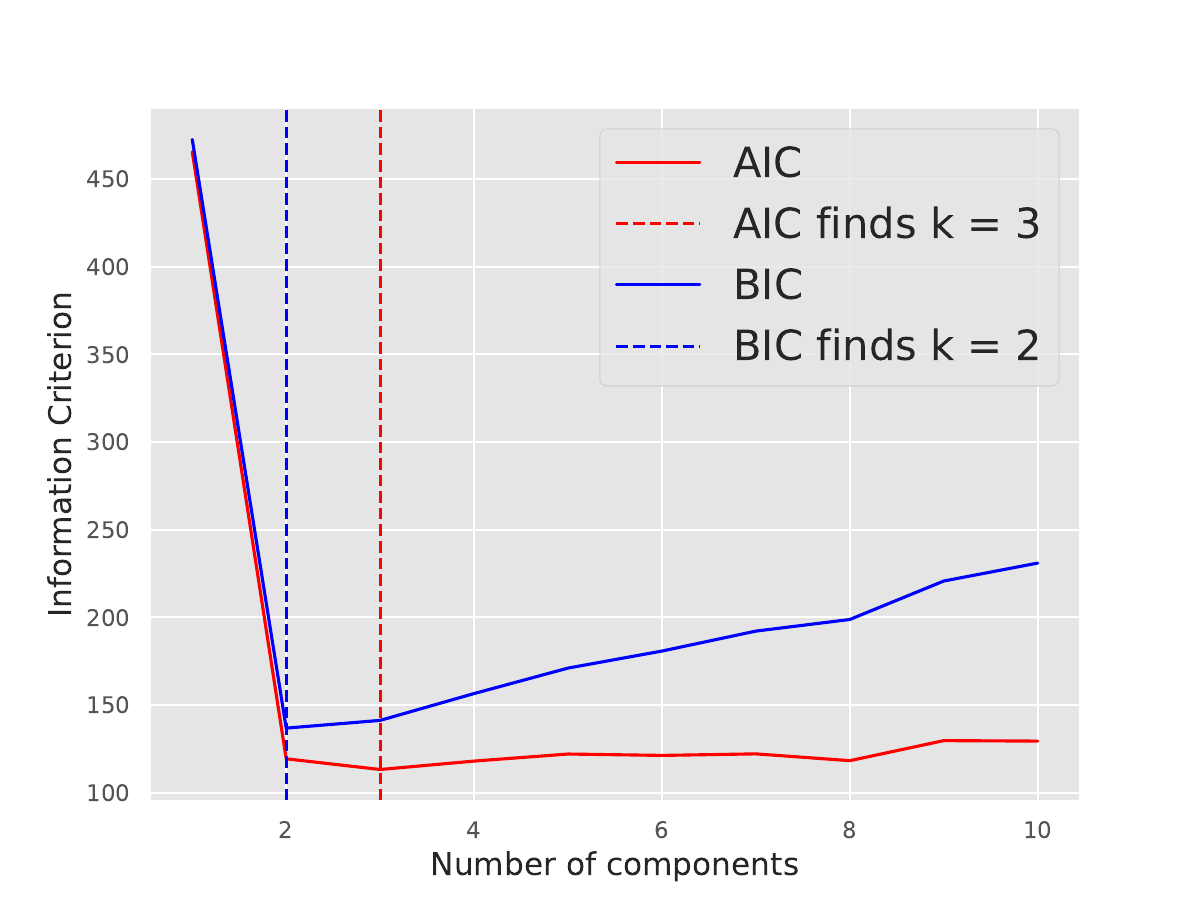}}
      \subcaptionbox*{\scriptsize (d) DIC \par}{\includegraphics[width = 0.48\textwidth]{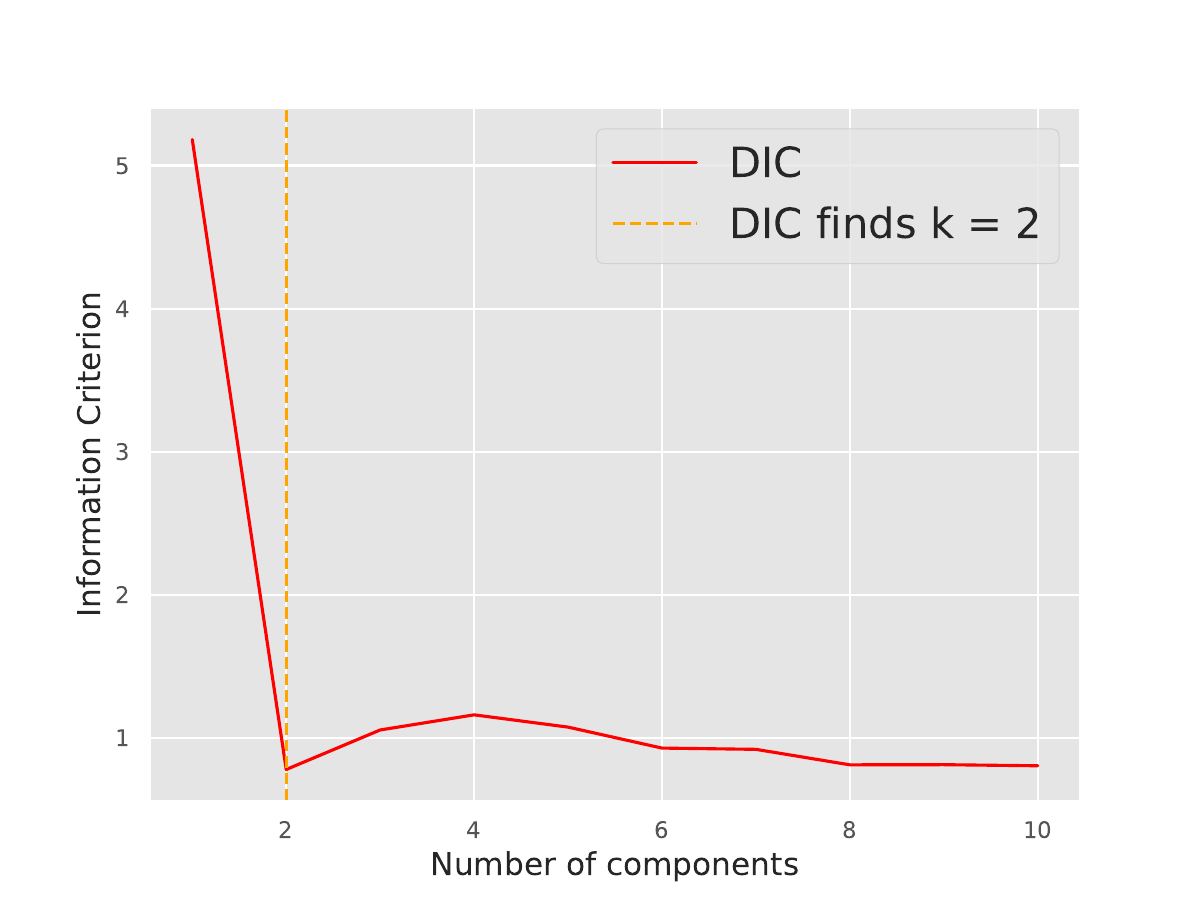}}
      \subcaptionbox*{\scriptsize (e) Heights of levels in dendrogram\par}{\includegraphics[width = 0.48\textwidth]{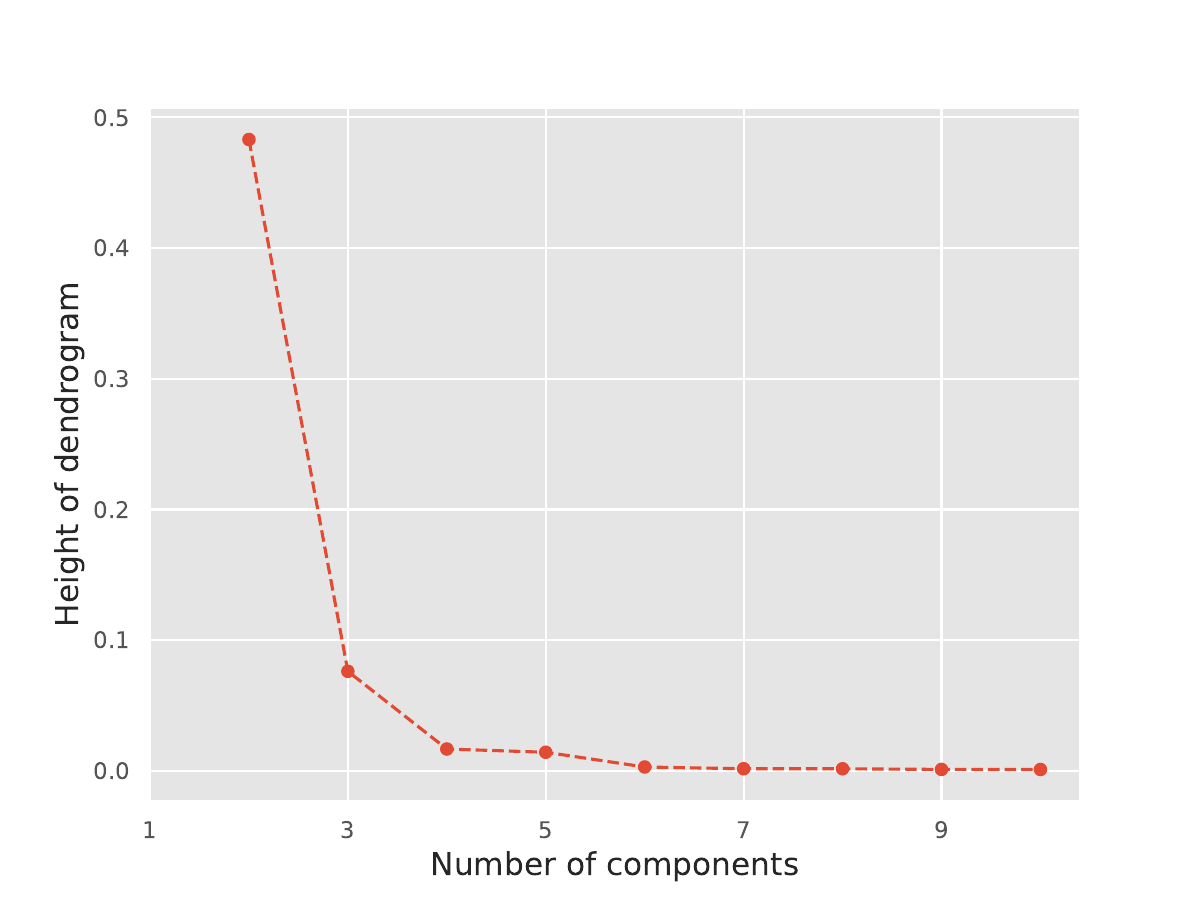}}
      \subcaptionbox*{\scriptsize (f) Likelihood of levels in dendrogram \par}{\includegraphics[width = 0.48\textwidth]{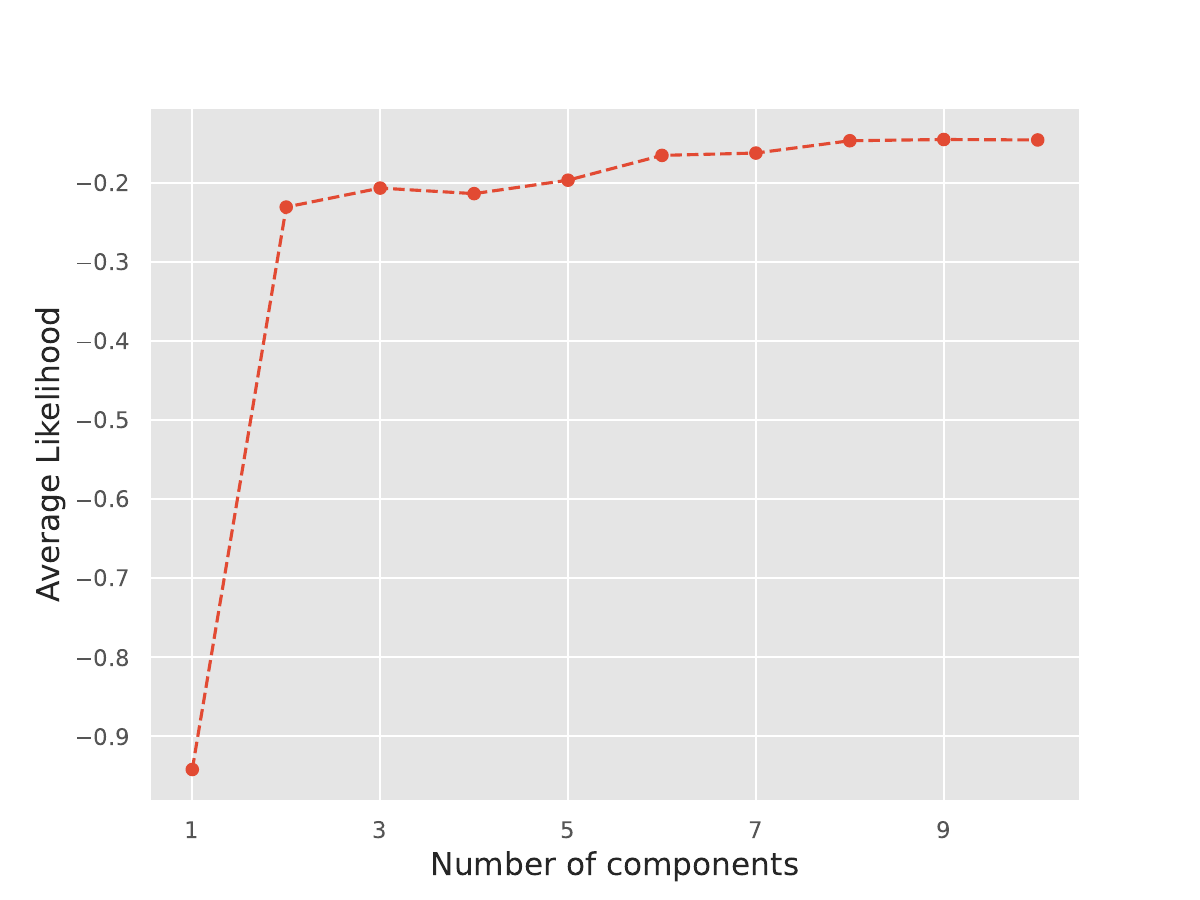}}
      \caption{\centering Dendrogram of mixing measures inferred from Enzyme data}\label{fig:enzyme}
\end{figure}

\begin{figure}[t!]
      \centering
      \subcaptionbox*{\scriptsize (a) Histogram with kernel density estimation \par}{\includegraphics[width = 0.495\textwidth]{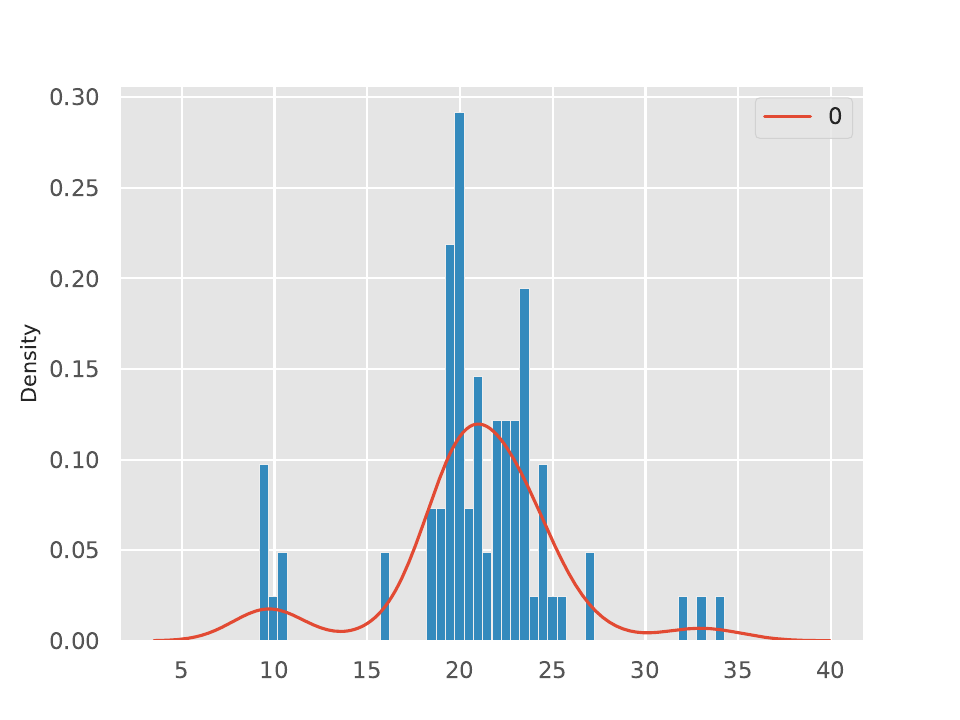}}
      \subcaptionbox*{\scriptsize (b) Dendrogram\par}{\includegraphics[width = 0.495\textwidth]{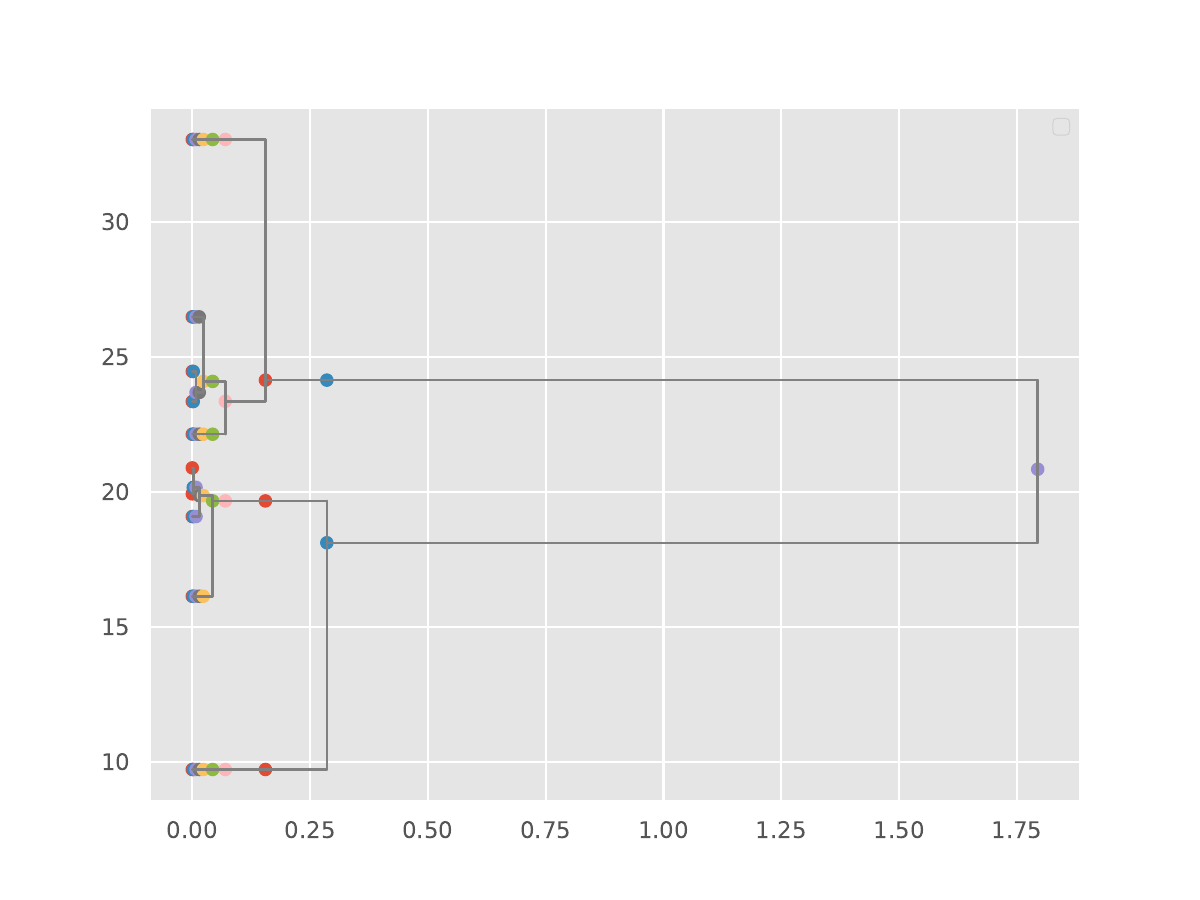}}
      \subcaptionbox*{\scriptsize (c) AIC and BIC\par}{\includegraphics[width = 0.48\textwidth]{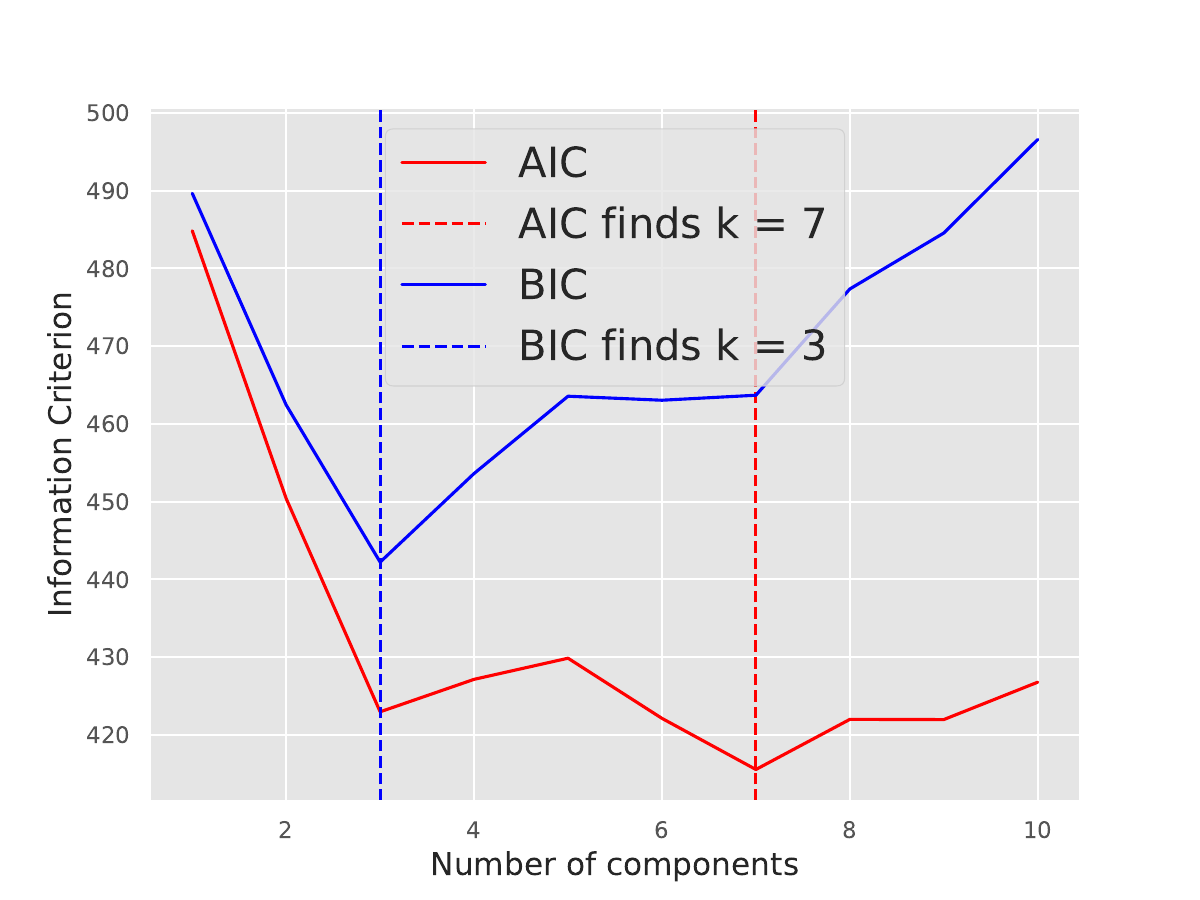}}
      \subcaptionbox*{\scriptsize (d) DIC \par}{\includegraphics[width = 0.48\textwidth]{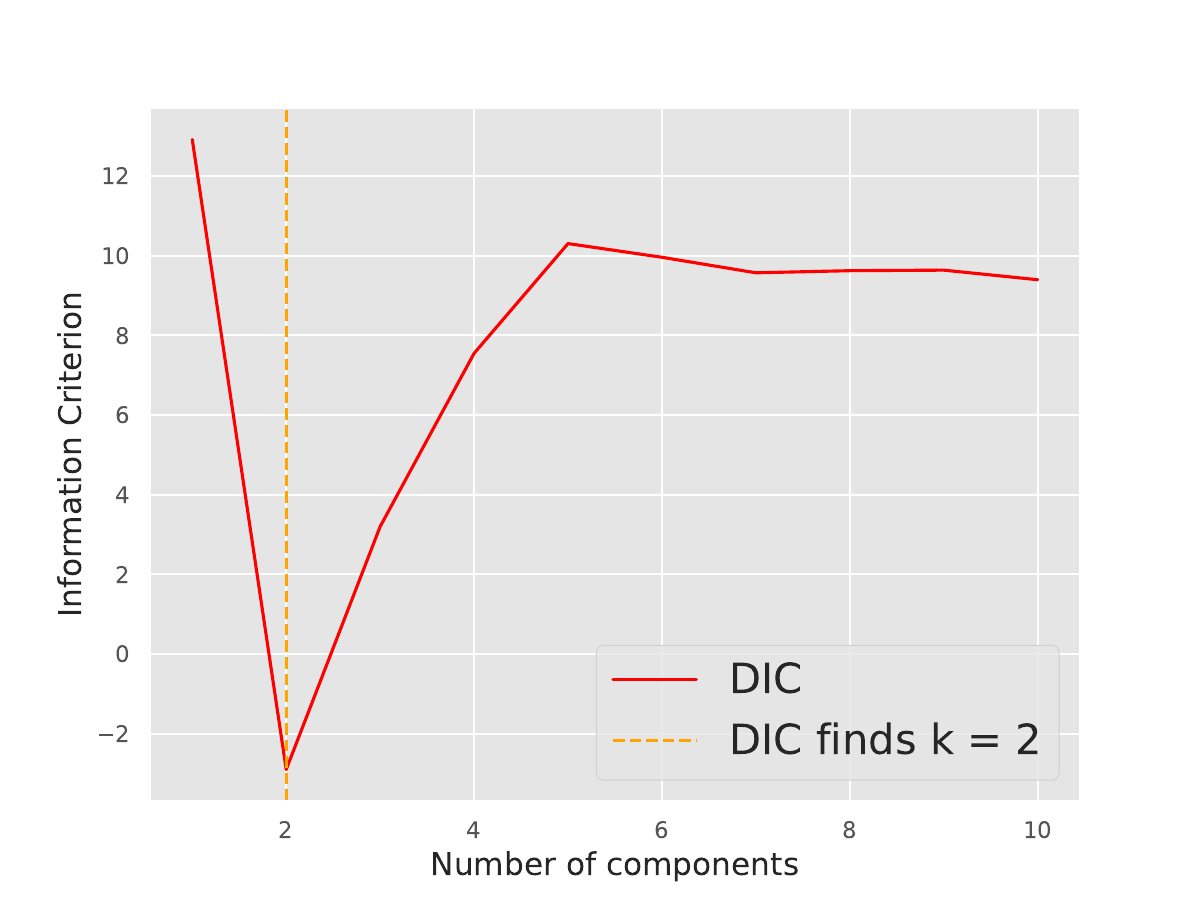}}
      \subcaptionbox*{\scriptsize (e) Heights of levels in dendrogram\par}{\includegraphics[width = 0.48\textwidth]{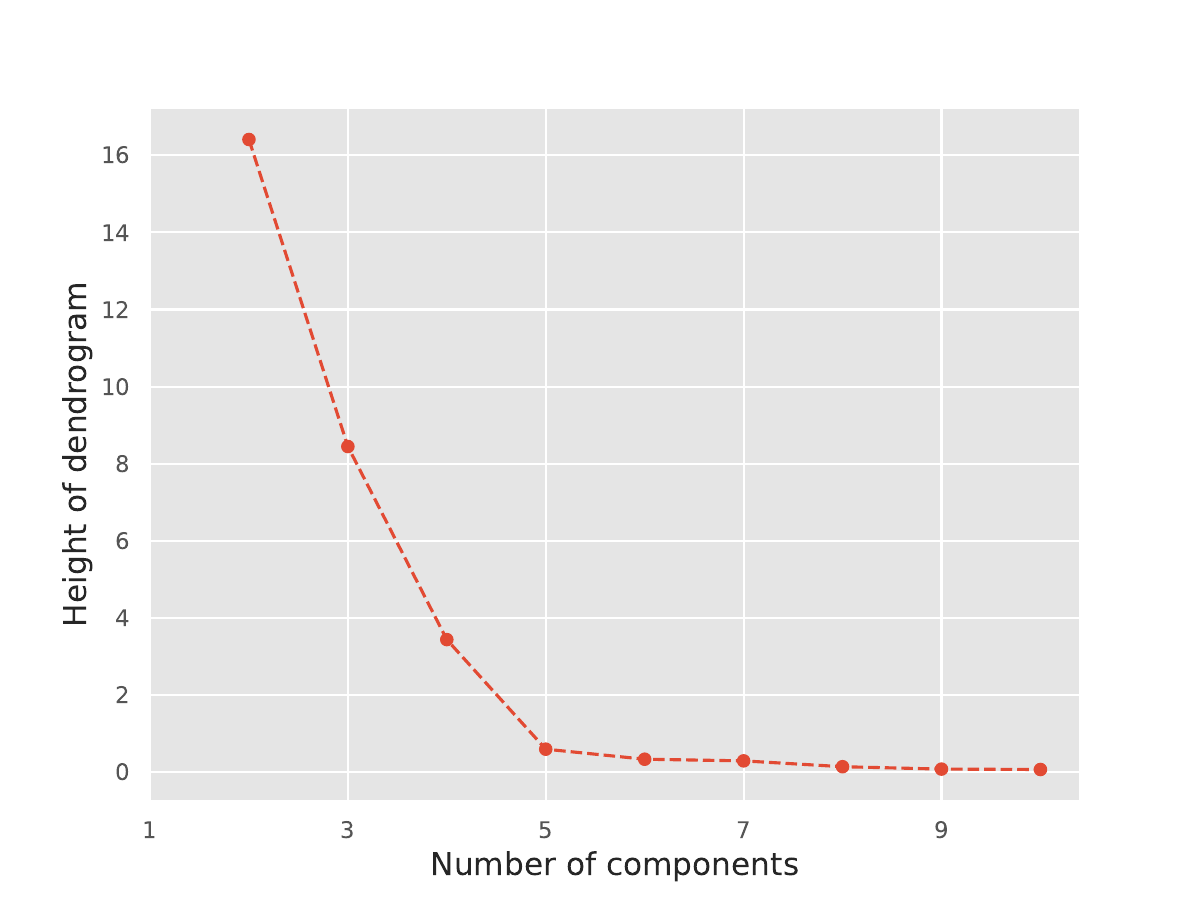}}
      \subcaptionbox*{\scriptsize (f) Likelihood of levels in dendrogram \par}{\includegraphics[width = 0.48\textwidth]{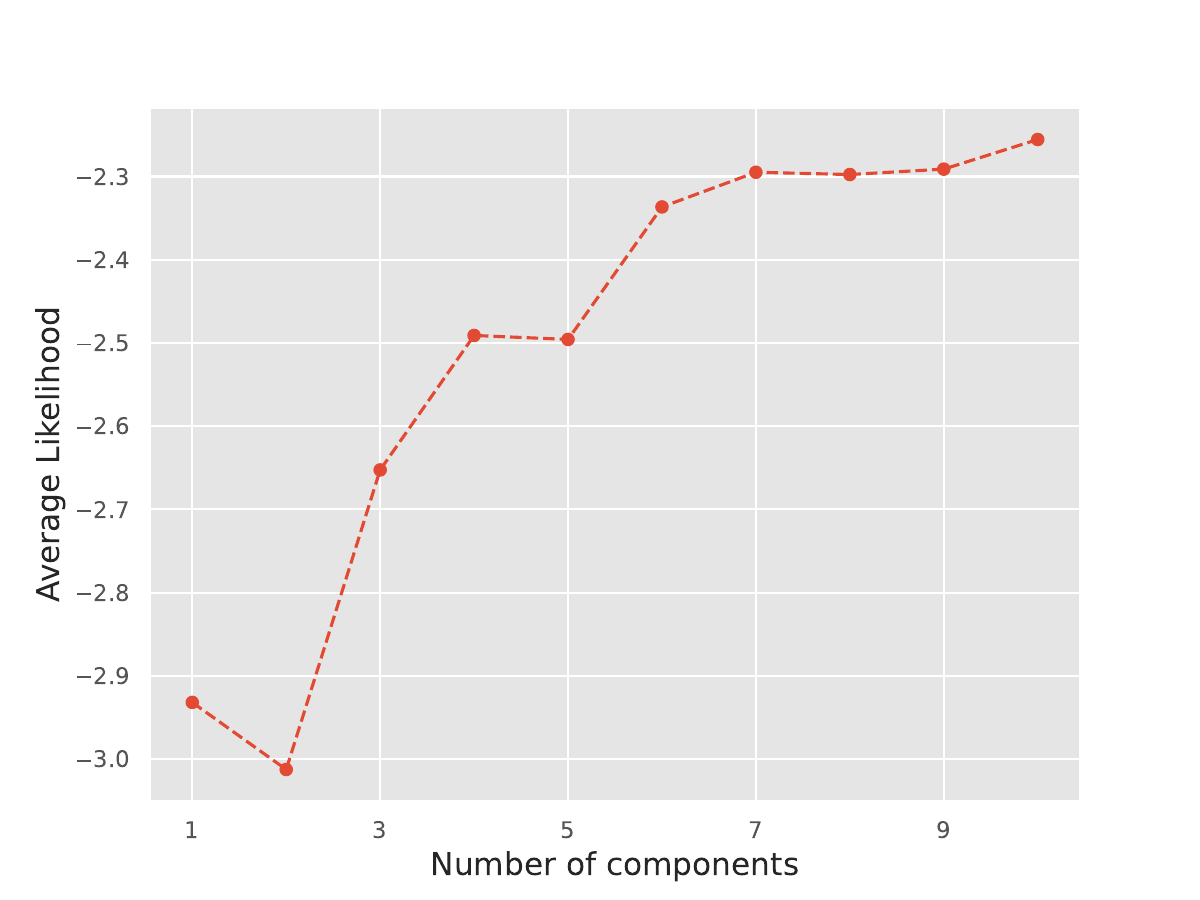}}
      \caption{\centering Dendrogram of mixing measures inferred from Galaxy data}\label{fig:galaxy}
\end{figure}

\section{Proofs of Section~\ref{sec:prelim}}
\subsection{Proof of Proposition~\ref{prop:density-rate}}
Proposition~\ref{prop:density-rate} is a direct consequence of the application of empirical process in M-estimation \cite{31Svan} (Chapter 7). We recall some main ingredients therein, then proceed to the proof of Proposition~\ref{prop:density-rate}. A similar procedure has been applied to prove the density estimation rate for mixtures (see, e.g., ~\cite{13Ho, 12Ho, 21Tudor}). 

For the true mixing measure $G_0$ and for any mixing measure $G \in \Ocal_{k}$, denote $\overline{G} = (G + G_0)/2$ and $\overline{p}_{G} = p_{\overline{G}} = (p_G + p_{G_0})/2$. Let $H_B(u, \Pcal, \nu)$ be the $u$-bracketing entropy number of some space $\Pcal$ with respect to $L^2(\nu)$, where $\nu$ is the dominating measure. Denote the empirical process for $G$:
$$\nu(G) = \sqrt{n} \int \log \dfrac{p_{G}}{p_{G_0}} d(P_n - P_{G_0}). $$
The following result plays the main role in uniformly bounding the empirical process.
\begin{theorem}[Theorem 5.11 in \cite{31Svan}]\label{thm:uniform-bound-empirical-process}
    Let positive numbers $R, C, C_1, a$ satisfy:
    \begin{equation}\label{eq:upper-bound-a}
        a \leq C_1 \sqrt{n} R^2 \wedge 8\sqrt{n} R,
    \end{equation}
    and
    \begin{equation}\label{eq:lower-bound-a}
        a \geq \sqrt{C^2(C_1 + 1)} \left(\int_{a/(2^6\sqrt{n})}^{R} H_B^{1/2} \left(\dfrac{u}{\sqrt{2}}, \left\{p_{G} : G\in \Ocal_{k}, h(p_{G}, p_{G_0})\leq R \right\}, \nu\right) du \vee R \right),
    \end{equation}
    then \begin{equation}\label{eq:uniform-bound-empirical-process}
        \Pbb_{G_0} \left(\sup_{\substack{G \in \Ocal_{k}\\ h(p_{G}, p_{G_0}) \leq R}} |\nu_n(G)| \geq a\right) \leq C \exp\left(-\dfrac{a^2}{C^2(C_1 + 1) R^2} \right).
    \end{equation}
\end{theorem}
Note that $h(p_{\overline{G}}, p_{G_0}) \leq h(p_G, p_{G_0})\leq 4 h(p_{\overline{G}}, p_{G_0})$ holds for all $G$, the conclusion of Theorem~\ref{thm:uniform-bound-empirical-process} still holds when changing the inequality constraints in~\eqref{eq:lower-bound-a} and~\eqref{eq:uniform-bound-empirical-process} for $G$ to $\overline{G}$, where all the constants only need to be adjusted by a universal constant. The advantage of working with $\overline{G}$ instead of $G$ is that $\dfrac{p_{\overline{G}}}{p_{G_0}}$ is always bounded below by $1/2$, so that $\nu(\overline{G})$ does not blow up for all $p_G$. We further define
$$\overline{\Pcal}_{k}^{1/2}(\Theta, \epsilon) = \{\overline{p}_G^{1/2} : G\in \Ocal_{k}, h(\overline{p}_{G}, p_{G_0})\leq \epsilon \},$$
and the bracketing entropy integral:
$$\mathcal{J}(\epsilon, \overline{P}^{1/2}(\Theta, \epsilon), \nu) = \int_{\epsilon^2 / 2^{13}}^{\epsilon} H_B^{1/2}(u, \overline{\Pcal}^{1/2}(\Theta, \epsilon), \nu) du \vee \epsilon.$$
\begin{theorem}[Theorem 7.4 in \cite{31Svan}]\label{thm:geer-density-rate}
    Take $\Psi(\epsilon) \geq \mathcal{J}(\epsilon, \overline{P}^{1/2}(\Theta, \epsilon), \nu)$ such that $\Psi(\epsilon) / \epsilon^2$ is a non-increasing function of $\epsilon$. Then for a universal constant $c$, and for 
    $$c \Psi(\epsilon_n) \leq \sqrt{n} \epsilon_n^2, $$
    we have that for all $\epsilon \geq \epsilon_n$
    $$\Pbb_{G_0} (h(p_{\widehat{G}_n}, p_{G_0})\geq \epsilon) \leq c \exp(n\epsilon^2/c^2). $$
\end{theorem}

\begin{proof}[Proof of Proposition~\ref{prop:density-rate}] 
Firstly, we notice that
\begin{align*}
    H_B(u, \overline{\Pcal}_k^{1/2}(\Theta, u), \nu) & \leq H_B(u, \overline{\Pcal}_k^{1/2}(\Theta), \nu)\\
    & = H_B(u / \sqrt{2}, \overline{\Pcal}_k(\Theta, u), h)\\
    & \leq H_B(u, {\Pcal}_k(\Theta, u), h),
\end{align*}
where the first inequality is obvious, the equality comes from the identity between $L^2(\nu)$ and Hellinger distance, and the second inequality comes from the fact that $h(\overline{p}_{G}, \overline{p}_{G'})\leq h(p_G, p_{G'}) / \sqrt{2}$. Hence, from the condition (\textbf{B}.), we have
\begin{align*}
    \mathcal{J}(\epsilon, \overline{\Pcal}_k^{1/2}(\Theta, \epsilon), \nu) & \leq \int_{\epsilon^2/2^{13}}^{\epsilon} (J\log(1/u))^{1/2} du\\
    &\leq J^{1/2} \epsilon \left(\log\dfrac{1}{ (\epsilon^2/2^{13})}\right)^{1/2}\\
    &\leq J' \epsilon (\log(1/\epsilon))^{1/2},
\end{align*}
for all $\epsilon$ small enough, for some constant $J'$ only depends on $\Theta$ and $k$. Hence, for $\Psi(\epsilon) = J' \epsilon (\log(1/\epsilon))^{1/2}$, we have $\Psi(\epsilon)$ is a non-increasing function. Let $\epsilon_n = \max\{ 1, c J'\} (\log n / n)^{1/2}$, we have
$$c \Psi(\epsilon_n) \leq c J'\epsilon_n (\log (1/\epsilon_n))^{1/2} \leq \epsilon_n \times (cJ'(\log(n))^{1/2})\leq  \epsilon_n^2 \sqrt{n}.$$
Substitute $\epsilon = \epsilon_n$ to the conclusion of Theorem~\ref{thm:geer-density-rate}, we have 
\begin{align*}    \Pbb_{G_0}\left(h(p_{\widehat{G}_n}, p_{G_0}) \geq \max\{ 1, c J'\} \left(\dfrac{\log n}{n}\right)^{1/2}\right)
&\leq c_1 \exp(-c_2\max\{ 1, c J'\}\log(n))\\
&\leq c_1 n^{-c_2},
\end{align*}
where $J'$ depends on $\Theta$ and $k$ only, and $c_1$ and $c_2$ are universal constants.
\end{proof}

\subsection{Proof of Proposition~\ref{prop:Bk-verify}}
\begin{proof}[Proof of Proposition~\ref{prop:Bk-verify}]
We aim to show that
\begin{equation}\label{eq:entropy-hellinger-bound}
    H_B (\epsilon, \Pcal_k, h) \lesssim \log(1/\epsilon).
\end{equation}
To do it, we first prove that
\begin{equation}\label{eq:covering-supnorm-bound}
    \log N (\epsilon, \Pcal_{k}, \norm{\cdot}_{\infty})  \lesssim \log(1/\epsilon),
\end{equation}
then show that~\eqref{eq:covering-supnorm-bound} implies~\eqref{eq:entropy-hellinger-bound}.

\paragraph{Proof of bound~\eqref{eq:covering-supnorm-bound}} Because the parameter space $\Theta$ is a compact set of $\Rbb^{d}$, we can choose an $\epsilon$-net $\mathcal{S}_1$ for it with the cardinality no more than $O\left(\dfrac{1}{\epsilon^{d}}\right)$. Besides, because of the space of weights $p_i$'s is $\Delta^{k-1}$, it also possible to choose an $\epsilon$-net $\mathcal{S}_2$ with cardinality no more than $O(1/\epsilon^{k-1})$. Let $\mathcal{S} = \mathcal{S}_1^{k} \times \mathcal{S}_2$. We have $\log |\mathcal{S}|\lesssim \log(1/\epsilon)$ and for every $G = \sum_{i=1}^{k} p_i \delta_{\theta_i}$, there exists $G' = \sum_{i=1}^{k} p_i' \delta_{\theta_i'}$, for $\theta_1', \dots, \theta_k'\in \mathcal{S}_1$ and $(p_1', \dots, p_k')\in \mathcal{S}_2$ such that $|p_i - p_i'|$ and $\norm{\theta_i - \theta_i'}\leq \epsilon$ for all $i\in [k]$. By triangle inequalities,
$$\norm{p_G - p_{G'}}_{\infty} \leq \sum_{i=1}^{k} |p_i - p_i'| \norm{f}_{\infty} + p_i \sup_{x} |f(x|\theta_i) - f(x|\theta_i')|\lesssim \epsilon,$$
thanks to the uniform bounded and Lipchitz with respect to sup norm assumptions. Hence,
\begin{equation*}
    \log N(\epsilon, \mathcal{P}_{k}, \norm{\cdot}_{\infty}) \lesssim \log(1/\epsilon).
\end{equation*} 

\paragraph{Proof of claim~\eqref{eq:entropy-hellinger-bound}} Now, from the entropy number with respect to the supremum norm, we are going to bound the bracketing number with respect to Hellinger distance. Let $\eta \leq \epsilon$, which will be chosen later. Let $p_{G_1}, \dots, p_{G_N}$ be an $\eta$-net for $\Pcal_{k}$, and
\begin{equation}
    H(x) = \begin{cases}
    d_1\exp(-d_2\norm{x}^{d_3}), & \quad \norm{x} \geq D,\\
    \sup_{\theta} \norm{f(\cdot|\theta)}_{\infty} , &\quad \text{otherwise}
    \end{cases}
\end{equation}
be an envelop for $f(x|\theta)$. We can construct brackets $[p_i^L, p_i^U]$ as follows. 
\begin{align*}
p_i^L(x) & = \max\{p_{G_i}(x) - \eta, 0 \},\\
p_i^U(x) & = \min\{p_{G_i}(x) + \eta, H(x) \}.
\end{align*}
Because for each $p\in \Pcal_{k}$, there is $p_{G_i}$ such that $\norm{p_i - p_{G_i}}_{\infty} < \eta$, we have $ p_i^{L}\leq p\leq p_i^{U}$. Moreover, for any $\overline{D}\geq D$,
\begin{align}
    \int_{\mathbb{R}^{d}} (p_i^{U} - p_i^{L}) dx & \leq \int_{\norm{x} \leq \overline{D}} 2\eta dx + \int_{\norm{x} \geq \overline{D}} H(x) dx \nonumber \\
    &\lesssim \eta \overline{D}^{d} + \overline{D}^{d}\exp\left( - b_2\overline{D}^{d_3}\right)\label{bracket-radi},
\end{align}
where we use spherical coordinates to have
\begin{equation*}
	 \int_{\norm{x} \leq \overline{D}}  dx = \dfrac{\pi^{d/2}}{\Gamma(d/2 + 1)} \overline{D}^d \lesssim \overline{D}^d,
\end{equation*}
and
\begin{align*}
	\int_{\norm{x} \geq \overline{D}}\exp\left(-d_2\norm{x}^{d_3} \right) & \lesssim \int_{r \geq \overline{D}} r^{d-1} \exp \left(-d_2 r^{d_3} \right) dr \\
	& = \dfrac{1}{d_3 d_2^{1/d_3}}\int_{\overline{D}^{d_3}}^{\infty} u^{d/d_3 - 1} \exp(-u) du \quad (\text{with } u = d_2 r^{d_3})\\
	& \leq \dfrac{1}{d_3 d_2^{1/d_3}} \overline{B}^{d-d_3} \exp(-\overline{D}^{d_3}).
\end{align*} 
Hence, in \eqref{bracket-radi}, choosing $\overline{D} = D(\log(1/\eta))^{1/d_3}$ gives
\begin{equation}\label{bracket-bound}
	\int_{\mathbb{R}^{d}} (p_i^{U} - p_i^{L}) dx \lesssim \eta \left(\log\left(\dfrac{1}{ \eta}\right)\right)^{d/d_3}.
\end{equation}
Therefore, there exists a positive constant $c$ which does not depend on $\eta$ such that 
\begin{equation*}
    H_B(c\eta \log(1/\eta)^{d/d_3}, \mathcal{P}_k, \norm{\cdot}_1)
    \lesssim \log(1/\eta).
\end{equation*}
Let $\epsilon = c \eta (\log(1/\eta))^{d/d_3}$, we have $\log(1/\epsilon) \asymp \log(1/\eta)$, which combines with inequality $\norm{\cdot}_1\leq h^2$ leads to 
\begin{equation*}
    H_B(\epsilon, \mathcal{P}_k, h)\leq H_B(\epsilon^2, \mathcal{P}_k, \norm{\cdot}_1)\lesssim \log(1/\epsilon^2) \lesssim \log(1/\epsilon).
\end{equation*}
Thus, bound~\eqref{eq:entropy-hellinger-bound} is proved.
\end{proof}

\section{Proofs of Section~\ref{sec:dendrogram-strong}}
\subsection{Proof of Proposition~\ref{prop:Wasserstein-variational}: Variational characterization of the dendrogram}
\begin{proof}[Proof of Proposition~\ref{prop:Wasserstein-variational}] The proof is organized into several small steps.

\paragraph{Step 1: Representation of Wasserstein projection.}
For $G^{(k)} = \sum_{i=1}^{k} p_i \delta_{\theta_i} \in \Ecal_{k}$, let $G^{*} = \sum_{j=1}^{k-1} p_j^* \delta_{\theta_j^*} \in \Ocal_{k-1}$ satisfy
    \begin{equation}\label{eq:G-k-1-proof-variational}
        G^{*} = \arg\inf_{G\in \Ocal_{k-1}} W_2^2(G, G^{(k)}).
    \end{equation}
    We proceed to show that $G^{*}$ has the representation described in Algorithm~\ref{alg:merge-atom}, i.e., $G^* = G^{(k-1)}$. Firstly, because $\Theta$ is assumed to be compact, we also have $\Ocal_{k-1}$ being compact. So there exists a minimizer $G^*$ to the infimum problem~\eqref{eq:G-k-1-proof-variational}. Recall that 
    \begin{equation*}
        W_2^2(G^{*}, G^{(k)}) = \min_{q \in \Pi(p, p^*)} \sum_{i, j=1}^{k, k-1} q_{ij} \norm{\theta_i - \theta_j^*}^2,
    \end{equation*}
    where $\Pi(p, p^*) = \{q\in \Rbb_{\geq 0}^{k \times (k-1)} : \sum_{i=1}^{k} q_{ij} = p_j^*, \sum_{j=1}^{k-1} q_{ij} = p_i\}$. 
    Therefore, the minimization problem~\eqref{eq:G-k-1-proof-variational} is equivalent to:
    \begin{equation*}
        \min_{\theta_1^*, \dots, \theta_{k-1}^*\in \Theta} \min_{q\in \Pi(p)}  \sum_{i, j=1}^{k, k-1} q_{ij} \norm{\theta_i - \theta_j^*}^2,
    \end{equation*}
    where $\Pi(p) = \{q\in \Rbb_{\geq 0}^{k \times (k-1)} : \sum_{j=1}^{k-1} q_{ij} = p_i\forall i \in [k]\}$. 
    
    \paragraph{Step 2. Optimal weights of $G^*$ given its atoms:}
    Now, consider any fixed $\theta_1^*, \dots, \theta_{k-1}^*$, for any $i\in [k]$, denote by $\tau(i) \in [k-1]$ the index satisfying:
    $$\tau(i) = \argmin_{j \in [k-1]} \norm{\theta_i - \theta_{j}^*}.$$
    In plain words, $\theta_{\tau(i)}^*$ is the atom $\theta_j^*$ nearest to $\theta_i$. If there is more than one such index, we can just pick one. Because
    \begin{equation*}
        \sum_{i, j=1}^{k, k-1} q_{ij} \norm{\theta_i - \theta_j^*}^2 \geq \sum_{i, j=1}^{k, k-1} q_{ij} \norm{\theta_i - \theta_{\tau(i)}^*}^2 = \sum_{i=1}^{k} p_i \norm{\theta_i - \theta_{\tau(i)}^*}^2,
    \end{equation*}
    we have that the optimal coupling $q_{ij}$ must put all mass on the pairs $(\theta_i, \theta_{\tau(i)}^{*})$ for $i\in [k]$, i.e., 
    \begin{equation*}
        q_{ij} = \begin{cases}
            p_i & \text{if } j = \tau(i) \\
            0 & \text{otherwise}.
        \end{cases}
    \end{equation*}
    Hence, the minimization problem~\eqref{eq:G-k-1-proof-variational} becomes:
    \begin{equation}\label{eq:G-k-1-proof-variational-2}
        \min_{\theta_1^*, \dots, \theta_{k-1}^*\in \Theta} \sum_{i=1}^{k} p_{i} \norm{\theta_i - \theta_{\tau(i)}^*}^2.
    \end{equation}
    \paragraph{Step 3. Optimal atoms of $G^*$:}
    Now, we claim that $\theta_1^*, \dots, \theta_{k-1}^*\in \Theta$ must be chosen so that $|\tau([k])| = k-1$. Because otherwise, $|\tau([k])|\leq k-2$, and there will have at least one index $j^*\in [k-1]$ such that $\theta_{j^*}^{*}$ does not appear in the minimization of~\eqref{eq:G-k-1-proof-variational-2}, and at least two indices $i_1, i_2\in [k]$ such that $\tau(i_1) = \tau(i_2)$. Because $\theta_{i_1} \neq \theta_{i_2}$, at least one of them must be different from $\theta_{\tau(i_1)}^*$. Assume that $\theta_{i_1} \neq \theta_{\tau(i_1)}^*$, then we can choose $\theta_{j^*}^* = \theta_{i_1}$ and further reduce the objective of~\eqref{eq:G-k-1-proof-variational-2}, which is a contradiction. Thus, the claim is proved.

    Hence, we showed that $G^{*}$ has exactly $k-1$ atoms. Because  $|\tau([k])| = k-1$, we can assume that there are two indices $i\neq i'$ such that $\tau(i) = \tau(i')$, while $\tau(1),\dots, \tau(i-1), \tau(i+1), \dots, \tau(k)$ are distinct. The problem~\eqref{eq:G-k-1-proof-variational-2} becomes:
    \begin{equation}\label{eq:G-k-1-proof-variational-3}
        \min_{\theta_1^*, \dots, \theta_{k-1}^*\in \Theta} \left(\sum_{j\neq i, i'} p_{j} \norm{\theta_j - \theta_{\tau(j)}^*}^2 + p_{i} \norm{\theta_{i} - \theta_{\tau(i)}^*}^2 + p_{i'} \norm{\theta_{i'} - \theta_{\tau(i)}^*}^2\right).
    \end{equation}
    Solving this problem yields $\theta_{\tau(j)}^* = \theta_{j}$ for all $j\neq i, i'$ and 
    $$\theta_{\tau(i)}^* = \dfrac{p_i}{p_i + p_{i'}} \theta_{i} +  \dfrac{p_{i'}}{p_i + p_{i'}} \theta_{i'}.$$
    Hence, $G^{*} = \sum_{j\neq i, i'} p_j \delta_{\theta_j} + \overline{p} \delta_{\overline{\theta}}$, where $\overline{p} = p_i + p_{i'}$ and $\overline{\theta} = \theta_{\tau(i)}^*$ as defined above. Finally, to find the optimal $i, i'\in [k]$, we notice that with this choice of $G^{*}$, 
    \begin{equation*}
        W_2^2(G^{*}, G^{(k)}) = p_i \norm{\theta_i - \overline{\theta}}^2 + p_{i'} \norm{\theta_{i'} - \overline{\theta}}^2 = \dfrac{1}{p_i^{-1} + p_{i'}^{-1}} \norm{\theta_i - \theta_{i'}}^2.
    \end{equation*}
    Therefore, $i$ and $i'$ are chosen to minimize $\divclus(p_i \delta_{\theta_i}, p_{i'} \delta_{\theta_{i'}})$, which is exactly the merging rule in Algorithm~\ref{alg:merge-atom}. Hence, the output of the Algorithm~\ref{alg:merge-atom} satisfies the variational characterization~\eqref{eq:G-k-1-proof-variational}.
\end{proof}

\subsection{Proof of Lemma~\ref{lem:inv-bound}: Inverse bound for strongly identifiable mixtures}
The proof proceeds by using contradictions, an usual technique for proving inverse bounds~\cite{3Chen, 23Nguyen, 21Tudor}. By identifiability, a sequence of 
overfitted mixing measures $G_n = \sum p_i^n \delta_{\theta_i^n} \to G_0$ can have (i) redundant atoms $\theta_i^n, \theta_j^n$ that converges to the same $\theta_i^0$, and (ii) some weights $p_j^n$ vanishing fast so $\theta_j^n$ can vary anywhere. Note that for any varying sequence $(\theta_i^{n})_{n=1}^{\infty}$, because $\Theta$ is compact, we can extract a subsequence that is convergent. In the following proof, we characterize the presentation of the sequence $G_n$ to capture both scenarios. 

The main novelty of this proof compared to existing work is that we fully utilize the second-order strong identifiability condition, which concern all partial derivatives up to the second order, while other proofs only use the linear independence of the zero and second-order derivatives of $f$. This technique yields a stronger inverse bound.

\begin{proof}[Proof of Lemma~\ref{lem:inv-bound}] The proof is divided into several small steps.

\paragraph{Step 1: Proving by contradiction and setup.} Suppose that the claim of the lemma is not true. Then there exists a sequence $G_n \in \Ocal_{k}(\Theta)$ such that $V(p_{G_n}, p_{G_0})\to 0$ and $\dfrac{V(p_{G_n}, p_{G_0})}{\divergence(G_n, G_0)} \to 0$. Because $\Theta$ is compact, by extracting a subsequence if needed, we can assume that $G_n \to G_{*} \in \Ocal_{k}(\Theta)$. Hence,
$$V(p_{G_*}, p_{G_0}) = \lim_{n\to \infty} V(p_{G_n}, p_{G_0}) = 0.$$
By the identifiability of $f$, we have that $G_* = G_0$. Thus $\lim_{n\to \infty} G_n = G_0$ in $W_1$. Therefore, we also have $\divergence(G_n, G_0) \to 0$ as $n\to \infty$. Without loss of generality, we can assume that the sequence $G_n$ has the following representation: 
$$G_n = \sum_{i=1}^{k_0}\sum_{j=1}^{s_i} p_{ij}^{n} \delta_{\theta_{ij}^{n}} + \sum_{i=1}^{k_0}\sum_{j=1}^{\tilde{s}_i}\sum_{t=1}^{\tilde{r}_{ij}} \tilde{p}_{ijt}^{n} \delta_{\tilde{\theta}_{ijt}^{n}} \in \Ecal_{k_*} \quad \forall n\in \Nbb,$$
where
$$\sum_{j=1}^{s_i} p_{ij}^{n}\to p_{i}^{0}, \quad \theta_{ij}^{n}\to \theta_i^0,\quad \forall i\in [k_0], $$
and
$$\tilde{p}^{n}_{ijt} \to 0, \quad \tilde{\theta}^{n}_{ijt} \to \tilde{\theta}_{ij}^0 \quad \forall t\in [\tilde{r}_{ij}], j\in [\tilde{s}_i], i\in [k_0],$$
where $\sum_{i=1}^{k_0} s_i + \sum_{i,j} \tilde{r}_{ij}\leq k$, $(\theta_i^0, \tilde{\theta}_{ij}^0)_{i, j}$ are distinct, and $\theta_{ij}^0\in V_i$ for all $j\in [\tilde{s}_i], i\in [k_0]$. Hence, for all $n$ large enough, we have
\begin{align*}
    \divergence(G_n, G_0) & = \sum_{i=1}^{k_0} \left|\sum_{j=1}^{s_i} p_{ij}^{n} + \sum_{j, t} \tilde{p}_{ijt}^{n} - p_i^0 \right| + \left\|\sum_{j=1}^{s_i} p_{ij}^{n} (\theta_{ij}^{n} -\theta_i^0) + \sum_{j, t} \tilde{p}_{ijt}^{n} (\tilde{\theta}^{n}_{ijt} - \theta_i^0) \right\| \\ 
    & + \sum_{j=1}^{s_i} p_{ij}^{n} \norm{\theta_{ij}^{n} - \theta_i^0}^2 + \sum_{j=1}^{\tilde{s}_i} \sum_{t=1}^{\tilde{r}_{ij}} \tilde{p}_{ijt}^{n} \norm{\tilde{\theta}_{ijt}^{n} - \theta_i^0}^2\\
    &\hspace{-.7cm} \lesssim \sum_{i=1}^{k_0} \left|\sum_{j=1}^{s_i} p_{ij}^{n} - p_i^0 \right| + \left\|\sum_{j=1}^{s_i} p_{ij}^{n} (\theta_{ij}^{n} -\theta_i^0)\right\| + \sum_{j=1}^{s_i} p_{ij}^{n} \norm{\theta_{ij}^{n} - \theta_i^0}^2 + \sum_{j=1}^{\tilde{s}_i} \sum_{t=1}^{\tilde{r}_{ij}} \tilde{p}_{ijt}^{n},
\end{align*}
where the multiplicative constant in this inequality only depends on $\Theta$, by the application of triangle inequalities.

\paragraph{Step 2: Taylor expansion.} We have the Taylor expansion of $f(x | \theta)$ around $\theta_i^0$:
\begin{align*}
    f(x|\theta_{ij}^{n}) - f(x|\theta_i^0) & = (\theta_{ij}^{n} - \theta_i^0)^{\top} \dfrac{\partial f(x | \theta_{i}^0)}{\partial \theta} + (\theta_{ij}^{n} - \theta_i^0)^{\top} \dfrac{\partial^2 f(x | \theta_{i}^0)}{\partial \theta^2} (\theta_{ij}^{n} - \theta_i^0) + R^{n}_{ij}(x),
\end{align*}
where $R^{n}_{ij}(x) = o\left(\norm{\theta_{ij}^{n} - \theta_i^0}^2\right)$ for all $x$. Therefore,
\begin{align*}
    & p_{G_n}(x) - p_{G_0}(x)  = \sum_{i=1}^{k_0} \sum_{j=1}^{s_i} p_{ij}^{n} f(x|\theta_{ij}^n) - \sum_{i=1}^{k_0} p_{i}^{0} f(x|\theta_{i}^0) + \sum_{i, t, j} \tilde{p}^{n}_{ijt} f(x|\tilde{\theta}_{ijt}^{n}) \\
    & = \sum_{i=1}^{k_0} \left(\sum_{j=1}^{s_i} p_{ij}^{n} - p_i^0 \right)f(x|\theta_i^0) + \sum_{i, t, j} \tilde{p}^{n}_{ijt} f(x|\tilde{\theta}_{ijt}^{n}) + \sum_{i=1}^{k_0} \sum_{j=1}^{s_i} p_{ij}^{n} \left( f(x|\theta_{ij}^n) - f(x|\theta_i^0)\right) \\
    & = \sum_{i=1}^{k_0} \left(\sum_{j=1}^{s_i} p_{ij}^{n} - p_i^0 \right)f(x|\theta_i^0) + \sum_{i, t, j} \tilde{p}^{n}_{ijt} f(x|\tilde{\theta}_{ijt}^{n}) + R^{n}(x)\\
    & + \sum_{i=1}^{k_0}\left(\sum_{j=1}^{s_i} p_{ij}^{n}(\theta_{ij}^{n} - \theta_i^0)\right)^{\top} \dfrac{\partial f(x | \theta_{i}^0)}{\partial \theta} + \sum_{i=1}^{k_0} \sum_{j=1}^{s_i} p_{ij}^{n} (\theta_{ij}^{n} - \theta_i^0)^{\top} \dfrac{\partial^2 f(x | \theta_{i}^0)}{\partial \theta^2} (\theta_{ij}^{n} - \theta_i^0) ,
\end{align*}
where $R^{n}(x) = o\left(\sum_{i=1}^{k_0}\sum_{j=1}^{s_i} p_{ij}^{n}\norm{\theta_{ij}^{n} - \theta_i^0}^2\right) = o(D_n)$ for all $x$, where $D_n := \divergence(G_n, G_0)$.

\paragraph{Step 3: Non-vanishing coefficients.} For each $i \in [k_0]$ and $u, v \in [d]$ (recall that $d$ is the dimension of $\theta$), let $a^{n}_i$ be the coefficient of $f(x|\theta_i^0)$, $b^{n}_{iu}$ the coefficient of $\partial f(x|\theta_i^0)/ \partial \theta^{(u)}$, and $c^{n}_{iuv}$ the coefficient of $\partial^2 f(x|\theta_i^0)/ \partial \theta^{(u)} \partial \theta^{(v)}$ in the display above, where $\theta^{(u)}$ is generally denoted the $u$-th dimension of $\theta$. We have 
\begin{align*}
    \sum_{i=1}^{k_0} \sum_{r=1}^{d} \left( |a_i| + |b_{iu}| + |c_{iuu}|\right) + \sum_{i, j, t } \tilde{p}_{ijt}^{n}\\
    \geq  \sum_{i=1}^{k_0} \left(\left|\sum_{j=1}^{s_i} p_{ij}^n - p_{i}^0 \right| + \left\|\sum_{j=1}^{s_i} p_{ij}^{n} (\theta_{ij}^n - \theta_i^0) \right\| + \sum_{j=1}^{s_i} p_{ij}^{n} \norm{\theta_{ij}^{n} - \theta_i^0}^2\right)  + \sum_{i, j, t } \tilde{p}_{ijt}^{n} \gtrsim D_n.
\end{align*}
Hence, when denote $M_n = \max_{i, u, v, j, t} \{|a^{n}_i|, |b^{n}_{iu}|, |c^{n}_{i u v}|, |\tilde{p}_{ijt}^{n}|
\}$, we have $M_n \gtrsim D_n$, which implies $R^{n} = o(M_n)$, and $\dfrac{V(p_{G_n}, p_{G_0})}{M_n} = 0$. Moreover, because $\dfrac{a_i^{n}}{M_n}$ is a sequence in a compact set $[-1, 1]$, we can WLOG assume that $\dfrac{a_i^{n}}{M_n} \to \alpha_i \in [-1, 1]$. Similarly, 
$$\dfrac{b^{n}_{iu}}{M_n} \to \beta_{iu}\in [-1, 1], \quad \dfrac{c^{n}_{i u v}}{M_n} \to \gamma_{iu v} \in [-1, 1], \quad \dfrac{\tilde{p}_{ijt}^{n}}{M_n} \to \pi_{ijt}\in [0, 1]\quad \forall i, u, v, j, t.$$
Besides, at least a coefficient in $(\alpha_i, \beta_{iu}, \gamma_{iuv}, \pi_{ijt})_{i, u, v, j, t}$ is 1. We also note that $\gamma_{i u v} = \gamma_{i v u}$ as $c_{iuv}^{n} = c_{ivu}^{n}$ for all $n$.

\paragraph{Step 4: Derive contradiction using the strong identifiability condition.}
According to Fatou's lemma,
\begin{align*}
    & \lim_{n\to \infty}\dfrac{2V(p_{G_n}, p_{G_0})}{M_n} = \lim_{n\to \infty} \int_{\mathcal{X}} \left|\dfrac{p_{G_n}(x) - p_{G_0}(x)}{M_n} \right| dx \\
    & \geq \int_{\mathcal{X}} \left|\lim_{n\to \infty} \dfrac{p_{G_n}(x) - p_{G_0}(x)}{M_n} \right| dx\\
    & \vspace{-1cm} = \int_{\mathcal{X}} \left|\sum_{i=1}^{k_0} \alpha_i f(x|\theta_i^0) + \sum_{j, t} \pi_{ijt} f(x|\theta_{ij}^0) + \sum_u \beta_{iu} \dfrac{f(x|\theta_i^0)}{\partial \theta^{(u)}} + \sum_{u, v} \gamma_{iuv} \dfrac{\partial^2 f(x|\theta_i^0)}{\partial \theta^{(u)} \partial \theta^{(v)}} \right| dx.
\end{align*}
This implies:
$$\sum_{i=1}^{k_0} \alpha_i f(x|\theta_i^0) + \sum_{j, t} \pi_{ijt} f(x|\theta_{ij}^0) + \sum_{u} \beta_{iu} \dfrac{f(x|\theta_i^0)}{\partial \theta^{(u)}} + \sum_{u, v} \gamma_{iuv} \dfrac{\partial^2 f(x|\theta_i^0)}{\partial \theta^{(u)} \partial \theta^{(v)}} = 0, \,\,\text{a.s. } x\in \mathcal{X},$$
which is contradictory to the strong identifiability condition, given that at least a coefficient in $(\alpha_i, \beta_{iu}, \gamma_{iuv}, \sum_{t} \pi_{ijt})_{i, u, v, j}$ is greater or equal to 1. Hence, the inverse bound is correct.
\end{proof}

\subsection{Proof of Lemma~\ref{lem:order-mix-measure}}
\begin{proof}[Proof of Lemma~\ref{lem:order-mix-measure}] We only need to show that for all $k \geq k_0 + 1$ and sequence $G_n \in \Ecal_{k}$ such that $\divergence(G_n, G_0)\to 0$, then 
\begin{equation}\label{eq:1-step-order-mix-measure}
   \divergence(G_n, G_0) \gtrsim \divergence(G_n^{(k-1)}, G_0). 
\end{equation}
The rest follows from the induction argument. To show inequality~\eqref{eq:1-step-order-mix-measure}, we use the usual technique of proving by contradiction, which is similar to Lemma~\ref{lem:inv-bound}. Assume that~\eqref{eq:1-step-order-mix-measure} is not correct, we have a sequence of $G_n \in \Ecal_{k}$ such that
\begin{equation}\label{eq:order-mix-measure-contradiction}
    \divergence(G_n, G_0)\to 0 \quad \text{and} \quad \dfrac{\divergence(G_n, G_0)}{\divergence(G_n^{(k-1)}, G_0)}\to 0,
\end{equation}
as $n\to \infty$. Without loss of generality, we can assume that the sequence $G_n$ has the following representation: 
$$G_n = \sum_{i=1}^{k_0}\sum_{j=1}^{s_i} p_{ij}^{n} \delta_{\theta_{ij}^{n}} + \sum_{i=1}^{k_0}\sum_{j=1}^{\tilde{s}_i}\sum_{t=1}^{\tilde{r}_{ij}} \tilde{p}_{ijt}^{n} \delta_{\tilde{\theta}_{ijt}^{n}} \in \Ecal_{k_*} \quad \forall n\in \Nbb,$$
where
$$\sum_{j=1}^{s_i} p_{ij}^{n}\to p_{i}^{0}, \quad \theta_{ij}^{n}\to \theta_i^0,\quad \forall i\in [k_0], $$
and
$$\tilde{p}^{n}_{ijt} \to 0, \quad \tilde{\theta}^{n}_{ijt} \to \tilde{\theta}_{ij}^0 \quad \forall t\in [\tilde{r}_{ij}], j\in [\tilde{s}_i], i\in [k_0],$$
where $\sum_{i=1}^{k_0} s_i + \sum_{i,j} \tilde{r}_{ij}\leq k$, $(\theta_i^0, \tilde{\theta}_{ij}^0)_{i, j}$ are distinct, and $\theta_{ij}^0\in V_i$ for all $j\in [\tilde{s}_i], i\in [k_0]$. For all $n$ large enough, we have
\begin{align*}
    \divergence(G_n, G_0) & = \sum_{i=1}^{k_0} \left|\sum_{j=1}^{s_i} p_{ij}^{n} + \sum_{j, t} \tilde{p}_{ijt}^{n} - p_i^0 \right| + \left\|\sum_{j=1}^{s_i} p_{ij}^{n} (\theta_{ij}^{n} -\theta_i^0) + \sum_{j, t} \tilde{p}_{ijt}^{n} (\tilde{\theta}^{n}_{ijt} - \theta_i^0) \right\| \\ 
    & + \sum_{j=1}^{s_i} p_{ij}^{n} \norm{\theta_{ij}^{n} - \theta_i^0}^2 + \sum_{j=1}^{\tilde{s}_i} \sum_{t=1}^{\tilde{r}_{ij}} \tilde{p}_{ijt}^{n} \norm{\tilde{\theta}_{ijt}^{n} - \theta_i^0}^2\\
    &\hspace{-.7cm} \lesssim \sum_{i=1}^{k_0} \left|\sum_{j=1}^{s_i} p_{ij}^{n} - p_i^0 \right| + \left\|\sum_{j=1}^{s_i} p_{ij}^{n} (\theta_{ij}^{n} -\theta_i^0)\right\| + \sum_{j=1}^{s_i} p_{ij}^{n} \norm{\theta_{ij}^{n} - \theta_i^0}^2 + \sum_{j=1}^{\tilde{s}_i} \sum_{t=1}^{\tilde{r}_{ij}} \tilde{p}_{ijt}^{n},
\end{align*}
where the multiplicative constant in this inequality only depends on $\Theta$, by the application of triangle inequalities. Besides, we also see that $\divergence(G_n, G_0) \gtrsim \tilde{p}_{ijt}^{n}$ for all $i, j, t$ due to the fact that $\tilde{\theta}_{ijt}^{n} \to \tilde{\theta}_{ij}^0 \neq \theta_i^0.$ Now consider different cases of merging atoms in $G_n$. Because there are a finite number of options to merge, by extracting a subsequence if needed, we can assume one of those cases happens for all $n\in \Nbb$. Denote $\overline{p}^{n} \delta_{\overline{\theta}^{n}}$ by the merged atom.

\paragraph{Case 1: Merge $p_{uv}^{n} \delta_{\theta_{uv}^{n}}$ and $p_{u v'}^{n} \delta_{\theta_{u v'}^{n}}$ for some $v\neq v'$ and common $u$.} Because of the convexity of $V_{u}$, we have $\overline{\theta}^{n}$ also belongs to $V_{u}$. Notice that $\overline{p}^{n} = p_{u v}^{n} + p_{u v'}^{n}$ and $\overline{p}^{n} \overline{\theta}^{n} = p_{u v}^{n} \theta_{u v}^{n} + p_{u v'}^{n} \theta_{u v'}^{n}$, we have
\begin{align*}
    \divergence(G_n, G_0) - \divergence(G_n^{(k-1)}, G_0) & = {p}_{u v}^{n} \norm{{\theta}_{u v}^{n} - \theta_i^0}^2 + {p}_{u v'}^{n} \norm{{\theta}_{u v'}^{n} - \theta_i^0}^2 - \overline{p}^{n} \norm{\overline{\theta}^{n} - \theta_i^0}^2\\
    &\geq 0,
\end{align*}
where the inequality follows by the convexity of $\norm{\cdot}^2$. Hence, 
$$\dfrac{\divergence(G_n, G_0)}{\divergence(G_n^{(k-1)}, G_0)} \geq 1 \not\to 0, $$
as $n\to \infty$, which is contradictory to~\eqref{eq:order-mix-measure-contradiction}.

\paragraph{Case 2: Merge $p_{u v}^{n} \delta_{\theta_{u v}^{n}}$ and $p_{u' v'}^{n} \delta_{\theta_{u' v'}^{n}}$ for $u\neq u'$.} Firstly, because $k \geq k_0 + 1$, we are still in the overfitted regime. Therefore, $\divclus(p_{u v}^{n} \delta_{\theta_{u v}^{n}}, p_{u' v'}^{n} \delta_{\theta_{u' v'}^{n}}) \to 0$ as $n\to\infty$. Consider three smaller cases.

\paragraph{Case 2.1: $s_u = 1$ and $s_v = 1$.} We have $p_{uv}^{n}\to p_u^0 > 0, p_{u'v'}^{n}\to p_{u'}^0 $ and $\norm{\theta_{uv}^n - \theta_{u'v'}^{n}}\to \norm{\theta_u^0 - \theta_v^0} > 0$ as $n\to \infty$. This implies
$$\dfrac{1}{(p_{uv}^{n})^{-1} + (p_{u'v'}^{n})^{-1}} \norm{\theta_{uv}^{n} - \theta_{u'v'}^{n}}^2 \not \to 0,$$
as $n \to \infty$, which is a contradiction.

\paragraph{Case 2.2: $s_u = 1, s_v > 1$ or $s_v = 1, s_u > 1$.} WLOG, we assume $s_u > 1$ and $s_{u'} = 1$. In this case, $p_{u'v'}^{n}\to p_{u'}^0 > 0$ and $\norm{\theta_{uv}^n - \theta_{u'v'}^{n}}\to \norm{\theta_u^0 - \theta_{u'}^0} > 0$ as $n\to \infty$. This implies
$$\dfrac{1}{(p_{uv}^{n})^{-1} + (p_{u'v'}^{n})^{-1}} \norm{\theta_{uv}^{n} - \theta_{u'v'}^{n}}^2 \asymp p^{n}_{uv},$$
as $n \to \infty$. However, as $s_u > 1$, there exists another atom $\theta_{uv''}^{n} \to \theta_{u}^0$. Hence,
$$\dfrac{1}{(p_{uv}^{n})^{-1} + (p_{u v''}^{n})^{-1}} \norm{\theta_{uv}^{n} - \theta_{u v''}^{n}}^2 \ll p_{uv}^{n}.$$
This also leads to a contradiction.

\paragraph{Case 2.3: $s_u > 1$ and $s_v > 1$.} WLOG, we can assume that $p_{uv}^{n} < p_{u'v'}^{n}$. Similar to the previous case, we also have
$$\dfrac{1}{(p_{uv}^{n})^{-1} + (p_{u'v'}^{n})^{-1}} \norm{\theta_{uv}^{n} - \theta_{u'v'}^{n}}^2 \asymp p^{n}_{uv}.$$
Then, we can use the same argument as the previous case to derive the contradiction. Hence, in other words, Case 2 eventually does not happen.

\paragraph{Case 3: Merge $p_{uv}^{n} \delta_{\theta_{uv}^{n}}$ with $p_{u v' t}^{n} \delta_{\theta_{uv't}^{n}}$.} The contradiction proceeds similarly to Case 1 because both atoms belong to the same Voronoi cell $V_u$.

\paragraph{Case 4: Merge $p_{uv}^{n} \delta_{\theta_{uv}^{n}}$ with $p_{u' v' t}^{n} \delta_{\theta_{u'v't}^{n}}$.} The contradiction proceeds similarly to Case 2, where we can show that there is always a better pair of atoms to merge.
\end{proof}

\subsection{Proof of Theorem~\ref{thm:asymptotic-dendrogram}: Asymptotic behavior of the mixing measures in the dendrogram}
\begin{proof}[Proof of Theorem~\ref{thm:asymptotic-dendrogram}]
    We divide the proof into two parts: overfitted levels and under-fitted levels.
    
    \paragraph{Part 1: Convergence rate on overfitted levels.}
    From Proposition~\ref{prop:density-rate}, there exists a constant $c$ depending on $\Theta$ and $k$ so that on an event $A_n$ with probability of at least $1- c_1 n^{-c_2}$, we have
    \begin{equation*}
        V\left(p_{\widehat{G}_n}, p_{G_0}\right) \leq \sqrt{2} h\left(p_{\widehat{G}_n}, p_{G_0}\right)\leq c \left(\dfrac{\log n}{n}\right)^{1/2}.
    \end{equation*}
    Besides, from Lemma~\ref{lem:inv-bound} and Lemma~\ref{lem:order-mix-measure}, there exists constants $w_k, \dots, w_{k_0} > 0$ depending on $G_0, \Theta$, and $k$ such that for all $G\in \Ocal_{k}$ so that $V(p_{G}, p_{G_0})$ small enough, we have
    \begin{equation*}
        V\left(p_{G}, p_{G_0}\right) \geq w_{\kappa} \divergence(G^{(\kappa)}, G_0), \quad\forall \kappa\in [k_0, k].
    \end{equation*}
    Apply this inequality for $G = \widehat{G}_n \in \Ocal_{k}$, we have that for all $n$ large enough, on event $A_n$,
    \begin{equation*}
        \divergence(\widehat{G}_n^{(\kappa)}, G_0) \leq \dfrac{1}{w_{\kappa}} V\left(p_{\widehat{G}_n^{(\kappa)}}, p_{G_0}\right)\leq \dfrac{c}{w_{\kappa}}  \left(\dfrac{\log n}{n}\right)^{1/2}.
    \end{equation*}
    Because $\divergence(\widehat{G}_n^{(\kappa)}, G_0) \gtrsim W_2^2(\widehat{G}_n^{(\kappa)}, G_0)$ for every $\kappa \geq k_0$ and $\divergence(\widehat{G}_n^{(k_0)}, G_0) \asymp W_1(\widehat{G}_n^{(k_0)}, G_0)$, we get the convergence rate for all level $\kappa\geq k_0$:
    \begin{equation*}
        W_2^2(\widehat{G}_n^{(\kappa)}, G_0)\lesssim \left(\dfrac{\log n}{n}\right)^{1/4},
    \end{equation*}
    and
    \begin{equation*}
        W_1(\widehat{G}_n^{(k_0)}, G_0)\lesssim \left(\dfrac{\log n}{n}\right)^{1/2},
    \end{equation*}
    where all constants only depends on $G_0, \Theta$, and $k$. 
    
    \paragraph{Part 2: Convergence rate on under-fitted levels.} To show the convergence rate for under-fitted levels, we further denote $\widehat{G}_n^{(k_0)} = \sum_{i=1}^{k_0} p^{n}_i\delta_{\theta_i^{n}}$. Rearranging the index of $\widehat{G}_n^{(k_0)}$ so that on $A_n$,
    \begin{equation}
        |p_i^{n} - p_i^0| \leq C\left(\dfrac{\log n}{n}\right)^{1/2}, \quad \norm{\theta_i^n - \theta_i^0} \lesssim\left(\dfrac{\log n}{n}\right)^{1/2},\quad \forall i\in [k_0].
    \end{equation}
    It is straightforward to show that for every pair $(i, j)\in [k_0]^2$, we have
    \begin{equation}\label{eq:conv-underfit-levels}
    \left|\dfrac{1}{(p_i^n)^{-1} + (p_j^n)^{-1}} \norm{\theta_i^n - \theta_j^n}^2 - \dfrac{1}{(p_i^0)^{-1} + (p_j^0)^{-1}} \norm{\theta_i^0 - \theta_j^0}^2\right|\lesssim \left(\dfrac{\log n}{n}\right)^{1/2}.
    \end{equation}
    Hence, on $A_n$, the optimal choice of indices $(i, j)$ to merge for $\widehat{G}_n^{(k_0)}$ will be the same as $G_0$ for every $n$ large enough. After merging, we also have 
    $$\left|(p_{i}^{n} + p_{j}^{n}) - (p_{i}^{0} + p_{j}^{0})\right| \lesssim \left(\dfrac{\log n}{n}\right)^{1/2},$$
    and
    $$\norm{\left(\dfrac{p_i^n}{p_i^{n}+p_j^{n}}\theta_i^{n} + \dfrac{p_j^n}{p_i^{n}+p_j^{n}}\theta_j^{n}\right) - \left(\dfrac{p_i^0}{p_i^{0}+p_j^{0}}\theta_i^{0} + \dfrac{p_j^0}{p_i^{0}+p_j^{0}}\theta_j^{0}\right)} \lesssim \left(\dfrac{\log n}{n}\right)^{1/2}.$$
    Hence, $W_1\left(\widehat{G}_n^{(k_0 - 1)}, G_0^{(k_0-1)} \right)\lesssim \left(\dfrac{\log n}{n}\right)^{1/2}$. The rest of the proof follows by means of induction. 
\end{proof}

\subsection{Proof of Theorem~\ref{thm:asymptotic-height}: Asymptotic behavior of the heights}
\begin{proof}[Proof of Theorem~\ref{thm:asymptotic-height}] Continue from the proof of Theorem~\ref{thm:asymptotic-dendrogram}, on $A_n$, for every $\kappa \geq k_0 + 1$, there exists two overfitted atoms $\theta_{k_1}^{n}$ and $\theta_{k_2}^{n}$ of $\widehat{G}_n^{(\kappa)}$ belongs to the same Voronoi cell $V_i$, so that
$$p_{k_1}^{n} \norm{\theta_{k_1}^{n} - \theta_{i}^0}^2 + p_{k_2}^{n} \norm{\theta_{k_2}^{n} - \theta_{i}^0}^2 \lesssim \left(\dfrac{\log n}{n}\right)^{1/2}.$$
Let $p^{n}_{*} = p^{n}_{k_1} + p^{n}_{k_2}$ and $\theta^{n}_{*} = \dfrac{p^{n}_{k_1}}{p_{*}}\theta_{k_1}^{n} + \dfrac{p^{n}_{k_2}}{p_{*}}\theta_{k_2}^{n}$, we can readily check the following equation:
\begin{equation}
    p_{k_1}^{n} \norm{\theta_{k_1}^{n} - \theta_{i}^0}^2 + p_{k_2}^{n} \norm{\theta_{k_2}^{n} - \theta_{i}^0}^2 = p_{*}^{n} \norm{\theta_{*}^{n} - \theta_{i}^0}^2 + \dfrac{1}{(p_{k_1}^{n})^{-1}+(p_{k_2}^{n})^{-1}} \norm{\theta_{k_1}^{n}-\theta_{k_2}^{n}}^2.
\end{equation}
Hence,
$$d_n^{(\kappa)}\leq  \dfrac{1}{(p_{k_1}^{n})^{-1}+(p_{k_2}^{n})^{-1}} \norm{\theta_{k_1}^{n}-\theta_{k_2}^{n}}^2 \lesssim \left(\dfrac{\log n}{n}\right)^{1/2}, \quad \forall \kappa \geq k_0 + 1.$$
When $\kappa \leq k_0$, the conclusion follows from inequality~\eqref{eq:conv-underfit-levels} in the proof of Theorem~\ref{thm:asymptotic-dendrogram}.
\end{proof}

\subsection{Proof of Theorem~\ref{thm:asymptotic-likelihood}: Asymptotic behavior of the likelihood}
Before proving this result, we would like to recall a known result used to bound the KL divergence by Hellinger distance (\cite{34Wong}, Theorem 5).
\begin{theorem}\label{thm:wong-shen-95}
    Let $p, q$ be two densities with $h(p, q)\leq \epsilon$. Suppose that $M_{\delta}^2 = \int_{\{p/q > e^{1/\delta}\}} p (p / q)^{\delta} < \infty$ for some $\delta \in (0, 1]$. Then for all $\epsilon \leq (1 - e^{-1})/\sqrt{2}$, we have
    \begin{equation}\label{eq:wong-shen-1}
        KL(p \| q) \leq \left[6 + \dfrac{2\log 2}{(1-e^{-1})^2} + \dfrac{8}{\delta} \max\{1, \log(M_{\delta} / \epsilon)\}\right] \epsilon^2,
    \end{equation}
    and
    \begin{equation}\label{eq:wong-shen-2}
        \int_{\mathcal{X}} p(x) \left(\log \dfrac{p(x)}{q(x)}\right)^2 dx \leq 5 \dfrac{1}{\delta} \max\{1, \log(M_{\delta} / \epsilon)\}^2 \epsilon^2,
    \end{equation}
\end{theorem}

\begin{proof}[Proof of Theorem~\ref{thm:asymptotic-likelihood}]
    We divide the proof into two parts: $\kappa \geq k_0$ (exact- and over-fitting) and $\kappa < k_0$ (under-fitting).
    
\paragraph{Part one: $\kappa \geq k_0$.} The empirical measure is denoted by $P_n = \dfrac{1}{n} \sum_{i=1}^{n} \delta_{x_i}$. For any $G$, we also denote $P_G$ by the distribution of $p_{G}$. We aim to provide the convergence rate for 
$$\left|\ell_n(\widehat{G}_n^{(\kappa)}) + H(p_{G_0})\right| = \left|\int \log p_{\widehat{G}_n^{(\kappa)}} dP_n - \int \log p_{G_0} dP_{G_0}\right|.$$
Denote by $D$ by the multiplicative constant in Theorem~\ref{thm:asymptotic-dendrogram}, i.e., we have that
$$\Pbb_{G_0}(W_2(\widehat{G}_n^{(\kappa)}, G_0)\leq D (\log n / n)^{1/4}) \geq 1 - c_1 n^{-c_2} \quad \forall \kappa \in [k_0, k],$$
for two universal constants $c_1, c_2 > 0$. Combining with the fact that $h \lesssim W_2$~\cite{nguyen2013convergence}, we also have 
$$\Pbb_{G_0}(h(p_{\widehat{G}_n^{(\kappa)}}, p_{G_0})\leq D (\log n / n)^{1/4}) \geq 1 - c_1 n^{-c_2} \quad \forall \kappa \in [k_0, k].$$ 
For a mixing measure $G$, we denote the increments of the empirical process by:
\begin{equation*}
    \nu_n(G) = \sqrt{n} \int_{\{p_{G_0} > 0\}} \log \dfrac{p_G}{p_{G_0}} d(P_n - P_{G_0}).
\end{equation*}
Using the triangle inequality, we have
\begin{align}\label{eq:bound-likelihood}
    \left|\int\log p_{\widehat{G}^{(\kappa)}_n} dP_n - \int \log p_{G_0} dP_{G_0}\right| & \overset{}{\leq} \left| \int\log \dfrac{p_{\widehat{G}^{(\kappa)}_n}}{p_{G_0}} d(P_n - P_{G_0})\right| \nonumber\\
    & + KL(p_{G_0} || p_{ \widehat{G}_n^{(\kappa)}}) +  \left|\int \log p_{G_0} (dP_n - dP_{G_0})\right|.
\end{align}
We will bound all three terms respectively. The first term can be bounded using Theorem~\ref{thm:uniform-bound-empirical-process}. Indeed, substitute $R = D \left(\dfrac{\log n}{n}\right)^{1/4}, a = D \dfrac{\log^{3/4}(n)}{n^{1/4}}$ in the result of that theorem, we have $a\leq \sqrt{n} R^2 \leq \sqrt{n} R$, and
$$a\geq R \left(\log\left(\dfrac{2^6\sqrt{n}}{a}\right)\right)^{1/2} \geq  \int_{a/(2^6\sqrt{n})}^{R} H_B^{1/2} \left(\dfrac{u}{\sqrt{2}}, \left\{p_{G} : G\in \Ocal_{k}, h(p_{G}, p_{G_0})\leq R \right\}, \nu\right) du.$$
Hence,
$$\Pbb_{G_0}\left(\sup_{h(p_G, p_{G_0})\leq D(\log n/n)^{1/4}} \left|\sqrt{n} \int\log \dfrac{p_{{G}}}{p_{G_0}} d(P_n - P_{G_0})\right|\geq D \dfrac{\log^{3/4}(n)}{n^{1/4}} \right)\leq n^{-c_2},$$
for some universal constant $c_2$. Combining with the bound on the Hellinger distance, we have
\begin{align*}
&\Pbb_{G_0}\left(\left|\int\log \dfrac{p_{\widehat{G}^{(\kappa)}_n}}{p_{G_0}} d(P_n - P_{G_0})\right|\geq D \dfrac{\log^{3/4}(n)}{n^{3/4}} \right)\leq \Pbb_{G_0}\left(h(p_{\widehat{G}_n^{(\kappa)}}, p_{G_0})\geq D (\log n / n)^{1/4}\right)\\
&+ \Pbb_{G_0}\left(\left|\int\log \dfrac{p_{\widehat{G}^{(\kappa)}_n}}{p_{G_0}} d(P_n - P_{G_0})\right|\geq D \dfrac{\log^{3/4}(n)}{n^{3/4}}, h(p_{\widehat{G}_n^{(\kappa)}}, p_{G_0})\leq D (\log n / n)^{1/4} \right)\\
&\leq c_1 n^{-c_2},
\end{align*}
for two universal constants $c_1$ and $c_2$. 

Next, we use Theorem~\ref{thm:wong-shen-95} to have that on the event $\{h(p_{\widehat{G}_n^{(\kappa)}}, p_{G_0}) \leq D (\log n / n)^{1/4}\}$, we also have $$KL(p_{G_0}||p_{\widehat{G}_n})\leq D c_3\log((\log n / n)^{-1/4})  (\log n / n)^{1/2}\leq D c_3 \dfrac{\log^{3/2}(n)}{n^{1/2}},$$
where $c_3$ is a universal constant. Hence, we also have $$KL(p_{G_0}||p_{\widehat{G}_n})\leq  \dfrac{\log^{3/2}(n)}{n^{1/2}}$$
with the probability at most $c_1 n^{-c_2}$.
To bound the last term, we can simply use the Chebyshev inequality,
$$\Pbb_{G_0}\left(\left|\int \log p_{G_0} d(P_n - P_{G_0})\right| \geq t \right)\leq \dfrac{\text{Var}(\log p_{G_0})}{nt^2}.$$
By choosing $t \gtrsim n^{-1/2}$, we can also bound this term. For example, we can choose $t = (\log n / n)^{1/4}$ and yield  
$$\Pbb_{G_0}\left(\left|\int \log p_{G_0} d(P_n - P_{G_0})\right| \geq (\log n / n)^{1/4} \right)\leq c_1 n^{-c_2}.$$
By combining those inequalities together, we have
$$\Pbb_{G_0}\left(\left|\int\log p_{\widehat{G}^{(\kappa)}_n} dP_n - \int \log p_{G_0} dP_{G_0}\right| \geq (\log n / n)^{1/4}\right)\leq c_1 n^{-c_2},$$
for some universal constants $c_1$ and $c_2$.

\paragraph{Part two: $\kappa < k_0$.} Because $|\log p_{G}(x)| \leq m(x)$ for a measurable function $m(x)$ for all $G\in \Ocal_{\kappa}$, we can use uniform law of large number (see, e.g., \cite{17Robert} Theorem 9.2) to have that
\begin{equation*}
    \sup_{G\in \Ocal_{\kappa}} \left|\ell_n(G) - \Ebb_{X\sim p_{G_0}} \log p_G(X) \right|\xrightarrow{P} 0,
\end{equation*}
where $\xrightarrow{P}$ means convergence in probability. Therefore,
\begin{equation*}
    \left|\ell_n(\widehat{G}_n^{(\kappa)}) - \Ebb_{X\sim p_{G_0}} \log p_{\widehat{G}_n^{(\kappa)}}(X) \right|\xrightarrow{P} 0,
\end{equation*}
Besides, by Theorem~\ref{thm:asymptotic-dendrogram}, we have $\log p_{\widehat{G}_n^{(\kappa)}} \to \log p_{\widehat{G}_0^{(\kappa)}}$ in probability, an application of Dominating Convergence Theorem yields:
\begin{equation*}
    \Ebb_{X\sim p_{G_0}} \log p_{\widehat{G}_n^{(\kappa)}}(X) \xrightarrow{P} \Ebb_{X\sim p_{G_0}} \log p_{\widehat{G}_0^{(\kappa)}}(X).
\end{equation*}
Combining the limits above together, we have
\begin{equation}
    \ell_n(\widehat{G}_n^{(\kappa)}) \xrightarrow{P} \Ebb_{p_{G_0}} \log p_{G_0^{(\kappa)}}(X).
\end{equation}
This, in turn, implies the conclusion of this theorem:
\begin{equation}
    \left|\ell_n(\widehat{G}_n^{(\kappa)}) + H(p_{G_0})\right| \xrightarrow{P} KL(p_{G_0}\| p_{G_0^{(\kappa)}}).
\end{equation}
\end{proof}

\subsection{Proof of Proposition~\ref{prop:cut-dendrogram}: Cuts of the dendrogram}
\begin{proof}[Proof of Proposition~\ref{prop:cut-dendrogram}]
    From Theorem~\ref{thm:asymptotic-height}, there exist a constant $C = C(G_0, \Theta, k)$ and universal constants $c_1, c_2$, such that for every $n$, the event
    $$A_n = \left\{d_n^{(\kappa)} \leq C \left(\dfrac{\log n}{n}\right)^{1/2} \forall \kappa \in [k_0 + 1, k] \right\} \cap \left\{d_n^{(k_0)} \geq d_0^{(k_0)} - C \left(\dfrac{\log n}{n}\right)^{1/2} \right\}$$
    has $\Pbb_{p_{G_0}}(A_n) \geq 1 - c_1 n^{-c_2}$. Because $d_0^{(k_0)}$ is a fixed and positive number, for every $n$ large enough, we have
    $$C\left(\dfrac{\log n}{n}\right)^{1/2} < \epsilon_n <  d_0^{(k_0)} - C \left(\dfrac{\log n}{n}\right)^{1/2},$$
    Therefore, $k_n = k_0$ on $A_n$. This implies $\Pbb_{p_{G_0}}(k_n = k_0)\geq 1 - c_1 n^{-c_2}\to 1$ as $n\to \infty$. Hence, $k_n \to k_0$ in probability. 
\end{proof}

\subsection{Proof of of Proposition~\ref{prop:DIC-consistency}: Consistency of DIC}
\begin{proof}[Proof of Proposition~\ref{prop:DIC-consistency}]
By combining Theorem~\ref{thm:asymptotic-height} and~\ref{thm:asymptotic-likelihood}, there exists a set $A_n$ with $\Pbb_{p_{G_0}}(A_n) \to 1$ such that on $A_n$, we have
\begin{equation}
d_n^{(\kappa)} = \begin{cases}
        O((\log n / n)^{1/2}) & \text{if } \kappa > k_0 \\
        d_0^{(\kappa)}+O((\log n / n)^{1/2}) & \text{if } \kappa \leq k_0,
    \end{cases}
\end{equation}
and
\begin{equation}
    \overline{\ell}_n^{(\kappa)} = \begin{cases}
        - H(p_{G_0}) + O((\log n / n)^{1/4}) & \text{if } \kappa \geq k_0 \\
        - H(p_{G_0}) - L^{(\kappa)} + o(1) & \text{if } \kappa < k_0,
    \end{cases}
\end{equation}
where $L^{(\kappa)} = KL(p_{G_0} \| p_{G_{0}^{(\kappa)}}) > 0$ for all $\kappa < k_0$, and the constants in the big $O$ and small $o$ notions only depends on $G_0, \Theta$, and $k$. 
Hence, we have
\begin{equation}
    \text{DIC}_n^{(\kappa)} = \begin{cases}
        \omega_n H(p_{G_0}) + O(\omega_n(\log n / n)^{1/4}) & \text{if } \kappa > k_0, \\
        \omega_n H(p_{G_0}) - d_0^{(k_0)} + O(\omega_n(\log n / n)^{1/4}) & \text{if } \kappa = k_0,\\
        \omega_n H(p_{G_0}) + \omega_n L^{(\kappa)} - d_0^{(\kappa)} + o(\omega_n) & \text{if } \kappa < k_0.
    \end{cases}
\end{equation}
Because $\omega_n\to \infty$, $\omega_n (\log n / n)^{1/4}\to 0$, and $L^{(\kappa)}$ is strictly positive, this implies that $\text{DIC}^{(k_0)}$ is the smallest number for all $n$ large enough. Hence, $\Pbb_{p_{G_0}}(k_n^* = k_0)\geq \Pbb_{p_{G_0}}(A_n)\to 1$, which means $k_n^* \to k_0$ in probability.
\end{proof}

\section{Proof of Section~\ref{sec:merge-weak}}\label{sec:proof-merge-weak}
\subsection{Proof of Lemma~\ref{lem:inv-bound-weak}: Inverse bound for location-scale Gaussian mixtures}
Recall that for any $k$, $\overline{r}(k)$ is defined as the smallest integer $r$ such that the system of polynomial equations
\begin{equation}\label{eq:system-algebraic-eq}
    \sum_{j=1}^{k} \sum_{\substack{n_1, n_2\in \Nbb\\ n_1 + 2 n_2 = \alpha}} \dfrac{c_j^2 a_j^{n_1} b_j^{n_2}}{n_1 ! n_2 !} = 0, \quad \forall \alpha = 1, \dots, r
\end{equation}
does not have any nontrivial solution $(a_j, b_j, c_j)_{j=1}^{k} \subset \Rbb$, where a set of solutions is considered nontrivial if all variable $c_j$'s are non-zero and at least one of $a_j$'s is non-zero.

This proof proceeds in the same way as of Lemma~\ref{lem:inv-bound}; see also~\cite{12Ho, 21Tudor}. The main difference compared to the proof of Lemma~\ref{lem:inv-bound} is the lack of strong identifiability for the location-scale Gaussian mixtures, so we need to use the heat equation~\eqref{eq:heat-eq} to transform all derivatives with respect to $(\mu, \Sigma)$ of the Taylor expansion around $p_{G_0}$ to derivative with respect to $\mu$ only. The contradiction is then derived from the fact that mixtures of location-scale Gaussian are strongly identifiable with respect to mean parameters up to any order. The set of algebraic equations~\eqref{eq:system-algebraic-eq} is used to guarantee that not all coefficients in the identifiability equation vanish. 

Similar to the previous inverse bound, we also would like to highlight that our result is tighter than the existing result in the literature due to the fact that we make use of more coefficients in the Taylor expansion. In particular, not only do the added first and second-order terms help the inverse bound to become stronger, but they also lead to a good merging scheme that can recover the root-$n$ convergence rate for parameter estimation in the overfitted location-scale Gaussian mixtures.

\begin{proof}[Proof of Lemma~\ref{lem:inv-bound-weak}]
The proof is divided into some small steps.

\paragraph{Step 1: Proving by contradiction and setup.} Similar to Step 1 of the proof of Lemma~\ref{lem:inv-bound}, we assume that the inverse bound is incorrect, which means there exists a sequence $G_n \in \Ocal_{k, c_0}$ such that 
\begin{equation}\label{eq:contradiction-location-scale}
V(p_{G_n}, p_{G_0})\to 0 \quad \text{and}\quad \dfrac{V(p_{G_n}, p_{G_0})}{\divergence_{\Gcal}(G_n, G_0)}\to 0.
\end{equation}
Due to the (zero-order) identifiability of location-scale Gaussian mixtures, we have 
$G_n\to G_0$ in $W_1$, so that we can represent
$$G_n = \sum_{i=1}^{k_0}\sum_{j=1}^{s_i} p_{ij}^{n} \delta_{(\mu_{ij}^{n}, \Sigma_{ij}^{n})} \in \Ecal_{k_*} \quad \forall n\in \Nbb,$$
where
$$\sum_{j=1}^{s_i} p_{ij}^{n}\to p_{i}^{0}, \quad \mu_{ij}^{n}\to \mu_i^0, \quad \Sigma_{ij}^{n}\to \Sigma_i^0,\quad \forall j\in [s_i], i\in [k_0],$$
and $\sum_{i=1}^{k_0} s_i = k_* \leq k$. Notice that we assume in this lemma that all mixing proportions $p_{ij}^{n}$ of $G_n$ are bounded below by a positive constant $c_0$. Therefore, we do not need to consider vanishing atoms. For all $n$ large enough, we have $(\mu_{ij}^{n}, \Sigma_{ij}^{n}) \in V_i$ for all $i, j$, so that $s_i = |V_i|$, and
\begin{align*}
    \divergence_{\Gcal}(G_n, G_0) & = \sum_{i=1}^{k_0} \left|\sum_{j=1}^{s_i} p_{ij}^{n} - p_i^0 \right| + \left\|\sum_{j=1}^{s_i} p_{ij}^{n} (\mu_{ij}^{n} - \mu_i^0)\right\| \\
    & + \left\|\sum_{j=1}^{s_i} p_{ij}^{n} \left(\Sigma_{ij}^{n} - \Sigma_{i}^{0} + (\mu_{ij}^{n} - \mu_i^0)(\mu_{ij}^{n} - \mu_i^0)^{\top}\right)\right\|\\
    & + \sum_{j=1}^{s_i} p_{ij}^{n} \left(\norm{\mu_{ij}^{n} - \mu_{i}^{0}}^{\overline{r}(s_i)} + \norm{\Sigma_{ij}^{n} - \Sigma_{i}^{0}}^{\overline{r}(s_i)/2} \right).
\end{align*}

\paragraph{Step 2: Taylor expansion.} The estimation of $p_{G_n}$ to $p_{G_0}$ can be written as following:
\begin{align}\label{eq:expansion-Gaussian}
    & p_{G_n}(x) - p_{G_0}(x)  = \sum_{i=1}^{k_0} \sum_{j=1}^{s_i} p_{ij}^{n} f(x|\mu_{ij}^{n}, \Sigma_{ij}^{n}) - \sum_{i=1}^{k_0} p_i^0 f(x|\mu_{i}^0, \Sigma_i^0) \nonumber\\
    & \hspace{1cm} = \sum_{i=1}^{k_0} \left(\sum_{j=1}^{s_i} p_{ij}^{n} - p_i^{0} \right) f(x|\mu_i^0, \Sigma_i^0) + \sum_{j=1}^{s_i} p_{ij}^{n} \left(f(x|\mu_{ij}^{n}, \Sigma_{ij}^{n}) - f(x|\mu_i^0, \Sigma_i^0)\right).
\end{align}
By means of Taylor expansion up to $\overline{r}(s_i)$-th order, we have
\begin{align}\label{eq:Taylor-for-one-Gaussian}
    f(x|\mu_{ij}^{n}, \Sigma_{ij}^{n}) - f(x|\mu_i^0, \Sigma_i^0) & = \sum_{\alpha, \beta} \dfrac{1}{\alpha! \beta!} (\mu_{ij}^{n} - \mu_i^0)^{\alpha} (\Sigma_{ij}^{n} - \Sigma_i^{0})^{\beta} \dfrac{\partial^{|\alpha| + |\beta|}}{\partial \mu^{\alpha} \partial \Sigma^{\beta}} f(x|\mu_i^0, \Sigma_i^0)\nonumber \\
    & + R^{n}_{ij}(x),
\end{align}
where the sum is over all the multi-indices $\alpha \in \Nbb^{d}, \beta \in \Nbb^{d\times d}$ such that $1\leq |\alpha| + |\beta|\leq \overline{r}(s_i)$, and $R^n_{ij}(x) = o\left(\norm{\mu_{ij}^{n} - \mu_i^0}^{\overline{r}(s_i)} + \norm{\Sigma_{ij}^{n} - \Sigma_i^0}^{\overline{r}(s_i)}\right)$, uniformly in $x$. In the display above, we use the usual notation for calculation with multi-indices, that is, for $\alpha = (\alpha_1, \dots, \alpha_d) \in \Nbb^{d}$, $\mu\in \Rbb^{d}$ and a function $f$ of $\mu$, we denote
$$|\alpha| = \sum_{u=1}^{d} \alpha_u, \quad \alpha! = \prod_{u=1}^{d} \alpha_u!, \quad \mu^{\alpha} = \prod_{u=1}^{d} \mu_u^{\alpha_u}, \quad \dfrac{\partial^{|\alpha|}}{\partial^{\alpha} \mu} f(\mu) = \dfrac{\partial^{|\alpha|}}{\partial \mu_1^{\alpha_1}\cdots \partial \mu_d^{\alpha_d}} f(\mu).$$
Similar notations apply for multi-index $\beta \in \Nbb^{d\times d}$ and $\Sigma\in \Rbb^{d\times d}$. By applying the heat equation~\eqref{eq:heat-eq}, we can write
\begin{align*}
    \dfrac{\partial^{|\alpha| + |\beta|}}{\partial \mu^{\alpha} \partial \Sigma^{\beta}} f(x | \mu, \Sigma) = \dfrac{1}{2^{|\beta|}} \dfrac{\partial^{|\alpha| + 2|\beta|}}{\partial \mu^{\tau_0(\alpha, \beta)}} f(x | \mu, \Sigma), \quad\forall \mu, \Sigma,
\end{align*}
where $\tau_{0}(\alpha, \beta) = (\alpha_u + \sum_{v} \beta_{uv} + \beta_{vu})_{u=1}^{d} \in \Nbb^{d}$ for any multi-indices $\alpha\in \Nbb^{d}$ and $\beta\in \Nbb^{d\times d}$. 
By combining this with equations~\eqref{eq:expansion-Gaussian} and~\eqref{eq:Taylor-for-one-Gaussian}, we have
$$p_{G_n}(x) - p_{G_0}(x) = \sum_{i=1}^{k_0} \sum_{|\tau| = 0}^{2\overline{r}(s_i)} a_{i, \tau} \dfrac{\partial^{|\tau|}}{\partial \mu^{\tau}} f(x|\mu_i^0, \Sigma_i^0) + R_n(x),$$
where $a_{i, \tau} = \left(\sum_{j=1}^{s_i} p_{ij}^{n} - p_i^0\right)$ when $\tau = 0 \in \Nbb^{d}$,
$$a_{i, \tau} = 
    \sum_{j=1}^{s_i} \sum_{\substack{\alpha\in \Nbb^{d}, \beta\in \Nbb^{d\times d} \\ 1 \leq |\alpha| + |\beta|\leq \overline{r}(s_i)\\ \tau_0(\alpha, \beta) = \tau}} \dfrac{1}{2^{|\beta|}\alpha! \beta!} p_{ij}^{n} (\mu_{ij}^{n} - \mu_i^0)^{\alpha} (\Sigma_{ij}^{n} - \Sigma_i^0)^{\beta}, \quad \forall |\tau| \geq 1, $$
and 
$$R_n(x) = o\left(\sum_{i, j}p_{ij}^{n}\left(\norm{\mu_{ij}^{n} - \mu_i^0}^{\overline{r}(s_i)} + \norm{\Sigma_{ij}^{n} - \Sigma_i^0}^{\overline{r}(s_i)}\right)\right) = o(D_n),$$
uniformly in $x\in \Rbb^{d}$, for $D_n = \divergence_{\Gcal}(G_n, G_0)$. Especially, we have that
$$\sum_{|\tau| = 1} |a_{i, \tau}| = \sum_{u=1}^{d} \left|\sum_{j=1}^{s_i} p_{ij}^n (\mu_{iju}^{n} - \mu_{iu}^0) \right| \asymp \left\|\sum_{j=1}^{s_i} p_{ij}^n (\mu_{ij}^{n} - \mu_{i}^0) \right\|,$$
and
\begin{align*}\sum_{|\tau| = 2} |a_{i, \tau}| &= \dfrac{1}{2} \sum_{u, v=1}^{d} \left|\sum_{j=1}^{s_i} p_{ij}^n (\mu_{iju}^{n} - \mu_{iu}^0)(\mu_{ijv}^{n} - \mu_{iv}^0)  + (\Sigma_{ijuv}^{n} - \Sigma_{iuv}^{0}) \right| \\
&\asymp \left\|\sum_{j=1}^{s_i} p_{ij}^n \left((\mu_{ij}^{n} - \mu_{i}^0)(\mu_{ij}^{n} - \mu_{i}^0)^{\top} + \Sigma_{ij}^{n} - \Sigma_i^0\right) \right\|.
\end{align*}

\paragraph{Step 3: Non-vanishing coefficient.} Our goal in this step is to show that $\max_{i, \tau} |a_{i, \tau}| \gtrsim D_n$ as $n\to \infty$. Because $D_n = \sum_{i=1}^{k_0} D_{n, i}$, where
\begin{align*}
    D_{ni} & = \left|\sum_{j=1}^{s_i} p_{ij}^{n} - p_i^0 \right| + \left\|\sum_{j=1}^{s_i} p_{ij}^{n} (\mu_{ij}^{n} - \mu_i^0)\right\| \\
    & + \left\|\sum_{j=1}^{s_i} p_{ij}^{n} \left(\Sigma_{ij}^{n} - \Sigma_{i}^{0} + (\mu_{ij}^{n} - \mu_i^0)(\mu_{ij}^{n} - \mu_i^0)^{\top}\right)\right\|\\
    & + \sum_{j=1}^{s_i} p_{ij}^{n} \left(\norm{\mu_{ij}^{n} - \mu_{i}^{0}}^{\overline{r}(s_i)} + \norm{\Sigma_{ij}^{n} - \Sigma_{i}^{0}}^{\overline{r}(s_i)/2} \right),    
\end{align*}
    it suffices to prove for every single $i\in [k_0]$ that $\max_{\tau} |a_{i, \tau}| \gtrsim D_{ni}$ as $n\to \infty$. Further decompose $D_{ni} \asymp D_{ni}^{(0)} + D_{ni}^{(1)} + D_{ni}^{(2)} + \sum_{u=1}^{d} D_{niu} + \sum_{u\neq v}^{d} D_{niuv}$, where $D_{ni}^{(0)} = \left|\sum_{j=1}^{s} p_{ij}^n - p_i^0\right|$, $D_{ni}^{(1)} = \norm{\sum_{j=1}^{s} p_j^n (\mu_j^n - \mu^0)}$,  $D_{ni}^{(2)} = \norm{\sum_{j=1}^{s} p_j^n \left(\Sigma_j^n - \Sigma^0 + (\mu_j^n - \mu^0)(\mu_j^n - \mu^0)^{\top}\right)}$, $D_n^{u} = \sum_{j=1}^{s} p_j^{n} |\mu_{ju}^{n} - \mu_u^{0}|^{\overline{r}(s)}$, and $D_u^{n} = \sum_{j=1}^{s} p_j^{n} |\mu_{ju}^{n} - \mu_u^{0}|^{\overline{r}(s)} + |\Sigma_{j uu}^{n} - \Sigma_{uu}^0|^{\overline{r}(s)/2}$, and $D_{uv}^{n} = \sum_{j=1}^{s} p_j^{n} |\Sigma_{j uv}^{n} - \Sigma_{uv}^0|^{\overline{r}(s)/2}$ for indices $u\neq v\in [d]$. We can see that
$$|a_0| = D_n^{(0)}, \quad \sum_{|\tau| = 1} |a_{\tau}| \asymp D_n^{(1)}, \quad \sum_{|\tau| = 2} |a_{\tau}| \asymp D_n^{(2)}.$$
Hence, it suffices to prove that $\max_{\tau} |a_{i, \tau}| \gtrsim D^n_{niu}, D^n_{niuv}$ for all $u\neq v\in [d]$. The claim of this step then proceeds similarly to Step 3 of Lemma 10 in~\cite{21Tudor}, via the definition of $\overline{r}$ and the system of equations~\eqref{eq:system-algebraic-eq}.

\paragraph{Step 4: Deriving a contradiction using strong identifiability of mixture of location Gaussians.} Set $m_n = \max_{i, \tau} \dfrac{|a_{i, \tau}|}{D_n} \not \to 0$ as $n\to \infty$. Because $\dfrac{|a_{i, \tau} / D_n|}{m_n} \leq 1$, by extracting a subsequence if needed, by the compactness of $[-1, 1]$, we have $\dfrac{|a_{i, \tau} / D_n|}{m_n} \to \gamma_{i, \tau} \in [-1, 1]$ for all $i, \tau$, and at least one of such coefficient is 1. By means of Fatou's lemma,
\begin{align*}
    0 = \lim_{n\to \infty}\dfrac{V(p_{G_n}, p_{G_0})}{D_n m_n} & \geq \int \left|\liminf_{n\to \infty} \sum_{i=1}^{k_0}\sum_{|\tau|=0}^{2\overline{r}(s_i)} \dfrac{a_{i, \tau}}{D_n m_n} \dfrac{\partial^{|\tau|}}{\partial \mu^{\tau}} f(x|\mu_i^0, \Sigma_i^0)\right|dx \\
    & = \left|\sum_{i=1}^{k_0}\sum_{|\tau|=0}^{2\overline{r}(s_i)} \gamma_{i, \tau}\dfrac{\partial^{|\tau|}}{\partial \mu^{\tau}} f(x|\mu_i^0, \Sigma_i^0)\right|dx.
\end{align*}
Hence, we have 
$$\sum_{i=1}^{k_0}\sum_{|\tau|=0}^{2\overline{r}(s_i)} \gamma_{i, \tau}\dfrac{\partial^{|\tau|}}{\partial \mu^{\tau}} f(x|\mu_i^0, \Sigma_i^0) = 0,$$
almost surely in $x$, with respect to the Lebesgue measure. By the identifiability of the location Gaussian mixture, all coefficients $\gamma_{i, \tau} = 0$, but this is contradictory with the claim that at least one of them is 1 at the beginning of this step. Hence, the assumption~\eqref{eq:contradiction-location-scale} is false. The inverse bound is proved by means of proving by contradiction.
\end{proof}
\subsection{Evolution of Gaussian divergences on the dendrogram}
Before establishing the asymptotic results for the dendrogram of mixtures of location-scale Gaussians, it is helpful to have a similar result as in Lemma~\ref{lem:order-mix-measure} so that the convergence rate goes through for all latent mixing measures on the dendrogram. Recall that for a latent mixing measure $G=\sum_{j=1}^{k} p_j \delta_{(\mu_j, \Sigma_j)}$, we denote $G = G^{(k)}, G^{(k-1)}, \dots, G^{(1)}$ the latent mixing measures on the dendrogram of $G$. 
\begin{lemma}\label{lem:order-mix-measure-weak}
    For $G^{(k)}\in \Ecal_{k, c_0}$, as $\divergence(G^{(k)}, G_0) \to 0$, we have
    \begin{equation}\label{eq:ineq-weak-order}
        \divergence(G^{(k)}, G_0) \gtrsim \divergence(G^{(k-1)}, G_0) \gtrsim \dots \gtrsim \divergence(G^{(k_0)}, G_0),
    \end{equation}
    where the multiplicative constants in these inequalities only depend on $G_0, \Theta$, and $k$.
\end{lemma}
\begin{proof}
    It suffices to prove the inequality $\divergence(G^{(k)}, G_0) \gtrsim \divergence(G^{(k-1)}, G_0)$ and the rest are similar. Suppose 
    $G_n = \sum_{i=1}^{k} p_i^n \delta_{(\mu_i^n, \Sigma_i^n)}\in \Ecal_{k, c_0}$ varies so that
    \begin{align*}
        \divergence_{\Gcal}(G, G_0) & = \sum_{i=1}^{k_0} \left(\left|\sum_{j\in V_i} p_j - p_i^0 \right| + \sum_{j\in V_i} p_j\left(\norm{\mu_j - \mu_i^0}^{\overline{r}(|V_i|)} + \norm{\Sigma_j - \Sigma_i^0}^{\overline{r}(|V_i|)/2}\right) \right.\nonumber \\
    & \hspace{-.7cm} \left. + \norm{\sum_{j\in V_i}p_j (\mu_j - \mu_i^0)}  + \norm{\sum_{j\in V_i}p_j \left((\mu_j - \mu_i^0) (\mu_j - \mu_i^0)^{\top} + \Sigma_j - \Sigma_i^0)\right)}\right)\to 0.
    \end{align*}
    Because all $p_j$'s are lower bounded by $c_0$, we have $\mu_j \to \mu_i^0, \Sigma_j\to \Sigma_i^{0}$ for all $j\in V_i$. Hence, $G\to G_0$, the merging pair of indices $(j_1, j_2)$ must belong to a common $V_{i}$ that WLOG we assume $j_1 = 1, j_2 = 2, i = 1$. Let the merged atom be $p_{*} \delta_{(\mu_{*}, \Sigma_{*})}$, i.e.,
    $$p_{*} = p_{1} + p_{2}, \quad \mu_{*} = \dfrac{p_{1}}{p_{*}} \mu_{1} + \dfrac{p_{2}}{p_{*}} \mu_{2},$$
    and
    \begin{align}\label{eq:Sigma-merge-equivalent}
        \Sigma_* & = \dfrac{p_{1}}{p_*} \left(\Sigma_1 + (\mu_1 - \mu_*)(\mu_1 - \mu_*)^{\top}\right) + \dfrac{p_{2}}{p_*} \left(\Sigma_2 +(\mu_2 - \mu_*)(\mu_2 - \mu_*)^{\top}\right) \nonumber\\
        & = \dfrac{p_{1}}{p_*} \Sigma_1 + \dfrac{p_{2}}{p_*}\Sigma_2 + \dfrac{1}{(p_{1})^{-1} + (p_{2})^{-1}} \left(\mu_1 - \mu_2\right)\left(\mu_1 - \mu_2\right)^{\top}
    \end{align}
    Hence, we have that
    $$\left|\sum_{j\in V_1} p_j - p_1^0\right| = \left|\sum_{j\in V_1, j\not \in \{1, 2\}} p_j + p_{*} - p_1^0\right|,$$
    $$\left\|\sum_{j\in V_1} p_j (\mu_j - p_1^0)\right\| = \left\|\sum_{j\in V_1, j\not \in \{1, 2\}} p_j(\mu_j - \mu_1^0) + p_*(\mu_* - \mu_1^0)\right\|,$$
    and
    \begin{align*}
        & \norm{\sum_{j\in V_i}p_j \left((\mu_j - \mu_i^0) (\mu_j - \mu_i^0)^{\top} + \Sigma_j - \Sigma_i^0)\right)} \\
        & = \Bigg\|\sum_{j\in V_i, j\not \in \{1, 2\}}p_j \left((\mu_j - \mu_i^0) (\mu_j - \mu_i^0)^{\top} + \Sigma_j - \Sigma_i^0) \right) \\
        & \hspace{2cm} + p_* \left((\mu_* - \mu_i^0) (\mu_* - \mu_i^0)^{\top} + \Sigma_* - \Sigma_i^0)\right)\Bigg\|
    \end{align*}
    To this end, because $\overline{r}(\kappa)\geq 4$ for all $\kappa\geq 2$, it suffices to show that for any ${\overline{r}}\geq 4$,
    \begin{equation}\label{eq:evol-weak-ineq-1}
        p_1\norm{\mu_{1}-\mu_1^0}^{\overline{r}} + p_2\norm{\mu_{2}-\mu_1^0}^{\overline{r}} \gtrsim (p_1+p_2)\norm{\mu_* - \mu_1^0}^{\overline{r}},
    \end{equation}
    and 
    \begin{align}\label{eq:evol-weak-ineq-2}
        p_1\norm{\mu_{1}-\mu_1^0}^{\overline{r}} + p_2\norm{\mu_{2}-\mu_1^0}^{\overline{r}} & + p_1\norm{\Sigma_{1}-\Sigma_1^0}^{\overline{r}/2} + p_2\norm{\Sigma_{2}-\Sigma_1^0}^{\overline{r}/2} \nonumber\\
        & \gtrsim (p_1+p_2)\norm{\Sigma_* - \Sigma_1^0}^{\overline{r}/2}.
    \end{align}
    The inequality~\eqref{eq:evol-weak-ineq-1} is true because of $\norm{\cdot}^{\overline{r}}$ is a convex function (as $\overline{r} \geq 4$). To prove~\eqref{eq:evol-weak-ineq-2}, we first use H\"older's inequality combining with~\eqref{eq:Sigma-merge-equivalent} to get
    $$p_*\norm{\Sigma_* - \Sigma_1^0}^{\overline{r}/2} \leq p_*\norm{\dfrac{p_{1}}{p_*} \Sigma_1 + \dfrac{p_{2}}{p_*}\Sigma_2 - \Sigma_1^0}^{\overline{r}/2} + p_*\left(\dfrac{1}{p_1^{-1} + p_2^{-1}}\right)^{\overline{r}/2}\norm{\mu_1 - \mu_2}^{\overline{r}}.$$
    The first term on the RHS can be bounded by $p_1\norm{\Sigma_{1}-\Sigma_1^0}^{\overline{r}/2} + p_2\norm{\Sigma_{2}-\Sigma_1^0}^{\overline{r}/2}$ (because of the convexity of norm). The second term can be bounded as follows.
    \begin{align*}
        p_1\norm{\mu_{1}-\mu_1^0}^{\overline{r}} + p_2\norm{\mu_{2}-\mu_1^0}^{\overline{r}}& \geq \min\{p_1, p_2\} \left(\norm{\mu_{1}-\mu_1^0}^{\overline{r}} + \norm{\mu_{2}-\mu_1^0}^{\overline{r}}\right)\\
        & \gtrsim \dfrac{1}{p_1^{-1} + p_2^{-1}} \norm{\mu_1 - \mu_2}^{\overline{r}}\\
        &\geq  (p_1+p_2)\left(\dfrac{1}{p_1^{-1} + p_2^{-1}}\right)^{2} \norm{\mu_1 - \mu_2}^{\overline{r}}\\
        &\geq  (p_1+p_2)\left(\dfrac{1}{p_1^{-1} + p_2^{-1}}\right)^{\overline{r}/2} \norm{\mu_1 - \mu_2}^{\overline{r}},
    \end{align*}
    where the first inequality is obvious, the second uses H\"older's inequality combined with the fact that the minimum of two numbers is no less than their harmonic mean, the third one is equivalent to $p_1p_2\leq 1$ (which is correct), and the last one use the condition that $\overline{r} \geq 4$. Hence, the inequality~\eqref{eq:ineq-weak-order} is proved.
\end{proof}

\subsection{Proof of Theorem~\ref{thm:asymp-theory-gaussian}}
\begin{proof}[Proof of Theorem~\ref{thm:asymp-theory-gaussian}]
Firstly, the proof of parameter estimation rate proceeds similarly to that of Theorem~\ref{thm:asymptotic-dendrogram}, but we use Lemma~\ref{lem:inv-bound-weak} and~\ref{lem:order-mix-measure-weak} in the place of Lemma~\ref{lem:inv-bound} and~\ref{lem:order-mix-measure}. 

For the convergence rate of the height at all levels $\kappa \geq k_0 + 1$, from the previous part combined with Lemma~\ref{lem:order-mix-measure}, we have
$$\divergence_{\Gcal}(\widehat{G}_n^{(\kappa)}, G_0) \lesssim \left(\dfrac{\log n}{n}\right)^{1/2}.$$
Because $\kappa \geq k_0 + 1$, by the pigeonhole principle, we have that there exist at least two indices $i, j\in [\kappa]$ such that two atoms $p_{i}^{n} \delta_{(\mu_i^{n}, \Sigma_i^{n})}$ and $p_{j}^{n} \delta_{(\mu_j^{n}, \Sigma_j^{n})}$ belongs to a common Voronoi cell of some $(\mu_t^0, \Sigma_t^0)$ (we suppress the dependence of $i$, $j$, and $V_t$ on $n$ for ease of notation). Hence, 
\begin{align*}
& p_i^{n} \left(\norm{\mu_i^{n} - \mu_t^0}^{\overline{r}(|V_t|)} + \norm{\Sigma_i^{n} - \Sigma_t^0}^{\overline{r}(|V_t|)/2}\right) \\
& + p_j^{n} \left(\norm{\mu_j^{n} - \mu_t^0}^{\overline{r}(|V_t|)} + \norm{\Sigma_j - \Sigma_t^0}^{\overline{r}(|V_t|)/2} \right) \lesssim \left(\dfrac{\log n}{n}\right)^{1/2}.
\end{align*}
Using the fact that $\min\{p_i^n, p_j^n\}\geq \dfrac{1}{(p_i^{n})^{-1} + (p_j^{n})^{-1}}$, $\overline{r}(\widehat{G}_n)\geq \overline{r}(|V_t|)\geq \overline{r}(2) = 4$, and combining with the H\"older's inequality, we have
\begin{align*}
& p_i^{n} \left(\norm{\mu_i^{n} - \mu_t^0}^{\overline{r}(|V_t|)} + \norm{\Sigma_i^{n} - \Sigma_t^0}^{\overline{r}(|V_t|)/2}\right) + p_j^{n} \left(\norm{\mu_j^{n} - \mu_t^0}^{\overline{r}(|V_t|)} + \norm{\Sigma_j - \Sigma_t^0}^{\overline{r}(|V_t|)/2} \right)\\
&\geq \dfrac{1}{(p_i^{n})^{-1} + (p_j^{n})^{-1}} \left(\norm{\mu_i^{n} - \mu_t^0}^{\overline{r}(|V_t|)} + \norm{\mu_j^{n} - \mu_t^0}^{\overline{r}(|V_t|)} \right.\\
& \hspace{3cm} \left.+ \norm{\Sigma_i^n - \Sigma_j^n}^{\overline{r}(|V_t|)/2} + \norm{\Sigma_j^n - \Sigma_j^n}^{\overline{r}(|V_t|)/2} \right)\\
& \gtrsim  \dfrac{1}{(p_i^{n})^{-1} + (p_j^{n})^{-1}} \left(\norm{\mu_i^{n} - \mu_j^n}^{\overline{r}(|V_t|)} + \norm{\Sigma_i^n - \Sigma_j^n}^{\overline{r}(|V_t|)/2} \right) \\
& \gtrsim \left(\dfrac{1}{(p_i^{n})^{-1} + (p_j^{n})^{-1}} \left(\norm{\mu_i^{n} - \mu_j^n}^{2} + \norm{\Sigma_i^n - \Sigma_j^n} \right)\right)^{\overline{r}(\widehat{G}_n)/2}.
\end{align*}
Because the height of the dendrogram is the minimum of dissimilarity $\divclus$ over all pairs $i, j$, we have that 
$$d_n^{(\kappa)} \lesssim \left(\dfrac{1}{(p_i^{n})^{-1} + (p_j^{n})^{-1}} \left(\norm{\mu_i^{n} - \mu_j^n}^{2} + \norm{\Sigma_i^n - \Sigma_j^n} \right)\right) \lesssim \left(\dfrac{\log n}{n}\right)^{1/\overline{r}(\widehat{G}_n)},$$
for all $\kappa\geq k_0 + 1$. The convergence of the height from the exact-fitted level $k_0$ proceeds similarly to the strong identifiability setting (Theorem~\ref{thm:asymptotic-height}), given that $W_1(\widehat{G}_n^{(k_0)}, G_0)\lesssim (\log n/ n)^{1/2}$.

The convergence of the average log-likelihood proceeds the same way as that of Theorem~\ref{thm:asymptotic-likelihood}. All the conditions required in the theorem are satisfied because we assumed $\Theta$ is compact and $\Omega$ is a compact subspace of positive definite matrices with bounded eigenvalues (both above and below by positive constants).
\end{proof}
\subsection{Proof of Proposition~\ref{prop:DIC-consistency-gaussian}}
\begin{proof}[Proof of Proposition~\ref{prop:DIC-consistency-gaussian}]
    This proof of this proposition proceeds in the same manner as that of Proposition~\ref{prop:DIC-consistency}, but replaces all the convergence rate results of the strongly identifiable setting with the weakly identifiable setting (Theorem~\ref{thm:asymp-theory-gaussian}). 
\end{proof}

\subsection{Discussion on the Wasserstein geometry of parameters in location-scale Gaussian mixtures}\label{subsec:discussion-dendrogram-locationscale}
It is desirable to obtain a nice geometric interpretation of the construction of the dendrogram as in Propostion~\ref{prop:Wasserstein-variational}. We now discuss the difficulty in achieving this mathematical perspective for the location-scale Gaussian mixtures. Recall that for the Gaussian kernel family $\{f(x | \mu, \Sigma) |  \mu\in \Theta, \Sigma\in \Omega\}$, we have the singular structure caused by the heat equation:
$$\dfrac{\partial^2}{\partial\mu \partial\mu^{\top}} f(x | \mu, \Sigma) = 2\dfrac{\partial}{\partial\Sigma} f(x | \mu, \Sigma) \quad \forall x\in \Rbb^{d},$$
so that the results for strongly identifiable mixtures do not hold. 
Similar to Proposition~\ref{prop:Wasserstein-variational}, we want to design a cost function $c: (\Theta, \Omega)\times (\Theta, \Omega)\to \Rbb$ and the corresponding Kantorovich formulation for optimal transport:
$$W_{c, 2}^{2} := \inf_{\pi \in \Pi(G, G')} \int_{(\Theta, \Omega)\times (\Theta, \Omega)} c^2((\mu, \Sigma), (\mu', \Sigma')) d\pi((\mu, \Sigma), (\mu', \Sigma')),$$
for all $G, G'\in \cup_k \Ocal_{k}(\Theta\times \Omega)$, satisfying:
\begin{enumerate}
    \item[(i)] $c$ is a distance metric on $\Theta\times \Omega$ (i.e., satisfy three axioms of a metric) so that the Kantorovich formulation $W_{c, 2}$ is a metric on $\cup_k \Ocal_{k}(\Theta\times \Omega)$. 
    \item[(ii)] For $G = G^{(k)} = \sum_{i=1}^{k} p_i \delta_{(\mu_i, \Sigma_i)}\in \Ecal_{k}(\Theta\times \Omega)$, the solution to the optimization problem:
    $$G^{(k-1)} = \argmin_{G\in \Ocal_{k-1}(\Theta\times \Omega)} W_{c, 2}^2(G, G^{(k)})$$
    is as described in Algorithm~\ref{alg:merge-atom-weak} so that the convergence rate for the merged mixing measure can be the fast root-$n$ rate (due to the inverse bound);
    \item[(iii)] The rate of convergence of the height $W_{c, 2}^2(G^{(k)}, G^{(k-1)})$ for overfitted mixing measure can easily derived from the inverse bound.
\end{enumerate}
Let us explicate the point (ii) further. Proceeding as in the proof of Proposition~\ref{prop:Wasserstein-variational}, we arrive at the minimization problem for two merged atoms $p_i \delta_{(\mu_i, \Sigma_i)}$ and $p_j \delta_{(\mu_j, \Sigma_j)}$:
\begin{equation}\label{eq:discuss-optimization-gaussian-merge}
    (\mu_*, \Sigma_*) = \argmin_{(\mu, \Sigma)\in \Theta\times \Sigma} p_i c^2\left((\mu_i, \Sigma_i), ({\mu}, {\Sigma})\right) + p_j c^2\left((\mu_j, \Sigma_j), ({\mu}, {\Sigma})\right),
\end{equation}
and we want this problem to have the following solution:
\begin{equation}\label{eq:discuss-optimization-gaussian-merge-mu}
    \mu_* = \dfrac{p_i}{p_i + p_j} \mu_i + \dfrac{p_j}{p_i + p_j} \mu_j, 
\end{equation}
and
\begin{equation}\label{eq:discuss-optimization-gaussian-merge-Sigma}
    {\Sigma}_* = \dfrac{p_i}{p_i + p_j} \left( \Sigma_i +  (\mu_i - \overline{\mu})(\mu_i - \overline{\mu})^{\top} \right) + \dfrac{p_j}{p_i + p_j} \left( \Sigma_j +  (\mu_j - \overline{\mu})(\mu_j - \overline{\mu})^{\top} \right).
\end{equation}
From here, there are two ways to construct $c$:
\begin{enumerate}
    \item[(1)] Because we are trying to merge two Gaussian distributions into one, this problem can naturally be cast as an MLE problem, and therefore $c$ can be defined as:
    \begin{align*}
        c^2((\mu', \Sigma'), (\mu, \Sigma)) & = \Ebb_{X\sim N(\mu, \Sigma)} -\log f(X | \mu', \Sigma') \\
        & = \dfrac{1}{2} \log(\det(\Sigma)) + \dfrac{1}{2} (\mu' - \mu)^{\top} \Sigma^{-1} (\mu'-\mu)
    \end{align*}
    Using this definition, we can derive the optimal $(\mu_*, \Sigma_*)$ of the minimization problem~\eqref{eq:discuss-optimization-gaussian-merge} to exactly be~\eqref{eq:discuss-optimization-gaussian-merge-mu} and~\eqref{eq:discuss-optimization-gaussian-merge-Sigma}. However, it is very difficult to make this $c$ to be a metric because it is not symmetric and does not satisfy the triangle inequality.
    
    \item[(2)] Choosing $c$ simple so that it satisfies (i) and (iii) and somewhat reflects (ii). A simple candidate is to put
    $$c^2((\mu', \Sigma'), (\mu, \Sigma)) = \norm{\mu' - \mu}^2 + \norm{\Sigma' - \Sigma}. $$
    Because of the heat equation~\eqref{eq:heat-eq}, it is reasonable to put a square in the distance between $\mu$'s but not between $\Sigma$'s. It is also relevant to the convergence rate of overfitted mixing measures: The convergence rate for $\Sigma$ is twice as fast as $\mu$ \cite{ho2019singularity}. Because of its simple form, both (i) and (iii) can be shown as in the previous section. For (ii), we can check that:
    $$W_{c, 2}^2(G^{(k-1)}, G^{(k)}) = \dfrac{1}{p_i^{-1} + p_j^{-1}}\left(\norm{\mu_i-\mu_j}^2 + 2 \norm{\Sigma_i -\Sigma_j} \right),$$
    so that the height is quite similar to the height obtained in the strongly identifiable setting, and it reflects the similarity between two merged atoms.
\end{enumerate}

\subsection{Discussion on merging covariance matrices}
A major difference in the merging scheme of location-scale Gaussian mixtures (Algorithm~\ref{alg:merge-atom-weak}) compared to strongly identifiable mixtures is the merging of covariance matrices. We recall that in equation~\eqref{eq:merge-variance}, the merged covariance is defined by
\begin{equation}
    \Sigma_* = \dfrac{p_{i}}{p_{*}} \left(\Sigma_{i} + (\mu_i - \mu_*)(\mu_i - \mu_*)^{\top} \right) + \dfrac{p_{j}}{p_{*}} \left(\Sigma_{j} + (\mu_j - \mu_*)(\mu_j - \mu_*)^{\top} \right),
\end{equation}
where the mean parameter $\mu$'s appears in the formulation. It is natural to ask if we naively merge $\Sigma$'s without $\mu$'s in this formulation (i.e., exactly the same as the strongly identifiable case), what would happen? Here we demonstrate the inefficiency of that strategy compared to Algorithm~\ref{alg:merge-atom-weak} by using the simulation described in Section~\ref{subsubsec:rate-dendrogram}. The true data distribution is a mixture of 3 location-scale Gaussians and the data is fitted by a mixture of 10 location-scale Gaussians. The merged results can be seen in Figure~\ref{fig:merge_demo_diff_cov}, which is quite off compared to the truth and is worse than the demonstration in Figure~\ref{fig:merge_demo}. An intuitive explanation is that the distribution of the mean parameters of over-fitted components also partially explains the direction of the covariance matrix. Indeed, when carefully inspecting the Gaussian component on the left of Figure~\ref{fig:merge_demo_diff_cov}(b), we can see that the mean parameters of 5 over-fitted components are aligned along the dominant direction of the covariance matrix of the true Gaussian. Hence, it is essential to take care of the "covariance of the means" as in equation~\eqref{eq:merge-variance} when working with this weakly identifiable mixture family. 
\begin{figure}[t!]
      \centering
      \subcaptionbox*{\scriptsize (a) Data with true contour plot \par}{\includegraphics[width = 0.32\textwidth]{figures/simulation-true-data-dist.pdf}}
      \subcaptionbox*{\scriptsize (b) Over-fitting $k=10$ \par}{\includegraphics[width = 0.32\textwidth]{figures/simulation-10.pdf}}
      \subcaptionbox*{\scriptsize (c) Merge $\kappa = 8$ \par}{\includegraphics[width = 0.32\textwidth]{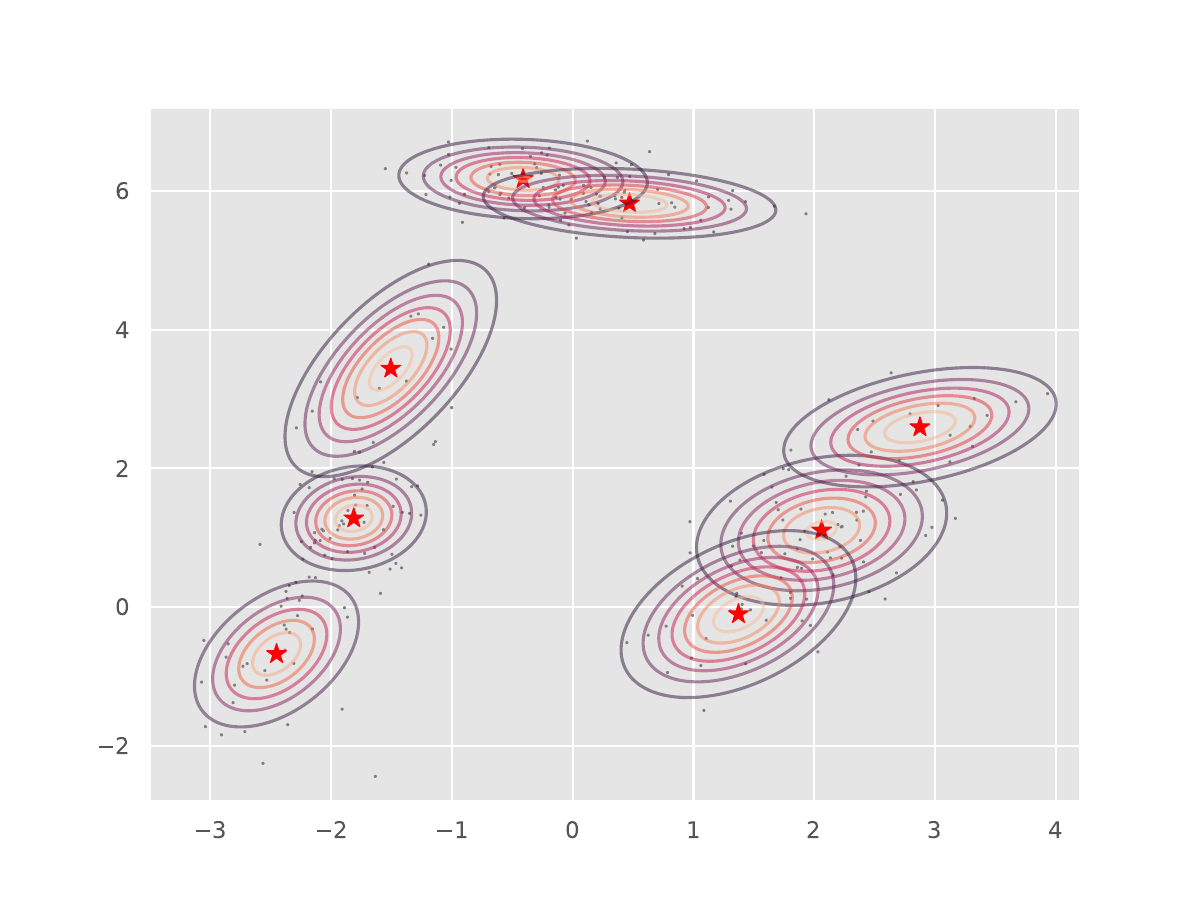}}
      \subcaptionbox*{\scriptsize (f) Merge $\kappa = 3$ \par}{\includegraphics[width = 0.32\textwidth]{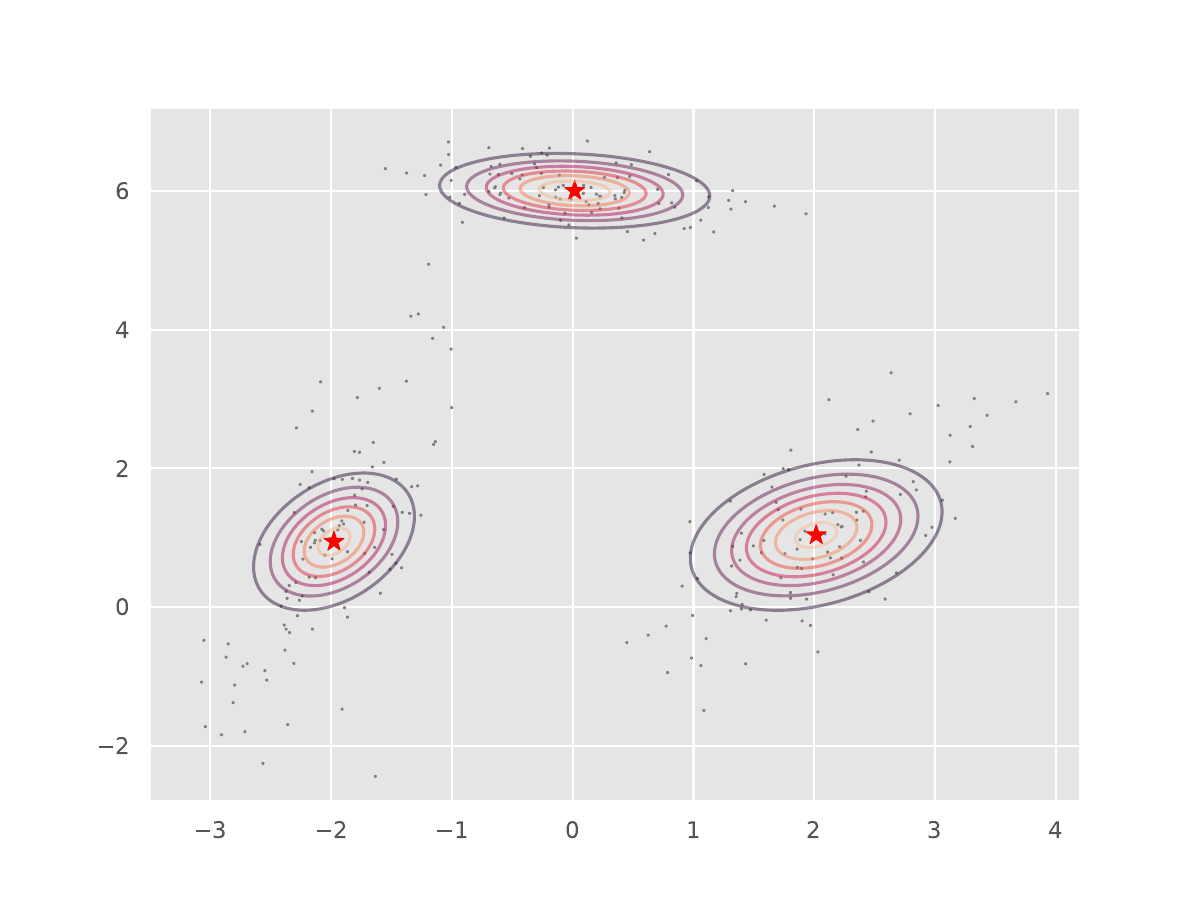}}
      \subcaptionbox*{\scriptsize (e) Merge $\kappa = 4$ \par}{\includegraphics[width = 0.32\textwidth]{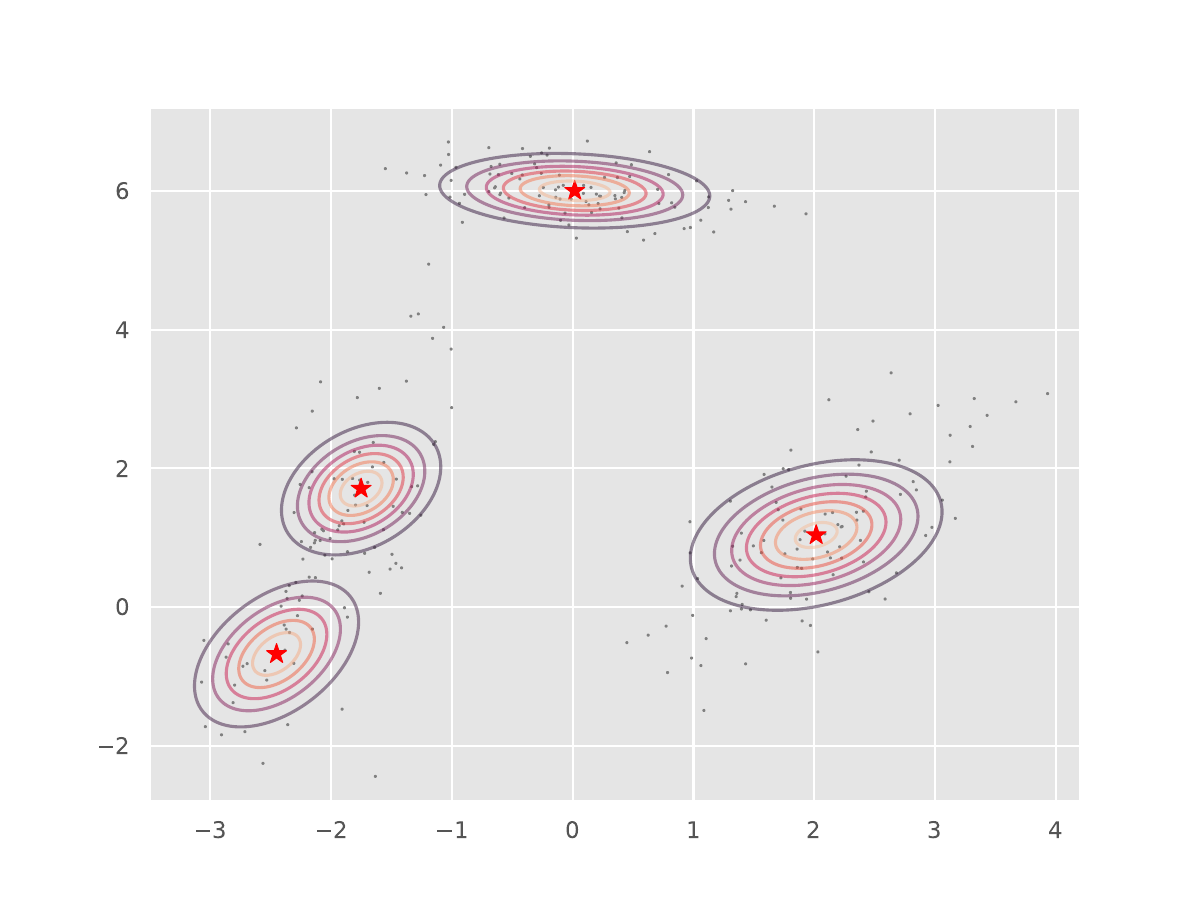}} 
      \subcaptionbox*{\scriptsize (d) Merge $\kappa = 6$ \par}{\includegraphics[width = 0.32\textwidth]{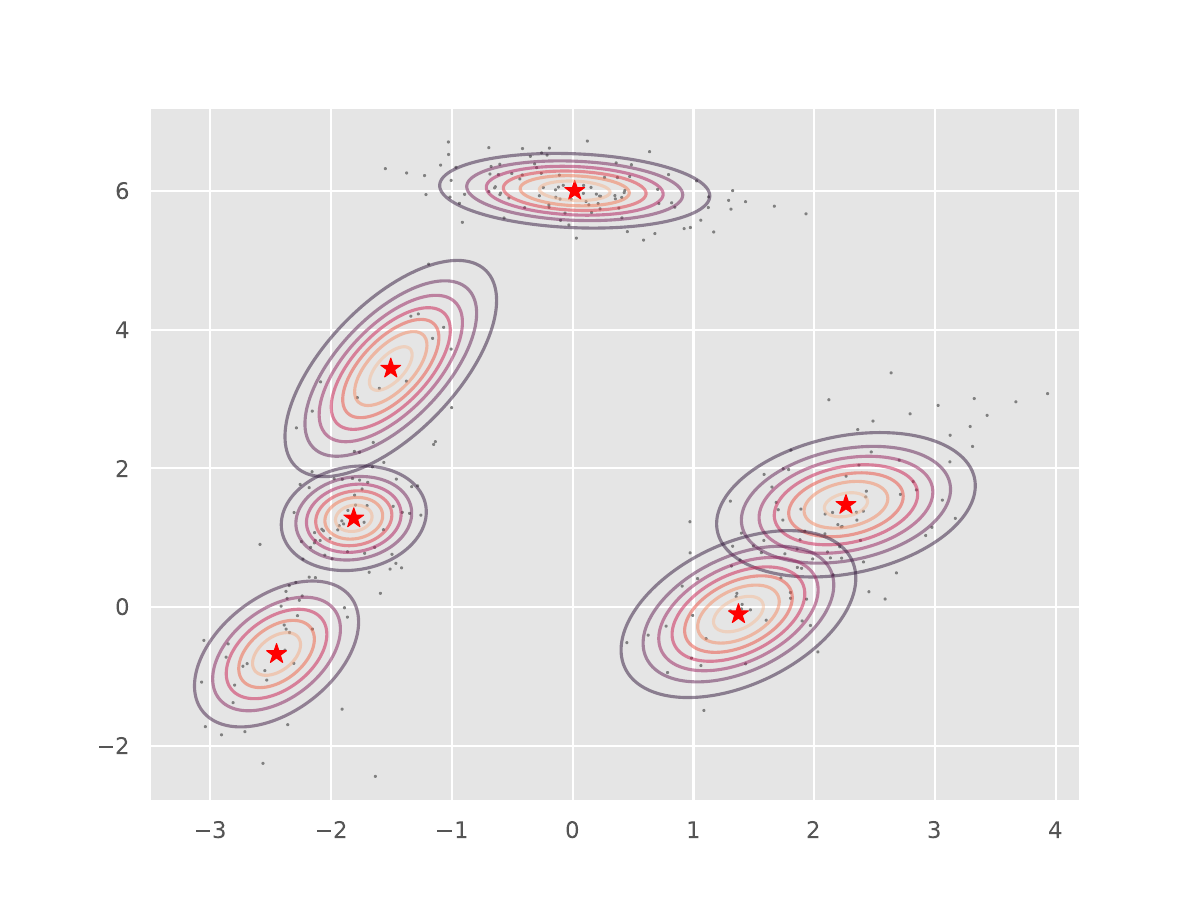}}  
      \caption{\centering Mixing measures on the dendrogram with $k=10$ and $k_0=3$ when naively merging covariances} \label{fig:merge_demo_diff_cov}
\end{figure}

\section{Dendrogram of mixing measures with single linkage}\label{sec:single-linkage}
In this section, we discuss the asymptotic behaviour of the dendrogram with a single linkage. To streamline the argument, we choose to describe it in the strong identifiability setting (as in Section~\ref{sec:dendrogram-strong}). Analogous to the single linkage of agglomeration hierarchical clustering, one can suggest the single linkage method of mixing measures as follows. Start with $G = G^{(k)} = \sum_{i=1}^{k} p_i \delta_{\theta_i} \in \Ecal_{k}$, we create a partition $\mathcal{I}^{(k)} = \{\{1\}, \dots, \{k\} \}$ and assign each atom with its corresponding index in this partition. We then find two atoms $p_i \delta_{\theta_i}$ and $p_j \delta_{\theta_j}$ that minimize the dissimilarity $\divclus$ and join them together to get the new atom $p_{ij} \delta_{\theta_{ij}}$ with 
$$p_{ij} := p_i + p_j, \quad \theta_{ij} := \dfrac{p_i}{p_{ij}} \theta_i + \dfrac{p_j}{p_{ij}}\theta_j.$$
We also merge the set $\{i\}$ and $\{j\}$ together to get a new partition $\mathcal{I}^{(k-1)}$ containing $(k-1)$ sets. In the $\kappa$-th step, we have partition $\mathcal{I}^{(\kappa)}$ of $[k]$ containing $k - \kappa$ sets $S_1, \dots, S_{k-\kappa}$ and corresponding to each set $S$ is the atom $p_{S} \delta_{\theta_{S}}$ being
\begin{equation}\label{eq:partition-to-atom}
p_{S} := \sum_{i\in S} p_{i}, \quad \theta_{S} := \dfrac{1}{p_{S}} \sum_{i\in S} p_i \theta_i.
\end{equation}

Choosing two atoms $p_i \delta_{\theta_i}$ and $p_j \delta_{\theta_j}$ belonging to two different sets $S_{\kappa_1}$ and $S_{\kappa_2}$ such that they minimize $\divclus$. Joining $S_{\kappa_1}$ and $S_{\kappa_2}$ while leaving other sets unchanged to get the new partition $\mathcal{I}^{(\kappa-1)}$. This new partition also corresponds to a mixing measure having $(\kappa-1)$ atoms specified by equation~\eqref{eq:partition-to-atom}. 

Recall that when we overfit an $n$ i.i.d. samples of a mixture of $k_0$ components by a mixture of $k$ components, where $k>k_0$, to get an estimator $G = \sum_{j=1}^{k} p_j \delta_{\theta_j}$, we have
$$\sum_{i=1}^{k_0} \sum_{j\in V_i} p_j \norm{\theta_j - \theta_i^0}^2 \lesssim\left(\dfrac{\log n}{n}\right)^{1/2}.$$
There may be the case where there exists an excess mass $p_*$ vanishing with the fast rate $\dfrac{1}{\sqrt{n}}$ (up to a logarithmic term), whereas ther merged atom $\theta_{*}$ can vary anywhere in the parameter space. Such a behaviour was found, e.g., in a Bayesian procedure~\cite{28Judith}. In this case, the dissimilarity of $p_* \delta_{\theta_*}$ to all other atoms will be of order $\dfrac{1}{\sqrt{n}}$ (up to a logarithmic term. Hence, it quickly merges into all other atoms and makes the height of the whole dendrogram very small. Merging using a single linkage criterion in this case might be not very useful and yield misleading inferences.
\begin{algorithm}[t]
\caption{Dendrogram of mixing measures with single linkage}\label{alg:single-linkage}
\begin{algorithmic}[1]
\Require A mixing measure $G^{(k)} = \sum_{i=1}^{k} p_i \delta_{\theta_i}$. 
\State Initiate $\Tcal(G) = (V, E, d)$, where the $k$-th level of $V$ contains all atoms of $G^{(k)}$, $E = \varnothing$, and $d = (d^{(k)}, d^{(k-1)}, \dots, d^{(2)})$ is an array of length $k-1$. $\mathcal{I}^{(k)} = \{\{1\}, \dots, \{k\} \}$
\For{$\kappa$ from $k$ to $2$} 
    \State Starting with the partition $\mathcal{I}^{(\kappa)}$ of $[k]$, find two atoms $p_i \delta_{\theta_i}$ and $p_j \delta_{\theta_j}$ belonging to two different sets $S_{\kappa_1}$ and $S_{\kappa_2}$ of $\mathcal{I}^{(\kappa)}$ such that they minimize $\divclus$.
    \State Merge $S_{\kappa_1}$ and $S_{\kappa_2}$, and keep other sets to build a new partition $\mathcal{I}^{(k-1)}$. The merge set $S = S_{\kappa_1} \cup S_{\kappa_2}$ corresponds to the atom
    $$p_{S} = \sum_{i\in S} p_{i}, \quad \theta_{S} = \dfrac{1}{p_{S}} \sum_{i\in S} p_i \theta_i.$$
    Hence, this new partition corresponds to a mixing measure $G^{(\kappa-1)}$ having $\kappa-1$ atoms.
    \State Add all atoms of $G^{(\kappa-1)}$ as vertices of the $(\kappa-1)$-th level of $V$;
    \State Add two edges to $E$ specifying which two atoms of $G^{(\kappa)}$ merge into atom of $G^{(\kappa-1)}$;
    \State Record $d^{(\kappa)} = \divclus(p \delta_{\theta}, \pi \delta_{\nu})$, where $p \delta_{\theta}$ and $\pi \delta_{\nu}$ are two merged atoms.
\EndFor
\State \textbf{return} $\Tcal(G) = (V, E, d)$.
\end{algorithmic}
\end{algorithm}

\end{document}